\PassOptionsToPackage{table}{xcolor}
\documentclass[sigplan,screen]{acmart}
\setcopyright{none}
\setcopyright{rightsretained}
\acmPrice{}
\acmDOI{10.1145/3453483.3454028}
\acmYear{2021}
\copyrightyear{2021}
\acmSubmissionID{pldi21main-p25-p}
\acmISBN{978-1-4503-8391-2/21/06}
\acmConference[PLDI '21]{Proceedings of the 42nd ACM SIGPLAN International Conference on Programming Language Design and Implementation}{June 20--25, 2021}{Virtual, UK}
\acmBooktitle{Proceedings of the 42nd ACM SIGPLAN International Conference on Programming Language Design and Implementation (PLDI '21), June 20--25, 2021, Virtual, UK}

\usepackage{booktabs}   %% For formal tables:
\usepackage{subcaption} %% For complex figures with subfigures/subcaptions
\usepackage{mathpartir}
\usepackage{xcolor}
\usepackage{listings}
\usepackage{array}
\usepackage{tabularx}
\usepackage{makecell}
\usepackage{booktabs}
\usepackage{cellspace}
\usepackage{multirow}
\usepackage{tcolorbox}
\usepackage{pifont}
\usepackage{mathtools}
\usepackage[plainruled,noline,linesnumbered,noend]{algorithm2e}
\usepackage{lipsum}
\usepackage[inline]{enumitem}
\usepackage{mdframed}
\usepackage{wrapfig}
\usepackage[T1]{fontenc}
\usepackage{wasysym}
\usepackage[export]{adjustbox}


\begin{document}
%%%%%%%%%%%%%%%%%%%%%%%%%%%%%%%%%%%%%%%%%%%%%%%%%%%%%%%%%%%%%%%%%%%
\newcommand{\tool}[0]{\textsc{Atropos}}        % name of the tool
\newcommand{\papername}[0]{
Repairing Serializability Bugs in Distributed Database Programs via Automated Schema Refactoring
}
  
%%%%%%%%%%%%%%%%%%%%%%%%%%%%%%%%%%%%%%%%%%%%%%%%%%%%%%%%%%%%%%%%%%%

%%%%%%%%%%%%%%%%%%%%%%%%%%%%%%%%%%%%%%%%%%%%%%%%%%%%%%%%%%%%%%%%%%%
\newcommand\BD[1]{\textcolor{red}{[\texttt{BD: #1}]}} % Ben's Comments
\newcommand\SJ[1]{\textcolor{red}{[\texttt{SJ: #1}]}} % Suresh's Comments
\newcommand\KR[1]{\textcolor{red}{[\texttt{KR: #1}]}} % Kia's Comments

\newcommand{\crr}[1]{\textsc{cr\Small{#1}}}  %consistency requirement
\newcommand{\passed}[0]{\color{tick}\ding{51}\color{black}}
\newcommand{\failed}[0]{\color{red}\ding{53}\color{black}}

%%%%%%%%%%%%%%%%%%%%%%%%%%%%%%%%%%%%%%%%%%%%%%%%%%%%%%%%%%%%%%%%%%%
%% NOTES
\newcommand\TODO[1]{\colorbox{yellow!20}{\textcolor{red}{[\texttt{#1}]}}}

%%%%%%%%%%%%%%%%%%%%%%%%%%%%%%%%%%%%%%%%%%%%%%%%%%%%%%%%%%%%%%%%%%%
%% TEXT STYLES
\newcommand{\stx}[1]{{#1}}   % For backward compatibility
\newcommand{\sql}[1]{\textcolor{sql}{\uppercase{\mathtt{#1}}\;}}
\newcommand{\param}[1]{{\color{param}#1}}

%% TEXT CONSTANTS
\newcommand{\intt}[0]{\mathbb{Z}}   % integers
\newcommand{\booll}[0]{\mathbb{B}}   % booleans
\newcommand{\ALT}{\,\mid\,}
\newcommand{\IN}[0]{\mathtt{IN}}

\newcommand{\orac}[0]{\mathcal{O}}

%%%%%%%%%%%%%%%%%%%%%%%%%%%%%%%%%%%%%%%%%%%%%%%%%%%%%%%%%%%%%%%%%%%
%% MATH STYLES
\newcommand{\seq}[1]{{#1}^*}
\newcommand{\set}[1]{\overline{#1}}
\newcommand{\func}[1]{\mathtt{#1}}

\newcommand{\vc}[0]{\mathbb{V}}
\newcommand{\rel}[0]{\mathbb{R}}
\newcommand{\txn}[0]{\mathbb{T}}
\newcommand{\hist}[0]{\mathcal{H}}
\newcommand{\map}[0]{\mathcal{M}}
\newcommand{\prog}[0]{\mathbb{P}}

%%%%%%%%%%%%%%%%%%%%%%%%%%%%%%%%%%%%%%%%%%%%%%%%%%%%%%%%%%%%%%%%%%%
%% OPERATIONAL SEMANTICS
\newcommand{\vis}[0]{\mathtt{vis}}
\newcommand{\dap}[0]{\chi}
\newcommand{\cnt}[0]{\mathtt{cnt}}
\newcommand{\store}[0]{\mathtt{str}}
\newcommand{\obs}[0]{\mathtt{obs}}
\newcommand{\st}[0]{\mathtt{st}}
\newcommand{\ar} [0]{\mathtt{ar}}
\newcommand{\so} [0]{\mathtt{so}}
\newcommand{\B}[2]{#2, #1 \Downarrow }
\newcommand{\I}[2]{#2, #1 \Downarrow }
\newcommand{\ir}[0]{$\mathcal{AR}$}
\newcommand{\rdd}[1]{\textcolor{eff}{\mathsf{rd}#1}}
\newcommand{\F}[0]{\mathcal{F}}
\newcommand{\wrr}[1]{\textcolor{eff}{\mathsf{wr}#1}}

\definecolor{bostonuniversityred}{rgb}{0.8, 0.0, 0.0}
\newcommand{\eval}[1]{\stx{[\![#1]\!]}}
\newcommand{\RULE}[2]{\footnotesize
\frac{\begin{array}{c}#1\end{array}}
     {\begin{array}{c}#2\end{array}}}
\newcommand{\ruleLabel}[1]{
\begin{flushleft}
    {\scriptsize
    \textrm{\sc{\color{bostonuniversityred} (#1)}}}
    \vspace{-1mm}
  \end{flushleft}}
\newcommand{\step}[3]{#1\xrightarrow{#2} #3}
\newcommand{\Rstep}[3]{#1\Rightarrow{#2} #3}
\newcommand{\hstepl}[3]{#1\xhookrightarrow{#2} #3}
\newcommand{\hstepc}[3]{#1\hookrightarrow^* #3}
\newcommand{\sstep}[2]{#1 \rightsquigarrow #2}
\newcommand{\sstepl}[3]{#1 \xrightsquigarrow{#2} #3}
\newcommand{\sstepc}[2]{#1 \rightsquigarrow^* #2}

\def\sectionautorefname{Section}
\def\subsectionautorefname{Section}
\def\figureautorefname{Figure}

\makeatletter
\newcommand{\removelatexerror}{\let\@latex@error\@gobble}
\makeatletter
\def\@algocf@pre@ruled{}%
\makeatother
\newcommand{\myalgorithm}{%
\begingroup
\removelatexerror% Nullify \@latex@error
\begin{algorithm*}[H]
%...
\end{algorithm*}
\endgroup}

%% MAIN COLORS
\definecolor{Main-0}{HTML}{DDDDF8}
\definecolor{Main-1}{HTML}{D5D5F6}
\definecolor{Main-2}{HTML}{CACAF4}
\definecolor{Main-3}{HTML}{BDBDF1}
\definecolor{Main-4}{HTML}{ACACED}
\definecolor{Main-5}{HTML}{9797E9}
\definecolor{Main-6}{rgb}{0.99,0.95,0.99}    
\definecolor{Main-7}{rgb}{0.95,0.99,0.95}    
\definecolor{Main-8}{HTML}{3333D3}
\definecolor{Main-9}{HTML}{0000C8}
\definecolor{Main-10}{HTML}{0000A0}
\definecolor{Main-11}{HTML}{000080}

%% COMPLEMENT COLORS
\definecolor{Comp-1}{HTML}{6D0041}
\definecolor{Comp-2}{HTML}{A33180}
\definecolor{Comp-3}{HTML}{C474C2}
\definecolor{Comp-4}{HTML}{DDB3E4}

%% TEXT COLORS
\definecolor{sql}{RGB}{0, 0, 160}
\definecolor{tick}{RGB}{0, 120, 0}
\definecolor{eff}{RGB}{10, 10, 170}
\definecolor{param}{RGB}{10, 10, 230}

%% GREYS
\definecolor{grey10}{rgb}{0.99,0.99,0.93}
\definecolor{grey11}{rgb}{0.93,0.93,0.93}
\definecolor{grey12}{rgb}{0.98,0.98,0.98}
\definecolor{grey9}{rgb}{0.9,0.9,0.9}
\definecolor{grey8}{rgb}{0.8,0.8,0.8}
\definecolor{grey7}{rgb}{0.7,0.7,0.7}
\definecolor{grey6}{rgb}{0.6,0.6,0.6}
\definecolor{grey5}{rgb}{0.5,0.5,0.5}
\definecolor{grey4}{rgb}{0.4,0.4,0.4}
\definecolor{grey3}{rgb}{0.3,0.3,0.3}
\definecolor{grey2}{rgb}{0.2,0.2,0.2}

%% Title information
\title[Distributed Database Programs Repair via Schema Refactoring] {\papername}
                                        %% [Short Title] is optional;
                                        %% when present, will be used in
                                        %% header instead of Full Title.
%\titlenote{Sisyphus: (n) Greek legend: a king in ancient Greece who offended Zeus and whose punishment was to roll a huge boulder to the top of a steep hill; each time the boulder neared the top it rolled back down and Sisyphus was forced to start again}            %% \titlenote is optional;
                                        %% can be repeated if necessary;
                                        %% contents suppressed with 'anonymous'
%\subtitle{Subtitle}                    %% \subtitle is optional
%\subtitlenote{with subtitle note}      %% \subtitlenote is optional;
                                        %% can be repeated if necessary;
                                        %% contents suppressed with 'anonymous'

%% Author information
%% Contents and number of authors suppressed with 'anonymous'.
%% Each author should be introduced by \author, followed by
%% \authornote (optional), \orcid (optional), \affiliation, and
%% \email.
%% An author may have multiple affiliations and/or emails; repeat the
%% appropriate command.
%% Many elements are not rendered, but should be provided for metadata
%% extraction tools.

% 1
%% Author with single affiliation.
\author{Kia Rahmani}
%% \authornote{with author1 note}          %% \authornote is optional;
\orcid{0000-0001-9064-0797}
\affiliation{
  %\position{Research Assistant}
  %\department{Department of Computer Science}              %% \department is recommended
  \institution{Purdue University}            %% \institution is required
 % \streetaddress{305 N. University Street}
  %\city{West Lafayette}
  %\state{Indiana}
 % \postcode{47906}
  %\country{USA}                    %% \country is recommended
}
\email{rahmank@purdue.edu}          %% \email is recommended

% 2
\author{Kartik Nagar}
\affiliation{
  %\position{Assistant Professor}
  %\department{Department of CSE}             %% \department is recommended
  \institution{IIT Madras}           %% \institution is required
 % \streetaddress{305 N. University Street}
 % \city{West Lafayette}
 % \state{Indiana}
 % \postcode{47906}
  %\country{India}                    %% \country is recommended
}
\email{nagark@cse.iitm.ac.in}         %% \email is recommended

% 3
\author{Benjamin Delaware}
\affiliation{
  %\position{Assistant Professor}
  %\department{Department of Computer Science}             %% \department is recommended
  \institution{Purdue University}           %% \institution is required
 % \streetaddress{305 N. University Street}
 % \city{West Lafayette}
 % \state{Indiana}
 % \postcode{47906}
  %\country{USA}                    %% \country is recommended
}
\email{bendy@purdue.edu}         %% \email is recommended

% 4
\author{Suresh Jagannathan}
\affiliation{
  %\position{Samuel D. Conte Professor}
  %\department{Department of Computer Science}             %% \department is recommended
  \institution{Purdue University}           %% \institution is required
 % \streetaddress{305 N. University Street}
 % \city{West Lafayette}
 % \state{Indiana}
 % \postcode{47906}
  % \country{USA}                    %% \country is recommended
}
\email{suresh@cs.purdue.edu}         %% \email is recommended

%% Abstract
%% Note: \begin{abstract}...\end{abstract} environment must come
%% before \maketitle command

%% 2012 ACM Computing Classification System (CSS) concepts
%% Generate at 'http://dl.acm.org/ccs/ccs.cfm'.
%\begin{CCSXML}
%<ccs2012>
%<concept>
%<concept_id>10011007.10011074.10011099.10011102.10011103</concept_id>
%<concept_desc>Software and its engineering~Software testing and debugging</concept_desc>
%<concept_significance>300</concept_significance>
%</concept>
%</ccs2012>
%\end{CCSXML}
%\ccsdesc[300]{Software and its engineering~Software testing and debugging}

%% Keywords
%% comma separated list
%\keywords{}  %% \keywords are mandatory in final camera-ready submission

%% \maketitle
%% Note: \maketitle command must come after title commands, author
%% commands, abstract environment, Computing Classification System
%% environment and commands, and keywords command.

\begin{abstract}

Serializability is a well-understood concurrency control mechanism
that eases reasoning about highly-concurrent database
programs. Unfortunately, enforcing serializability has a high
performance cost, especially on geographically distributed database
clusters.  Consequently, many databases allow programmers
to \emph{choose} when a transaction must be executed under
serializability, with the expectation that transactions would only be
so marked when necessary to avoid serious concurrency bugs.  However,
this is a significant burden to impose on developers, requiring them
to (a) reason about subtle concurrent interactions among potentially
interfering transactions, (b) determine when such interactions would
violate desired invariants, and (c) then identify the minimum number
of transactions whose executions should be serialized to prevent these
violations.  To mitigate this burden, this paper presents a sound
and fully automated schema refactoring procedure that transforms a
program's data layout -- rather than its concurrency control logic --
to eliminate statically identified concurrency bugs, allowing more
transactions to be safely executed under weaker and more performant
database guarantees.  Experimental results over a range of realistic
database benchmarks indicate that our approach is highly effective in
eliminating concurrency bugs, with safe refactored programs showing an
average of 120\% higher throughput and 45\% lower latency compared to
a serialized baseline.

\end{abstract}

\maketitle

%%%%%%%%%%%%%%%%%%%%%%%%%%%%%%%%%%%%%%%%%%%%%%%%%%%%%%
\section{Introduction}

Programs that concurrently access shared data are ubiquitous: bank
accounts, shopping carts, inventories, and social media applications
all rely on a shared database to store information. For performance
and fault tolerence reasons, the underlying databases that manage
state in these applications are often replicated and distributed
across multiple, geographically distant
locations~\cite{psi,WF+13,Cassandra,mongo}.  Writing programs which
interact with such databases is notoriously difficult, because the
programmer has to consider an exponential space of possible
interleavings of database operations in order to ensure that a client
program behaves correctly. One approach to simplifying this task is to
assume that sets of operations, or \emph{transactions}, executed by
the program are \emph{serializable}~\cite{Papadimitriou:1979:Ser},
i.e. that the state of the database is always consistent with some
sequential ordering of those transactions. One way to achieve this is
to rely on the underlying database system to seamlessly enforce this
property.  Unfortunately, such a strategy typically comes at a
considerable performance cost.  This cost is particularly significant
for distributed databases, where the system must rely on expensive
coordination mechanisms between different replicas, in effect limiting
\emph{when} a transaction can see the effects of another in a way that
is consistent with a serializable execution~\cite{BA13}. This cost is
so high that developers default to weaker consistency guarantees,
using careful design and testing to ensure correctness, only relying
on the underlying system to enforce serializable transactions when
serious bugs are discovered~\cite{LBL04,pldi15,GO16,Rahmani:Syncope:2018}.

Uncovering such bugs is a delicate and highly error-prone task even in
centralized environments: in one recent study, \citet{WB17} examined
12 popular eCommerce applications used by over two million well-known
websites and discovered 22 security vulnerabilities and invariant
violations that were directly attributable to non-serializable
transactions.  To help developers identify such bugs, the community
has developed multiple program analyses that report potential
\emph{serializability anomalies}~\cite{GO16, KR18, KJ18,
  BDM+17, BR18}.  Automatically repairing these anomalies, however, has
remained a challenging open problem: in many cases full application
safety is only achievable by relying on the system to enforce strong
consistency of all operations.  Such an approach results in developers
either having to sacrifice performance for the sake of correctness, or
conceding to operate within a potentially restricted ecosystem with
specialized services and APIs~\cite{BB+17} architected with
strong consistency in mind.

% Q: What is the problem we are trying to solve:
% A: Automatically repair serialability
% Q: Why is it important?
% A: Database programs which concurrently access a shared database are
% ubiqutus: banks, shopping carts, etc. Why use a distributed database
% program? Performance reasons, fault tolerance.

% It is hard to reason about the possible
% interleavings of accesses is notoriously difficult, but this burden
% can be simplified with serializability (define). This allows users
% to... The underlying system
% can guarantee this property, but at a peformance cost. This cost is
% particularly significant for distributed databases, where it requires
% expensive synchronization between the different copies, in effect
% limiting \emph{when} various operations can see the effects of another
% in a way that is consistent with a serializable execution.

% This cost is so high that developers typically use weaker consistency
% guarantees, relying on testing and smarts to ensure they are correct,
% using serializable transactions only as a last resort. Identifying
% errors is is a delicate an error-prone endeavor: Balis. To help users
% address this problem, many tools have been developed to identify
% serializability anomalies.

% However, effectively repairing these anomalies is hard. Default is to
% selectively turn on serializablle transactions. Unfortunately... this
% forces developers to chose between correctness and performance when
% elminating concurrency bugs in clients of distrbuted databases.

In this paper, we propose a novel language-centric approach to
resolving concurrency bugs that arise in these distributed
environments.  Our solution is to alter the \emph{schema}, or data
layout, of the data maintained by the database, rather than the
consistency levels of the transactions that access that data.  Our key
insight is that it is possible to modify shared state to remove
opportunities for transactions to witness changes that are
inconsistent with serializable executions.  We, therefore, investigate
automated schema transformations that change \emph{how} client
programs access data to ensure the absence of concurrency bugs, in
contrast to using expensive coordination mechanisms to limit
\emph{when} transactions can concurrently access the database.

For example, to prevent transactions from observing non-atomic updates
to different rows in different tables, we can fuse the offending
fields into a single row in a single table whose updates are
guaranteed to be atomic under any consistency guarantee.  Similarly,
consecutive reads and writes on a row can be refactored into
``functional'' inserts into a new table, which removes the race
condition between concurrently running instances of the program.
%to prevent transactions
%from prematurely witnessing updates to a table, we can refactor the
%schema and client programs to replace updates with ``functional''
%inserts and queries with aggregation over these inserts suffices to
%prevent the violation, all without imposing undue restrictions on how
%the transaction observes other non-problematic state changes.
By changing the schema (and altering how transactions access data
accordingly), without altering a transaction's atomicity and isolation
levels, we can make clients of distributed databases \emph{safer}
without \emph{sacrificing performance}. In our experimental
evaluation, we were able to fix on average 74\% of all identified
serializability anomalies with only a minimal impact (less than 3\% on average)
on performance in an environment that provides only weak eventually
consistent guarantees~\cite{Bu14}.  For the remaining 26\% of anomalies that were
not eliminated by our refactoring approach, simply marking the
offending transactions as serializable yields a provably safe program
that nonetheless improves the throughput (resp. latency) of its fully serialized
counterpart by 120\% (resp. 45\%) on average.

%% For those bugs that weIn fact, it is possible tox
%% eliminate \emph{all} anomalies by using database-enforced serializable
%% transactions on the repaired program with 120\% performance gain over
%% the original program.

This paper makes the following contributions:
\begin{enumerate}
\item We observe that serializability violations in database programs
  can be eliminated by changing the schema of the underlying database
  and the client programs in order to eliminate problematic accesses to
  shared database state.

\item Using this observation, we develop an automated refactoring
  algorithm that iteratively repairs statically identified
  serializability anomalies in distributed database clients.  We
  show this algorithm both preserves the semantics of the original
  program and eliminates many identified serializability anomalies.

% \item We formalize our techniques and present the conditions under
%   which a refactoring is sound and does not introduce new
%   anomalies. We then prove our approach is sound.

\item We develop a tool, (\tool), implementing these ideas, and
  demonstrate its ability to reduce the number of serializability
  anomalies in a corpus of standard benchmarks with minimal
  performance impact over the original program, but with substantially
  stronger safety guarantees.

\end{enumerate}

\noindent The remainder of the paper is structured as follows. The next section
presents an overview of our approach. ~\autoref{sec:database_programs}
defines our programming model and formalizes the notion of concurrency
bugs.  ~\autoref{sec:refactor} provides a formal treatment of our
schema refactoring strategy. Sections \ref{sec:fix} and \ref{sec:impl}
describe our repair algorithm and its implementation, respectively.
\autoref{sec:eval} describes our experimental evaluation.  Related
work and conclusions are given in~\autoref{sec:related}
and~\autoref{sec:conc}.

\section{Overview}
\label{sec:over}

    To illustrate our approach, consider an online course management
    program that uses a database to manage a list of
    course offerings and registered
    students.  \autoref{fig:over_example} presents a simplified code
    snippet implementing such a program. The database consists of
    three tables, maintaining information regarding courses, students,
    and their email addresses. The $\func{STUDENT}$ table maintains a
    reference to a student's email entry in schema $\func{EMAIL}$ (via
    secondary key $\mathtt{st\_em\_id}$) and a reference to a course
    entry in table $\func{COURSE}$ (via secondary key
    $\mathtt{st\_co\_id}$) that the student has registered for.  A
    student's registration status is stored in field
    $\func{st\_reg}$. Each entry in table $\func{COURSE}$ also stores
    information about the availability of a course and the number of
    enrolled students.

    The program includes three sets of database operations or
    \emph{transactions}.  Transaction $\func{getSt}$, given a
    student's id, first retrieves all information for that student
    (\texttt{S1}). It then performs two queries, (\texttt{S2} and
    \texttt{S3}), on the other tables to retrieve their email address
    and course availability.  Transaction $\func{setSt}$ takes a student's
    id and updates their name and email address. It includes a query
    (\texttt{S4}) and an update (\texttt{U1}) to table
    $\func{STUDENT}$ and an update to the $\func{EMAIL}$ table
    (\texttt{U2}). Finally, transaction $\func{regSt}$ registers a
    student in a course. It consists of an update to the student's
    entry (\texttt{U3}), a query to $\func{COURSE}$ to determine the
    number of students enrolled in the course they wish to register for
    (\texttt{S5}), and an update to that course's availability
    (\texttt{U4}) indicating that it is available now that a
    student has registered for it.

    The desired semantics of this program is these transactions should
    be performed \emph{atomically} and in \emph{isolation}.
    Atomicity guarantees that a transaction never observes
    intermediate updates of another
    transaction.  Isolation guarantees that a transaction never
    observes changes to the database
    by other committed transactions
    once it begins executing.  Taken together, these properties ensure
    that all executions of this program are \emph{serializable},
    yielding behavior that corresponds to some sequential interleaving
    of these transaction instances.

%%%%%%%%%%%%%%%%%%%%%%%%%%%
\begin{figure}[t]

\definecolor{pgrey}{rgb}{0.26,0.25,0.28}
\definecolor{javared}{rgb}{0.2,0.2,0.7} % for strings
\definecolor{javagreen}{rgb}{0.2,0.45,0.3} % comments
\definecolor{javapurple}{rgb}{0.5,0,0.35} % keywords
\definecolor{javadocblue}{rgb}{0.25,0.35,0.75} % javadoc
\definecolor{weborange}{RGB}{0,75,0}

\lstset{language=Java,
  basicstyle=\ttfamily\scriptsize,
  breaklines=true,
  backgroundcolor=\color{Main-6},
  frame=single,
  rulecolor=\color{pgrey},
  keywordstyle=\color{black},
  stringstyle=\color{javared},
  commentstyle=\color{javagreen},
  morecomment=[s][\color{javadocblue}]{/**}{*/},
  numbers=left,
  %title=\footnotesize{getStudent(id)},
  xleftmargin=0.2em,
  framexleftmargin=1.25em,
  numberstyle=\tiny\color{pgrey},
  stepnumber=1,
  numbersep=5pt,
  tabsize=1,
  captionpos=b
  showspaces=false,
  showstringspaces=false,
  classoffset=2, % starting new class  
  morekeywords={update,select,set,from,where},
  keywordstyle=\color{sql},
  moredelim=[is][\textcolor{red}]{\%\%}{\%\%},
}
\vspace{-2mm}
\begin{flushleft}
\includegraphics[width=0.48\textwidth]{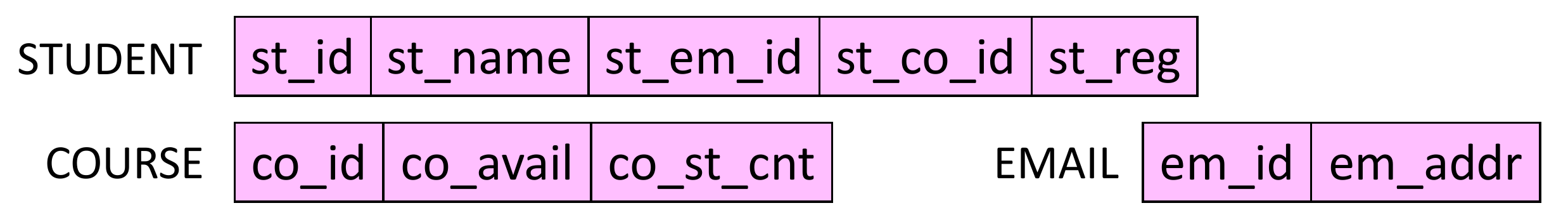}
\end{flushleft}
%%-----------------
% code figure
%
\begin{minipage}[b]{0.485\textwidth}
\begin{lstlisting}[]
getSt(id):
 x:=select * from STUDENT where st_id=id//S1
 y:=select em_addr from EMAIL where em_id=x.st_em_id//S2
 z:=select co_avail from COURSE where co_id=x.st_co_id//S3
\end{lstlisting}
\end{minipage}
\begin{minipage}[b]{0.485\textwidth}
\begin{lstlisting}[]
setSt(id,name,email):
 x := select st_em_id from STUDENT where st_id=id//S4
 update STUDENT set st_name=name where st_id=id//U1
 update EMAIL set em_addr=email where em_id=x.st_em_id//U2         
\end{lstlisting}
\end{minipage}
\begin{minipage}[b]{0.485\textwidth}
\begin{lstlisting}[]
regSt(id,course):
 update STUDENT set st_co_id=course, st_reg=true 
    where st_id=id//U3
 x:=select co_st_cnt from COURSE where co_id=course//S5
 update COURSE set co_st_cnt=x.co_st_cnt+1,
    co_avail=true where co_id=course  //U4
\end{lstlisting}
\end{minipage}

\caption{Database schemas and code snippets from an online course management program}
\label{fig:over_example}
\end{figure}

%%%%%%%%%%%%%%%%%%%%%%%%%%%

    While serializability is highly desirable, it requires using
    costly centralized locks~\cite{ullmanbook} or complex version
    management systems~\cite{BG81}, which severely reduce the system's
    available concurrency, especially in distributed environments
    where database state may be replicated or partitioned to improve
    availability.  In such environments, enforcing serializability
    typically either requires coordination among all replicas whenever
    shared data is accessed or updated, or ensuring replicas always
    witness the same consistent order of operations~\cite{spanner}.
    As a result, in most modern database systems, transactions can be
    executed under weaker isolation levels, e.g. permitting them to
    observe updates of other committed transactions during their
    execution~\cite{mysqliso,postgresiso,mongo,Cassandra}.
    Unfortunately, these weaker guarantees can result in
    \emph{serializability anomalies}, or behaviors that would not
    occur in a serial execution.  To illustrate, \autoref{fig:anml}
    presents three concurrent executions of this program's transaction
    instances that exhibit non-serializable behaviors.  %We note that
    %these executions are permitted in popular widely-used databases
    %such as MySQL~\cite{mysqliso} or Postgres~\cite{postgresiso}.

    The execution on the left shows instances of the $\func{getSt}$
    and $\func{setSet}$ transactions. Following the order in which
    operations execute (denoted by red arrows), observe that (S2)
    witnesses the update to a student's email address, but (S1) does
    not see their updated name. This anomaly is known as a
    \emph{non-repeatable read}. The execution in the center depicts
    the concurrent execution of instances of $\func{getSt}$ and
    $\func{regSt}$.  Here, (S1) witnesses the effect of (U3) observing
    that the student is registered, but (S3) sees that the course is
    unavailable, since it does not witness the effect of (U4). This is
    an instance of a \emph{dirty-read} anomaly.  Lastly, the execution
    on the right shows two instances of $\func{regSt}$ that attempt to
    increment the number of students in a course. This undesirable
    behavior, known as a \emph{lost update}, leaves the database in a
    state inconsistent with any sequential execution of the two
    transaction instances.  All of these anomalies can arise if
    the strong atomicity and isolation guarantees afforded by
    serializability are weakened.
    %%%%%%%%%%%%%%%%%%%%%%%%%%%
%%-----------------
% settings

\begin {figure}[h]

\includegraphics[scale=0.295]{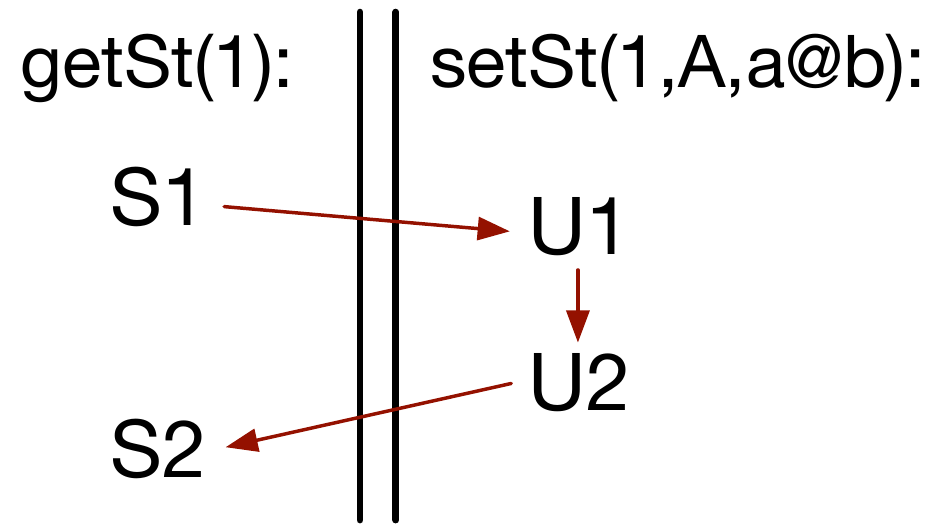}
~
\hfill
~
\includegraphics[scale=0.295]{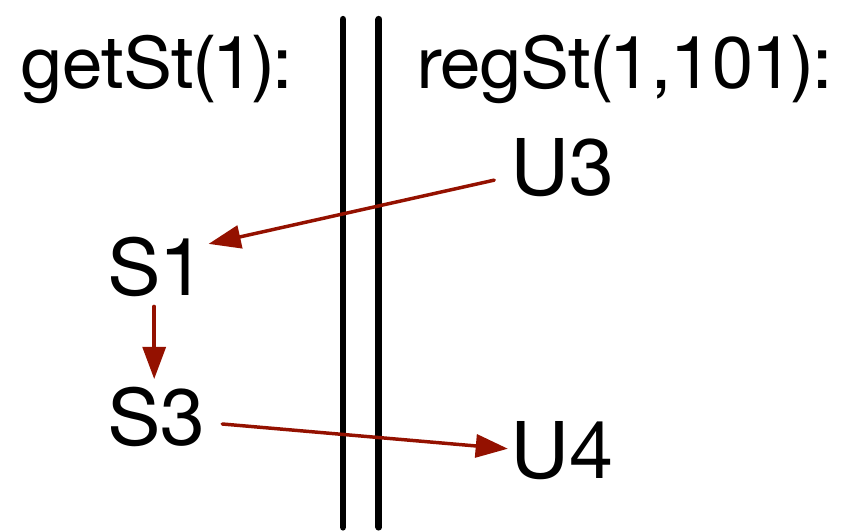}
~
\hfill
~
\includegraphics[scale=0.295]{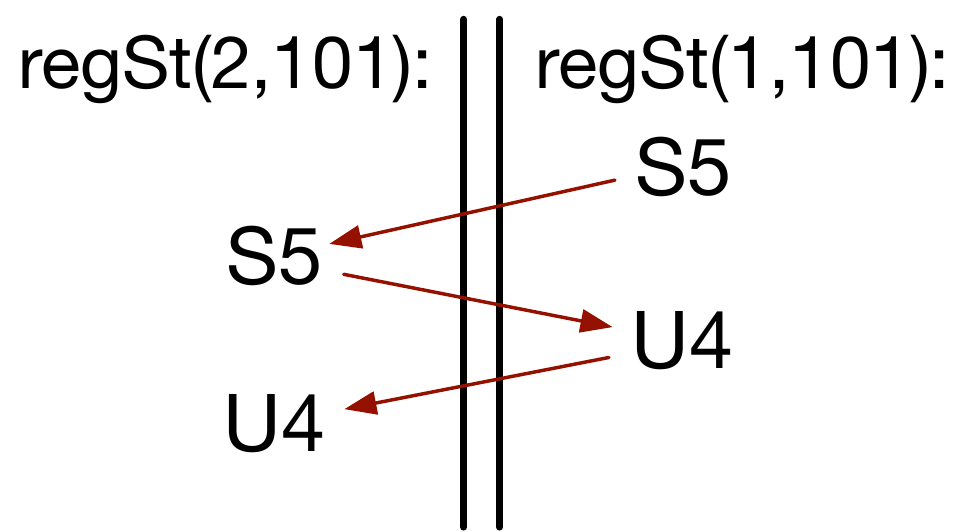}
\caption{Serializability Anomalies}
\label{fig:anml}
\end{figure}

%%%%%%%%%%%%%%%%%%%%%%%%%%%

Several recent proposals attempt to identify such undesirable
behaviors in programs using a variety of static or dynamic program
analysis and monitoring techniques~\cite{BDM+17,BR18,KR18,WB17}. Given
potential serializability violations, the standard solution is to
strengthen the atomicity and isolation requirements on the offending
transactions to eliminate the undesirable behaviour, at the cost of
increased synchronization overhead or reduced
availability~\cite{Bailis:2014:RAMP,pldi15,GO16}.

 %   (e.g., by
 %   demanding that the underlying implementation execute transactions
 %   at a certain isolation level); such strengthening restricts
 %   allowable behaviors at the cost of potentially increased
 %   synchronization overhead or reduced
 %   availability~\cite{Bailis:2014:RAMP,pldi15,GO16}, unpalatable
 %   tradeoffs that require developers to make difficult choices on how
 %   to best balance the need for correctness with the desire for
 %   performance and scalability.

%%%%%%%%%%%%%%%%%%%%%%%%%%%%%%%%%%%%%%%%%%%%%%%%%%%%%%
\subsection{\tool}
%%%%%%%%%%%%%%%%%%%%%%%%%%%%%%%%%%%%%%%%%%%%%%%%%%%%%%

%This paper asks if this tradeoff between correctness and
%performance is always necessary to remediate serializability errors.
%In other words,
Are developers obligated to sacrifice concurrency and availability in
order to recover the pleasant safety properties afforded by
serializability?  Surprisingly, we are able to answer this question in
the negative.  To see why, observe that a database program consists of
two main components - a set of computations that includes
transactions, SQL operations (e.g., selects and updates), locks,
isolation-level annotations, etc.; and a memory abstraction expressed
as a relational schema that defines the layout of tables and the
relationship between them.  The traditional candidates picked for
repairing a serializability anomaly are the transactions from the
computational component: by injecting additional concurrency control
through the use of locks or isolation-strengthening annotations,
developers can control the degree of concurrency permitted,
albeit at the expense of performance and availability.

This paper investigates the under-explored alternative of transforming
the program's schema
% (rather than the program's control-flow)
to reduce the number of potentially conflicting accesses to shared
state.  For example, by \emph{aggregating} information found in
multiple tables into a single row on a single table, we can exploit
built-in \emph{row-level atomicity} properties to eliminate
concurrency bugs that arise because of multiple non-atomic accesses to
different table state.  Row-level atomicity, a feature supported in
most database systems, guarantees that other concurrently executing
transactions never observe partial updates to a particular row.
Alternatively, it is possible to \emph{decompose} database state to
minimize the number of distinct updates to a field, for
example by logging state changes via table \emph{inserts}, rather than
recording such changes via \emph{updates}.  The former effectively acts as a
functional update to a table.  To be sure, these transformations
affect read and write performance to database tables and change the memory
footprint, but they notably impose no additional synchronization
costs.  In scalable settings such as replicated distributed
environments, this is a highly favorable tradeoff since the cost of
global concurrency control or coordination is often problematic in
these settings, an observation that is borne our in our experimental results.

      %% Instead of analyzing programs \emph{as they are} and then assigning
      %% them stronger database semantics, in this paper we propose an unorthodox
      %% approach for elimination of serializability anomalies. We suggest that
      %% programs should first be \emph{refactored} into
      %% semantically equivalent versions, which are however better suited for the
      %% data access semantics of modern database systems. We show that through
      %% this refactoring stage, we can eliminate many
      %% of the undesired concurrency bugs from the program, without paying the cost of
      %% any stronger database level guarantee.

%%%%%%%%%%%%%%%%%%%%%%%%%%%
\begin{figure}[t]

\definecolor{pgrey}{rgb}{0.26,0.25,0.28}
\definecolor{javared}{rgb}{0.2,0.2,0.7} % for strings
\definecolor{javagreen}{rgb}{0.2,0.45,0.3} % comments
\definecolor{javapurple}{rgb}{0.5,0,0.35} % keywords
\definecolor{javadocblue}{rgb}{0.25,0.35,0.75} % javadoc
\definecolor{weborange}{RGB}{0,75,0}

\lstset{language=Java,
  basicstyle=\ttfamily\scriptsize,
  breaklines=true,
  backgroundcolor=\color{Main-7},
  frame=single,
  rulecolor=\color{pgrey},
  keywordstyle=\color{black},
  stringstyle=\color{javared},
  commentstyle=\color{javagreen},
  morecomment=[s][\color{javadocblue}]{/**}{*/},
  numbers=left,
  %title=\footnotesize{getStudent(id)},
  xleftmargin=1.9em,
  framexleftmargin=1.25em,
  numberstyle=\tiny\color{pgrey},
  stepnumber=1,
  numbersep=5pt,
  tabsize=1,
  captionpos=b
  showspaces=false,
  showstringspaces=false,
  classoffset=2, % starting new class  
  morekeywords={into,insert,update,select,set,from,where,values},
  keywordstyle=\color{sql},
  moredelim=[is][\textcolor{red}]{\%\%}{\%\%},
}

%%-----------------
% code figure
%
\begin{flushleft}
 \includegraphics[width=0.47\textwidth]{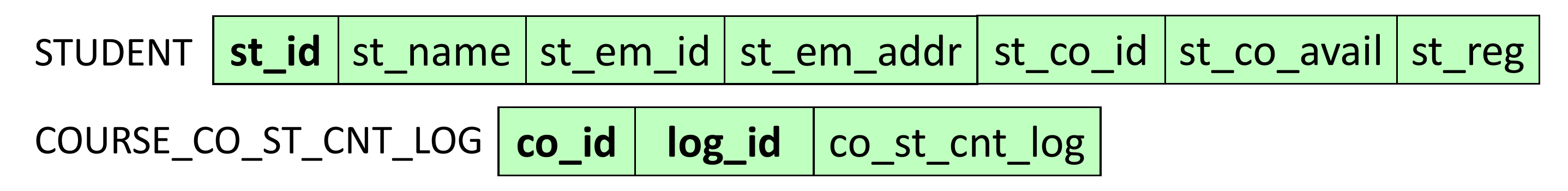}
 \end{flushleft}
\begin{minipage}[b]{0.475\textwidth}
\begin{lstlisting}[]
getSt(id):
 x:=select * from STUDENT where st_id=id //RS1,RS2,RS3
\end{lstlisting}
\end{minipage}
\begin{minipage}[b]{0.475\textwidth}
\begin{lstlisting}[]
setSt(id,name,email):
 update STUDENT set st_name=name,st_em_addr=email 
    where st_id=id //RU1,RU2
\end{lstlisting}
\end{minipage}
\begin{minipage}[b]{0.475\textwidth}
\begin{lstlisting}[]
regSt(id,course):
 update STUDENT set st_co_id=course, st_co_avail=true, 
    st_reg=true where st_id=id //RU3
 insert into COURSE_CO_ST_CNT_LOG values 
    (co_id=course,log_id=uuid(),co_st_cnt_log=1) //RU4
\end{lstlisting}
\end{minipage}

\caption{Refactored transactions and database schemas}
\label{fig:over_example_refactored}
\end{figure}

%%%%%%%%%%%%%%%%%%%%%%%%%%%
%%%%%%%%%%%%%%%%%%%%%%%%%%%
\begin {figure}[t]
\vspace{-2mm}
    \includegraphics[width=0.49\textwidth]{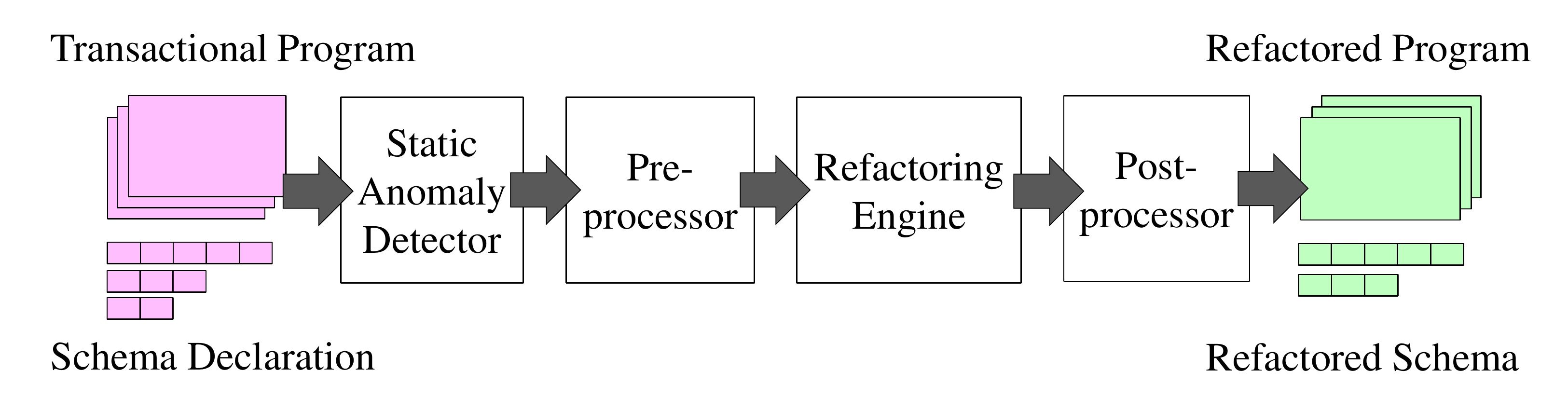}
  \caption{Schematic overview of \tool}
\label{fig:overview}
\end{figure}

%%%%%%%%%%%%%%%%%%%%%%%%%%%
To illustrate the intuition behind our approach, consider the database
program depicted in \autoref{fig:over_example_refactored}. This
program behaves like our previous example, despite featuring very
different database schemas and transactions. The first of the two
tables maintained by the program, $\func{STUDENT}$, removes the
references to other tables from the original $\func{STUDENT}$ table,
instead maintaining independent fields for the student's email address
and their course availability.  These changes make the original course
and email tables obsolete, so they have been removed.  In addition,
the number of students in each course is now stored in a dedicated
table $\func{COURSE\_CO\_ST\_CNT\_LOG}$.  Each time the enrollment of
a course changes, a new record is inserted into this table to record
the change. Subsequent queries can retrieve all corresponding records
in the table and aggregate them in the program itself to determine the
number of students in a course.

The transactions in the refactored program are also modified to
reflect the changes in the data model.  The transaction $\func{getSt}$
now simply selects a single record from the student table to retrieve
all the requested information for a student.  The transaction
$\func{setSt}$ similarly updates a single record. Note that both these
operations are executed atomically, thus eliminating the problematic
data accesses in the original program.  Similarly, $\func{regSt}$
updates the student's $\func{st\_co\_id}$ field and inserts a new
record into the schema $\func{COURSE\_CO\_ST\_CNT\_LOG}$. Using the
function $\func{uuid()}$ ensures that a new record is inserted every
time the transaction is called.  These updates remove potential
serializability anomalies by replacing the disjoint updates to fields
in different tables from the original with a simple atomic row
insertion.  Notably, the refactored program can be shown to be a
meaningful \emph{refinement} of the original program, despite eliminating problematic serializability errors found in it. Program refinement ensures that 
the refactored program maintains \emph{all} information maintained by the original program without exhibiting \emph{any} new behaviour.

The program shown in \autoref{fig:over_example_refactored} is the result of several database schema refactorings
\cite{Ambler:2006:Refactoring,Faroult:2008:Refactoring,
  Fowler:2019:Refactoring}, incremental changes to a database
program's data model along with corresponding semantic-preserving
modifications to its logic. Manually searching for such a refactored
program is unlikely to be successful. On one hand, the set of
potential solutions is large \cite{Ambler:2006:Refactoring}, rendering
any manual exploration infeasible. On the other hand, the process of
rewriting an application for a (even incrementally) refactored schema
is extremely tedious and error-prone~\cite{Wang:2019:Synthesizing}.

We have implemented a tool named \tool{} that, given a database
program, explores the space of its possible schema and program
refactorings, and returns a new version with possibly many fewer
concurrency bugs. The refactored program described above, for example,
is automatically generated by \tool\ from the original shown in
\autoref{fig:over_example}. \autoref{fig:overview} presents the \tool\
pipeline.  A static analysis engine is used to identify potential
serializability anomalies in a given program.  The program is then
preprocessed to extract the components which are involved in at least
one anomaly, in order to put it into a form amenable for our analysis.
Next, a refactoring engine applies a variety of transformations in an
attempt to eliminate the bugs identified by our static analysis.
Finally, the program is analyzed to eliminate dead code, and the
refactored version is then reintegrated into the program from which it
was extracted.

   % for example, allow
%    transactions to operate under weaker isolation levels~\cite{AD00},
%    permitting them to observe updates of other committed transactions
%    during their execution.
    %These systems do so to improve performance
    %and scalability by minimizing the use of centralized locks that
    %can reduce concurrency; systems that use
    %optimistic concurrency control methods that do not rely on locking
    %must deal with other equally-challenging complexities such as
    %timestamp and version management~\cite{BG81}.
   % In distributed
%    environments where database state may be replicated to ensure
%    availability in the face of network partitions and long latencies,
%    guaranteeing that all updates performed by a transaction on one
%    node are witnessed by transactions executing on other nodes is
%    complicated, and often entails significant
%    overhead~\cite{BA13,Bu14}.

    %implementations in these settings may
    %allow transactions to observe partially committed state to reduce
    %these complexities.

%%%%%%%%%%%%%%%%%%%%%%%%%%%%%%
\section{Database Programs}
\label{sec:database_programs}
%%%%%%%%%%%%%%%%%%%%%%%%%%%%%%

%
\begin{figure}[t]
  \centering
  \begin{minipage}{0.485\textwidth}
    \begin{mdframed}[backgroundcolor=grey12,linecolor=grey11]
%%%%%%%%%%%%%%%%%%%%%%%%%%%%%%%%%%%%%%%%
\hspace{-3mm}
\begin{minipage}[t]{0.44\textwidth}
    \begin{mathpar}
  \footnotesize
  \begin{array}{lcl}
    f & \in & \mathtt{FldName} \\
    \rho & \in &  \mathtt{SchmName} \\
    t & \in & \mathtt{TxnName} \\
    a & \in &  \mathtt{Arg} \\
    x &\in &\mathtt{Var}  \\
    n &\in &\mathtt{Val} \\
    \mathtt{agg} &\in& \{\mathtt{sum,min,max}\} 
    \\[2mm] 
\stx{e} 
& \hspace{-9mm} \coloneqq & \hspace{-5mm}
n
\ALT
\stx{a}
\ALT \stx{e\oplus e}
\ALT \stx{e\odot e}
\ALT e \circ e
\ALT 
\mathtt{iter}
\ALT \stx{\mathtt{agg}(x.f)} \ALT
\stx{\func{at}^e(x.f)}
\\
\phi_{} 
& \hspace{-9mm} \coloneqq & \hspace{-5mm}
\mathtt{this}.f \odot e \ALT
%\IN\,(x) \ALT
\phi_{} \circ \phi_{}
\\
\stx{q} 
& \hspace{-9mm} \coloneqq & \hspace{-5mm}
\stx{x:=}\; \sql{select} {\set{f}}\; \sql{from} {R}\ \sql{where} \stx{\phi}\ALT
\sql{update} {R}\;\sql{set} \set{{f} = e} \; \sql{where} \phi
%\\ & &
%\sql{insert} \sql{into} {R}\ \sql{values}\, \set{f=e} \ALT
%\sql{delete \;from} {R}\; \sql{where} \phi
 \\[0.2mm]
\stx{c} 
& \hspace{-9mm} \coloneqq & \hspace{-5mm}
\stx{q}
\ALT
\mathtt{iterate(e)\{c\}}
\ALT
\mathtt{if(e)}\{c\}
\ALT
\mathtt{skip}
\ALT
c;c
%\\
%T & \coloneqq & t(\seq{a})\{c\}
%\\
%R & \coloneqq & r:\seq{f}
\end{array}
  \end{mathpar}
\end{minipage}
\hfill
%%%%%%%%%%%%%%%%%%%%%%%%%%%%%%%%%%%%%%%%
\begin{minipage}[t]{0.43\textwidth}
  \begin{mathpar}
    \footnotesize
  \begin{array}{rcl}
    \oplus & \in & \{+, -, \times, /\} \\
    \odot & \in & \{<, \leq, =, >, \geq\} \\
    \circ & \in & \{\wedge,\vee\}  \\
T & \coloneqq & t(\set{a})\{c; \func{return}\; e \}
\\
R & \coloneqq & \rho:\set{f}
\\
F & \coloneqq & \langle \overline{f : n} \rangle
\\
P & \coloneqq & (\set{R},\set{T})
  \end{array}
  \end{mathpar}
\end{minipage}
\vspace{1mm}
\end{mdframed}
\end{minipage}
  \caption{Syntax of database programs}
  \label{fig:syntax}
\end{figure}

The syntax of our database programs is given in
\autoref{fig:syntax}. A program $P$ is defined in terms of a set of
database schemas ($\set{R}$), and a set of transactions ($\set{T}$).
A database schema consists of a schema name ($\rho$) and a set of
field names ($\set{f}$). A database \emph{record} ($F$) for schema $R$
is comprised of a set of value bindings to $R$'s fields. A database
table is a set of records.  Associated with each schema is a non-empty
subset of its fields that act as a \emph{primary key}.  Each
assignment to these fields identifies a unique record in the table. In
the following, we write $R_\func{id}$ to denote the set of all
possible primary key values for the schema $R$. In our model, a table
includes a record corresponding to \emph{every} primary key.  Every
schema includes a special Boolean field, $alive\in\func{FldName}$,
whose value determines if a record is actually present in the
table. This field allows us to model $\sql{DELETE}$and
$\sql{INSERT}$commands without explicitly including them in our
program syntax.

Transactions are uniquely named, and are defined by a sequence of
parameters, a body, and a return expression.  The body of a
transaction ($c$) is a sequence of \emph{database commands} ($q$) and
\emph{control commands}.  A database command either modifies or
retrieves a subset of records in a database table.  The records
retrieved by a database query are stored locally and can be used in
subsequent commands. Control commands consist of conditional guards,
loops, and return statements. Both database commands
($\sql{select}$and $\sql{update}\!$) require an explicit \emph{where
  clause} ($\phi$) to filter the records they retrieve or update.
$\phi_\func{fld}$ denotes the set of fields appearing in a clause $\phi$.
%In addition to arithmetic comparisions on fields and
%standard boolean operators, where clauses can use the $\IN(x)$ clause
%to reference records that were stored in the variable $x$ by a
%previous query. \BD{Should include an concrete example of the $\IN(x)$
%  clause in action.}

Expressions ($e$) include constants, transaction arguments, arithmetic and
Boolean operations and comparisons, iteration counters and field
accessors. The values of field $f$ of records stored in a variable $x$
can be aggregated using $\func{agg}(x.f)$, or accessed individually,
using $\func{at}^e(x.f)$.

%%%%%%%%%%%%%%%%%%%%%%%%%%%%%%
\subsection{Data Store Semantics}
\label{subsec:system_model}
%%%%%%%%%%%%%%%%%%%%%%%%%%%%%%

\begin{comment}
The precise semantics of programs written in this core language
depends upon the consistency and isolation guarantees on transactional
operations provided by the underlying database.  These guarantees
influence the visibility of updates performed by one transaction to
another.  We express these visibility properties axiomatically, over
sequences of events called \emph{histories} and visibility relations
defined over histories, which we describe below.
\end{comment}

Database states $\Sigma$ are modeled as a triple $(\store,\vis,\cnt)$,
where $\store$ is a set of \emph{database events} ($\eta$) that
captures the history of all reads and writes performed by a program
operating over the database, and $\vis$ is a partial order on those
events.  The execution counter, $\cnt$, is an integer that represents
a global timestamp that is incremented every time a database command
is executed; it is used to resolve conflicts among concurrent
operations performed on the same elements, which can be used to define
a \emph{linearization} or \emph{arbitration order} on
updates~\cite{burckhardt14}. Given a database state ($\Sigma$), and a
primary key $r \in R_\texttt{id}$, it is possible to reconstruct
each field $f$ of a record $r$, which we denote as
$\Sigma(r.f)$.

Retrieving a record from a table $R$ generates a set of \emph{read
  events}, $\rdd{(\tau, r, f)}$, which witness that the field $f$ of
the record with the primary key $r\in R_\func{id}$ was
accessed when the value of the execution counter was $\tau$.
Similarly, a \emph{write event}, $\wrr{(\tau,r,f, n)}$,
records that the field $f$ of record $r$ was assigned the value $n$ at
timestamp $\tau$.  The timestamp (resp. record) associated with an event $\eta$ is
denoted by $\eta_\tau$ (resp. $\eta_r$).

% Central to our development is the notion of \emph{visibility} that
% captures when a transaction's effects become available to other
% transactions.

Our semantics enforces record-level atomicity guarantees: transactions
never witness intermediate (non-committed) updates to a record in a
table by another concurrently executing one.  Thus, all updates to
fields in a record from a database command happen atomically.  This
form of atomicity is offered by most commercial database systems, and
is easily realized through the judicious use of locks. Enforcing
stronger multi-record atomicity guarantees is more challenging,
especially in distributed environments with replicated database
state~\cite{kemmevldb,twopl,redblueosdi,bernsigmod13,bailisvldb}. In
this paper, we consider behaviors induced when the database guarantees
only a very weak form of consistency and isolation that allows
transactions to 
%(a) witness \emph{partial} updates of committed
%actions and (b) 
see an \emph{arbitrary subset} of committed updates by
other transactions.  Thus, a transaction which accesses multiple
records in a table is not obligated to witness \emph{all} updates
performed by another transaction on these records.
%and is allowed to
%witness changes to the records it does observe by other transactions
%as it executes.
%The behaviors admitted by our programming model is
%akin to those possible under eventually consistent data stores %that
%operatee in geo-replicated distributed environments, such as %DynamoDB
%and Cassandra~\cite{Cassandra,DynamoDB,HBase}.

% To capture these behaviors, our formalization introduces
% a \emph{visibility} relation between events, $\vis$, that
% relates two events if one witnesses the other at the time of its creation,
% and establishes a \emph{local view} of the database  whenever a
% database operation ($\sql{select}$or $\sql{update}$\!) is executed:
% =======

To capture these behaviors, we use a \emph{visibility}
relation between events, $\vis$, that relates two events when one
witnesses the other in its \emph{local view} of the database at the
time of its creation. A local view is captured by
the relation $\lhd\subseteq\Sigma\times\Sigma$ between database states,
which is constrained as follows:
 \ruleLabel{ConstructView}
 \vspace{-1.25mm}
\begin{mathpar}
\RULE{
\store'\subseteq\store
\qquad
\forall_{\eta'\in\store'} \forall_{\eta\in\store}\,
(\eta_r\,=\,\eta'_r\,\wedge\,\eta_{\tau}\,=\,\eta'_{\tau}) \Rightarrow (\eta\in\store')
\\
\vis'=\vis|_{\store'}
\qquad
\cnt'=\cnt
}{(\store',\vis',\cnt') \lhd (\store,\vis,\cnt)}
\end{mathpar}
The above definition ensures that an event can only be present in a
local view, $\store'$, if all other events {\bf on the same record with the same counter value} are also present in $\store'$ (ensuring record-level atomicity).  Additionally, the visibility
relation permitted on the local view, $\vis'$, must be consistent with
the global visibility relation, $\vis$.

\def\horizontalSpace{2.6mm}
\begin{figure*}[t]
  \begin{minipage}{1.06\textwidth}
      \centering
      \hspace{-9mm}
  \begin{minipage}{\textwidth}
\hfill
\begin{minipage}[t]{.3\textwidth}
  \ruleLabel{txn-invoke}
$$
\RULE{
  n\in\func{Val}
  \qquad
  t(\set{a})\{c; \mathtt{return}\; e\} \in P_\func{txn}
}
%% ------------------------------------------------------------------------------------------------
{
\stx{
  \Rstep{\Sigma,\Gamma}
    {}% ->
    {\Sigma,\Gamma \!\cup\! \{t: c[\overline{a/n}]; \func{skip}\; ; e[\overline{a/n}]; \varnothing \}}
  }
}
$$
\end{minipage}
\hfill
\begin{minipage}[t]{.28\textwidth}
  \ruleLabel{txn-step}
$$
\RULE{
\step{\Sigma, \Delta, c}{}{\Sigma', \Delta', c'}

}
%% ------------------------------------------------------------------------------------------------
{
\stx{
  \Rstep{\Sigma,\{t: c; e; \Delta\}\!\cup\!\Gamma}
    {}% ->
    {\Sigma',\{t: c'; e; \Delta'\}\!\cup\!\Gamma}
  }
}
$$

%\caption{A simplified rule capturing execution of an $\sql{update}$ operation}
\end{minipage}
\hfill
\begin{minipage}[t]{.32\textwidth}
  \ruleLabel{txn-ret}
$$
\RULE{
  e \not\in \mathtt{Val}
  \qquad
  \Delta,e\Downarrow m
}
%% ------------------------------------------------------------------------------------------------
{
\stx{
  \Rstep{\Sigma,\{t: \func{skip}; e; \Delta\}\!\cup\!\Gamma}
    {}% ->
    {\Sigma,\{t: \func{skip}; m; \Delta\}\!\cup\!\Gamma}
  }
}
$$

%\caption{A simplified rule capturing execution of an $\sql{update}$ operation}
\end{minipage}
\hfill
\\[\horizontalSpace]
\hfill
\begin{minipage}[t]{.17\textwidth}
\ruleLabel{seq}
  $$
\RULE{
  \stx{
    \step{\Sigma, \Delta, c}
    {}% ->
    {\Sigma', \Delta', c''}
  }

}
{
\stx{
  \step{\Sigma, \Delta, c; c'}
    {}% ->
    {\Sigma', \Delta', c''; c'}
  }
}
$$
\end{minipage}
\hfill
\begin{minipage}[t]{.15\textwidth}
\ruleLabel{skip}
$$
\RULE{
  \\
}
%% ------------------------------------------------------------------------------------------------
{
\stx{
  \step{\Sigma, \Delta, \mathtt{skip};c}
    {}% ->
    {\Sigma, \Delta, c}
  }
}
$$
\end{minipage}
\hfill
%
%
%
%
%
%
%%%%%%%%%%%%%%%%%%%%%%%%%%%%%%%%%%%%%%%%%%%%%%%%%%%%%%%%%%%%%%%%%%%%%%%%%%%%%%%%%%%%%%%%%%%%%%%%%%%%
\begin{minipage}[t]{.17\textwidth}
\ruleLabel{cond-t}
$$
\RULE{
\Delta, e \Downarrow \mathtt{true}
%\\
%\Lambda'=\Lambda[i\xmapsto{\cup}(e,\mathtt{true})]
}
{
\stx{
  \step{\Sigma, \Delta, \mathtt{if}(e)\{c\}}
    {}% ->
    {\Sigma, \Delta, c}
  }
}
$$
\end{minipage}
\hfill
\begin{minipage}[t]{.19\textwidth}
\ruleLabel{cond-f}
$$
\RULE{
\Delta, e \Downarrow \mathtt{false}
%\\
%\Lambda'=\Lambda[i\xmapsto{\cup}(e,\mathtt{false})]
}
{
\stx{
  \step{\Sigma, \Delta, \mathtt{if}(e)\{c\}}
    {}% ->
    {\Sigma, \Delta, \mathtt{skip}}
  }
}
$$
\end{minipage}
\hfill
\begin{minipage}[t]{0.28\textwidth}
\ruleLabel{iter}
  $$
\RULE{
  \Delta, e\Downarrow n
%  \\
%\Lambda'=\Lambda\cup\{(e,n)\}
}
{
\stx{
  \step{\Sigma, \Delta, \mathtt{iterate}(e)\{c\} }
    {}% ->
    {\Sigma, \Delta, \mathtt{concat}(n,c)}
  }
}
$$
\end{minipage}
\hfill
\end{minipage}
\\[\horizontalSpace]
\begin{minipage}{1.02\textwidth}
\centering
\hspace{-13mm}
      \hfill
\begin{minipage}[b]{0.49\textwidth}
\ruleLabel{select}
\vspace{-1mm}
    $$
\RULE{
%
%  \Sigma \equiv (\store,\vis,\func{cnt})
%  \qquad
  \Sigma' \lhd \Sigma 
  \qquad
  \varepsilon_1 = \{\rdd{(\cnt,r,f)}\ALT r\in R_\func{id} \;\wedge\; f\in \phi_\func{fld} \}
  \\
   \qquad
% \varepsilon = \varepsilon_1\cup\varepsilon_2
  \func{results} = \{(r, \langle \overline{f : n}\rangle) \ALT r\in R_\func{id}
  \; \wedge\; 
  \Sigma'(r) = \langle \overline{f' : n'} \rangle
  \;\wedge\; \\ \qquad \qquad
  \Delta, \phi(\langle \overline{f' : n'} \rangle) \Downarrow \mathtt{true}
  \;\wedge\;
  \overline{f : n} \subseteq \overline{f' : n'}
  \}\\
  \varepsilon_2 = \{\rdd{(\cnt,r,f'_i)}\ALT
  (r, \langle \overline{f' : n'}\rangle) \in \func{results}
  \;\wedge\;
  f'_i\in \overline{f}
  \}
\\
\store' = \Sigma.\store\cup \varepsilon_1 \cup \varepsilon_2
\qquad
\vis'= \Sigma.\vis\cup\{(\eta,\eta')|\, \eta'\in\varepsilon_1 \cup
\varepsilon_2\; \wedge\; \eta \in
\Sigma'.\store)\}
  }
{
\stx{
  \step{\Sigma, \Delta,
    x:=\sql{select} \overline{f}\; \sql{from} R \; \sql{where} \phi
}
    {}% ->
    {(\store',\vis',\cnt+1), \Delta[x\mapsto \mathtt{results}], \texttt{skip}}
  }
}
$$
\end{minipage}
\hfill
\begin{minipage}[b]{0.42\textwidth}
\ruleLabel{update}
\vspace{-1mm}
  $$
\RULE{
% \Sigma \equiv (\store,\vis,\cnt)
% \qquad
% \sigma\sqsubseteq \Sigma
  \Sigma' \lhd \Sigma \\
% \qquad
\stx{\varepsilon} =
  \{\wrr{(\cnt,r,f_i,m)} \,|\,
  r\in R_\func{id}
  \;\wedge\;
  \Sigma'(r) = \langle \overline{f : n} \rangle
  \;\wedge\;
  \qquad\qquad\qquad
  \\
  \qquad
  \Delta, \phi(\langle \overline{f : n} \rangle) \Downarrow \mathtt{true}
  \;\wedge\;
  (f_i=e_i) \in \overline{f=e}
  \;\wedge\;
  \Delta,e_i \Downarrow m
  \}
\\
\store' = \Sigma.\store\cup\varepsilon
\qquad
\vis'= \Sigma.\vis\cup\{(\eta,\eta')|\, \eta'\in\varepsilon \wedge \eta
\in \Sigma'.\store \}
}
{
\stx{
  \step{\Sigma, \Delta,
    \sql{update} R \;\sql{set} \overline{f=e}\; \sql{where} \phi
}
    {}% ->
    {(\store',\vis',\cnt+1), \Delta, \texttt{skip}}
  }
}
$$
\end{minipage}
\hfill
\hspace{1mm}
\end{minipage}
\end{minipage}
\caption{Operational semantics of weakly-isolated database programs.}
\label{fig:operational_semantics}
\end{figure*}

\autoref{fig:operational_semantics} presents the operational semantics
of our language, which is defined by a small-step reduction relation,
$\Rightarrow\ \subseteq \Sigma\times\Gamma\times\Sigma\times\Gamma$,
between tuples of data-store states ($\Sigma$) and a set of currently
executing transaction instances
($\Gamma \subseteq c \times e \times (\mathtt{Var} \rightharpoonup
\overline{R_\texttt{id} \times F})$). A transaction instance is a
tuple consisting of the unexecuted portion of the transaction body
(i.e., its continuation), the transaction's return expression, and a
local store holding the results of previously processed query
commands. The rules are parameterized over a program $P$ containing
a set of transactions, $P_\func{txn}$.  At every step, a new transaction
instance can be added to the set of currently running transactions via
(\textsc{txn-invoke}). Alternatively, a currently running transaction
instance can be processed via (\textsc{txn-step}). Finally, if the
body of a transaction has been completely processed, its
$\func{return}$ expression is evaluated via (\textsc{txn-ret}); the
resulting instance simply records the binding between the transaction
instance ($t$) and its return value ($m$).

The semantics of commands are defined using a local reduction relation
($\rightarrow$) on database states, local states, and commands. The
semantics for control commands are straightforward outside of the
(\textsc{iter}) rule, which uses an auxiliary function
$\func{concat}(n,c)$ to sequence $n$ copies of the command $c$.
Expression evaluation is defined using the big-step relation
$\Downarrow \,\subseteq (\mathtt{Var} \rightharpoonup
\overline{R_\texttt{id} \times F}) \times e\times \func{Val}$ which,
given a store holding the results of previous query commands,
determines the final value of the expression. The full definition of
$\Downarrow$ can be found in the supplementary material.

The semantics of database commands, given by the (\textsc{select}) and
(\textsc{update}) rules, expose the interplay between global and local
views of the database.  Both rules construct a local view of the
database $(\Sigma' \lhd \Sigma)$ that is used to select or update the
contents of records.  Neither rule imposes any restrictions on
$\Sigma'$ other than the consistency constraints defined by
(\textsc{ConstructView}). The key component of each rule is how it
defines the set of new events ($\varepsilon$) that are added to the
database. In the \textsc{select} rule, $\varepsilon_1$ captures the
retrievals that occur on database-wide scans to identify records
satisfying the $\sql{select}$command's where clause. In an abuse of
notation, we write
$\Delta,\phi(\langle \overline{f : n} \rangle) \Downarrow n$ as
shorthand for
$\Delta, \phi[\overline{\mathtt{this.f}/n}] \Downarrow n$.
$\varepsilon_2$ constructs the appropriate read events of these
retrieved records.  The (\textsc{update}) rule similarly defines
$\varepsilon$, the set of write events on the appropriate fields of
the records that satisfy the where clause of the $\sql{update}$command
under an arbitrary (but consistent) local view ($\Sigma'$) of the
global store ($\Sigma$).  Both rules increment the local timestamp,
and establish new global visibility constraints reflecting the
dependencies introduced by the database command, i.e., all the
generated read and write events depending upon the events in the local
view.  All updates are performed atomically, as the set of
corresponding write events all have the same timestamp value, however,
other transactions are not obligated to see all the effects of an
update since their local view may only capture a subset of these
events.

%%%%%%%%%%%%%%%%%%%%%%%%%%%%%%
\subsection{Anomalous Data Access Pairs}
%%%%%%%%%%%%%%%%%%%%%%%%%%%%%%

We reason about concurrency bugs on transactions induced by our data
store programming model using \emph{execution histories}; finite
traces of the form:
$\Sigma_1,\Gamma_1 \Rightarrow \Sigma_2,\Gamma_2 \Rightarrow \dots
\Rightarrow \Sigma_k,\Gamma_k$ that capture interleaved execution of
concurrently executing transactions.  A \emph{complete} history is one
in which all transactions have finished, i.e., the final $\Gamma$ in
the trace is of the form:
$\{ t_1 : \texttt{skip}; m_1, \Delta_1\}\cup \ldots \cup \{t_k :
\texttt{skip}; m_k, \Delta_k \}$.  As a shorthand, we refer to the
 final state in a history $h$ as $h_\func{fin}$. % In the following, we
% write $\hist(P)$ to represent the set of all complete execution
% histories for a program $P$.
A \emph{serial} execution history satisfies two important properties:

\begin{enumerate}
\item {\bf Strong Atomicity:}
 $(\forall{\eta,\eta'}.\; \eta_{\mathtt{cnt}} < \eta'_{\mathtt{cnt}} \Rightarrow
           \vis(\eta,\eta')) \wedge \forall{\eta,\eta',\eta''}.\; \st(\eta,\eta') \wedge
          (\vis(\eta,\eta'')\Rightarrow \vis(\eta',\eta''))$
\item {\bf Strong Isolation:}:
 $\forall{\eta,\eta',\eta''}.\;
\st(\eta,\eta')\wedge \vis(\eta'',\eta')\Rightarrow \vis(\eta'',\eta)$.
\end{enumerate}

The strong atomicity property prevents non-atomic interleavings of
concurrently executing transactions.  The first constraint linearizes
events, relating timestamp ordering of events to visibility.  The
second generalizes this notion to multiple events, obligating
\emph{all} effects from the same transaction (identified by the $\st$
relation) to be visible to another if any of them are; in particular,
any recorded event of a transaction $T_1$ that precedes an event in
$T_2$ requires all of $T_1's$ events to precede all of $T_2$'s.

The strong isolation property prevents a transaction from observing
the commits of other transactions once it begins execution.  It does so
through visibility constraints on a transaction $T$ that
require any event $\eta''$ generated by any other transaction that is
visible to an event $\eta'$ generated by $T$ to be visible to any
event $\eta$ that precedes it in $T$'s execution.

\begin{comment}
While serial executions represent a comprehensible and thus desirable
execution model for developers, they impose significant cost on
implementations, by severely restricting the set of admissible
behaviors allowed~\cite{BA13}.  In recent years, many static and
dynamic program analysis, monitoring, and verification methods have
been introduced, which aim to find undesirable non-serial execution
histories of a given program allowed by implementations i.e. to
detect \emph{serializability anomalies}~\cite{KR18,BDM+17,BR18, WB17}.
\end{comment}

A \emph{serializability anomaly} is an execution history with a final state
that violates at least one of the above constraints.  These sorts of
anomalies capture when the events of a transaction instance
are either not made visible to other events in totality (in the
case of a violation of strong atomicity) or which themselves
witness different events (in the case of a violation of
strong isolation). Both kinds of anomalies can be eliminated by
identifying commands which generate sets of problematic events and
altering them to ensure \emph{atomic execution}.  Two events are
executed atomically if they witness the same set of events and they
are both made visible to other events simultaneously, i.e.
$
\func{atomic}(\eta,\eta')\equiv\forall \eta''.\;
(\vis(\eta,\eta'')\Rightarrow
\vis(\eta',\eta''))\wedge(\vis(\eta'',\eta)\Rightarrow \vis(\eta'',\eta'))
$.

%In the remainder of this paper, our goal is to ensure these atomicity
%requirements are realized without requiring potentially expensive
%run-time checks, lock injection, or other concurrency control
%mechanisms.  Instead, we propose to refactor a program's database
%schema to prevent conflicting events that lead to an anomaly.

Given a program $P$, we define a \emph{database access pair} ($\dap$)
as a quadruple $(c_1,\overline{f}_1,c_2,\overline{f}_2)$ where $c_1$
and $c_2$ are database commands from a transaction in $P$, and
$\overline{f}_1$ (resp. $\overline{f}_2$) is a subset of the fields
that are accessed by $c_1$ (resp. $c_2$). An access pair is anomalous
if there is at least one execution in the execution history of P that
results in an event generated by $c_1$ accessing a field
$f_1\in \overline{f}_1$ which induces a serializability anomaly with
another event generated by $c_2$ accessing field
$f_2\in \overline{f}_2$. An example of an anomalous access pair for
the program from in \autoref{sec:over}, is
$(S1,\{\func{st\_name}\},S2,\{\func{em\_addr}\})$ and
$(U1,\{\func{st\_name}\},$ $U2, \{\func{em\_addr}\})$; this pair
contributes to that program's non-repeatable read anomaly from
\autoref{fig:anml}.

We now turn to the development of an automated static repair strategy
that given a program $P$ and a set of anomalous access pairs produces
a semantically equivalent program $P'$ with fewer anomalous access
pairs. In particular, we repair programs by \emph{refactoring} their
database schemas in order to benefit from record-level atomicity
guarantees offered by most databases, without introducing new
observable behaviors. We elide the details of how anomalous access
pairs are discovered, but note that existing tools~\cite{BR18,
  Rahmani:2019:Clotho} can be adapted for this
purpose. \autoref{sec:impl} provides more details about how this works
in \tool{}.

% \begin{definition}
%   Given a program $P$, an anomaly oracle, $\mathcal{O}:P \rightarrow
%   \set{\dap}$, returns a set of anomalous access pairs from that program.
% \end{definition}

% The focus of the following section is the development of an automated
% static repair strategy that given a program $P$ and an anomaly oracle
% $\mathcal{O}$ finds a semantically equivalent program $P'$ that has
% fewer anomalous access pairs than $P$ as determined by $\mathcal{O}$.
% \begin{comment}
% Specifically, we develop a tool that repairs programs by
% \emph{refactoring} their database schemas in order to benefit from
% record-level atomicity guarantees offered by most databases, without
% introducing new observable behaviors.

% We first formalize the notion
% of refactoring and repair, and discuss the realization of this
% formalism in an implementation thereafter.
% \end{comment}
% The discussion elides
% specific details on how $\mathcal{O}$ is implemented to simplify the
% presentation.  Although the choice of a particular oracle is
% orthogonal to our discussion, existing tools such as~\cite{BR18} or
% ~\cite{Rahmani:2019:Clotho} can be adapted for this purpose;
% we provide additional details in~\autoref{sec:impl}.

\section{Refactoring of Database Programs}
\label{sec:refactor}

\begin{comment}
  In this section, we identify conditions under which a program
  refactoring is safe to perform, using the notion of \emph{program
    refinement}.  A refactored program $P'$ is a refinement of the
  original program $P$ if:
\begin{enumerate}
\item[(A1)] Every complete history of $P'$ has a corresponding
  complete history in $P$ with the same set of finalized transaction
  instances,
  \item[(A2)] Every serializable execution of $P$ has a corresponding
  execution in $P'$.
\end{enumerate}
The first condition ensures that $P'$ does not introduce any new
behaviors over $P$, while the second ensures that $P'$ does not remove
any desirable behavior exhibited by $P$.
\end{comment}

In this section, 
we establish the soundness properties on the space of database program refactorings and then introduce our particular choice of sound refactoring rules.

The correctness of our approach relies on being able to show that each
program transformation maintains the invariant that \emph{at every
step in any history of a refactored program, it is possible to
completely recover the state of the data-store for a corresponding
history of the original program}. To establish this property, we begin by
formalizing the notion of a \emph{containment} relation between tables.

\subsection{Database Containment}
\label{sec:DBContain}
Consider the
tables in \autoref{fig:example}, which are instances of the schemas
from \autoref{sec:over}. Note that every field of
$\func{COURSE}_0$ can be computed from the values of some other field
in either the $\func{STUDENT}_0$ or $\func{COURSE\_ST\_CNT\_LOG}_0$
tables: $\func{co\_avail}$ corresponds to the value of the
$\func{st\_co\_avail}$ field of a record in $\func{STUDENT}_0$,
while $\func{co\_st\_cnt}$ can be recovered by summing up the values
of the $\func{co\_cnt\_log}$ field of the records in
$\func{COURSE\_ST\_CNT\_LOG}_0$ whose $\func{co\_id}$ field has the same
value as the original table.

The containment relation between a table (e.g. $\func{COURSE}_0$) and
a set of tables (e.g. $\func{STUDENT}_0$ and
$\func{COURSE\_ST\_CNT\_LOG}_0$) is defined using a set of mappings
called \emph{value correspondences}~\cite{Wang:2019:Synthesizing}.  A
value correspondence captures how to compute a field in the contained
table from the fields of the containing set of tables.
%assuming the first table is contained in the second
%set.
Formally, a value correspondence between field $f$ of schema $R$
and field $f'$ of schema $R'$ is defined as a tuple $(R, R',
f,f',\theta,\alpha)$ in which:
\begin{enumerate*}[label=(\roman*)]
\item a total \emph{record correspondence function}, denoted by
  $\theta:R_\func{id}\rightarrow \set{R'_\func{id}}$,
  relates
  every record of any instance of $R$ to
  a \emph{set} of records in any instance of $R'$
  and
\item a fold function on values,
  denoted by $\alpha: \set{\func{Val}}\rightarrow \func{Val}$ is used
  to aggregate a set of values.
\end{enumerate*}
We say that a table $X$ is contained by a set of tables
$\overline{X}$ under a set of value correspondences $V$, if $V$
accurately explains how to compute $X$ from $\overline{X}$, i.e.
\small
\begin{mathpar}
\hspace{-4mm}
  \begin{array}{lll}
    X \sqsubseteq_V \overline{X} & \hspace{-1mm} \equiv  &\hspace{-1mm}
    \forall f\in X_\func{fld}.\;
      \exists (R,R',f,f',\theta,\alpha)\in V.\;
    \exists X' \in \overline{X}.\\ & &
  \hspace{-10mm}\forall r\in R_\func{id}.\;
    X(r.f) ~=~ \alpha(\{m\,|\,r'\in\theta(r)\wedge X'(r'.f')=m\})

  \end{array}
\end{mathpar}
\normalsize

\begin{figure}[t]
  \centering
  \begin{minipage}{0.49\textwidth}
    \centering
  \includegraphics[width=\textwidth]{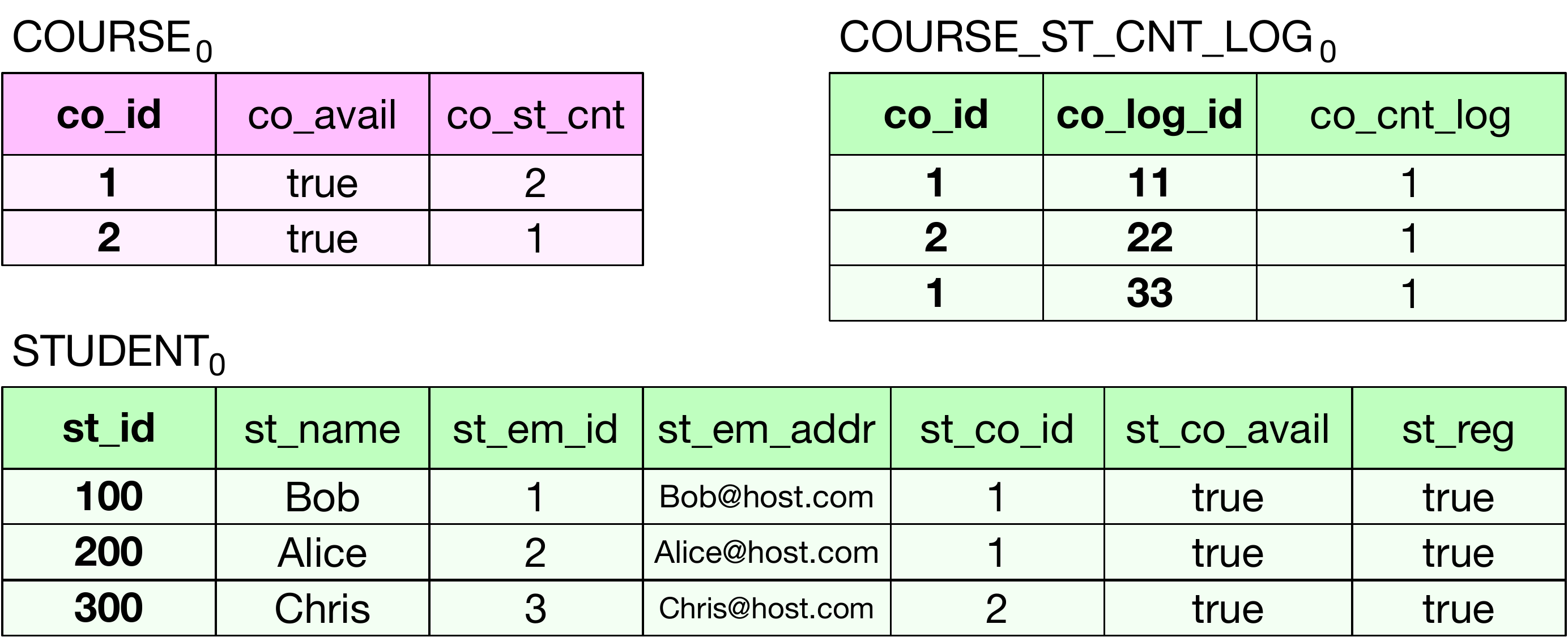}
  \caption{An example illustrating value correspondences.
  %Information
  %  found in $\func{COURSE}_0$ can be extracted from data found in
  %  tables $\func{COURSE\_ST\_CNT\_LOG}_0$ and $\func{STUDENT}_0$
  }
  \label{fig:example}
\end{minipage}
\end{figure}

For example, the table
$\func{COURSE}_0$ is contained in the set of tables
\(\{\func{STUDENT}_0, \func{COURSE\_ST\_CNT\_LOG}_0\}\) under the pair of value
correspondences,
\((\func{COURSE}, \)
\(\func{STUDENT}, \)
\(\func{co\_avail}, \)\;\;\;
\(\func{st\_co\_avail}\),
\(\theta_1,\)
\(\func{any})\) and
\((\func{COURSE},\) \(\func{COURSE\_ST\_CNT\_LOG},\)
\(\func{co\_st\_cnt},\func{co\_cnt\_log},\) \(\theta_2,\func{sum})\),
where  $\theta_1(1)=\{100,200\}$, $\theta_1(2)=\{300\}$,
$\theta_2(1)=\{(1,11),(1,33)\}$ and
$\theta_2(2)=\{(2,22)\}$.
The aggregator function $\func{any}:\set{\func{Val}}\rightarrow \func{Val}$  returns a non-deterministically chosen value from a
set of values. The containment relation on tables is straightforwardly lifted to
data store states, denoted by $\Sigma\sqsubseteq_V\Sigma'$, if all tables in $\Sigma$ are contained by the set of tables in $\Sigma'$.

We define the soundness of our program refactorings using a pair of
\emph{refinement} relations between execution histories and between
programs.  An execution history $h'$ (where
$h'_\func{fin} = (\Sigma', \Gamma')$) is a refinement of an execution
$h$ (where $h_\func{fin}=(\Sigma,\Gamma)$) if and only if $\Gamma'$
and $\Gamma$ have the same collection of finalized transaction
instances and there is a set of value correspondences $V$ under which
$\Sigma$ is contained in $\Sigma'$, i.e.
$\Sigma \sqsubseteq_V \Sigma'$.

Lastly, we define a refactored program $P'$ to be a refinement of the original program $P$ if the following conditions are satisfied:
\begin{enumerate}[label=(\Roman*)]
\item Every history $h'$ of $P'$ has a corresponding
  history $h$ in $P$ such that $h'$ is a refinement of $h$.
  \item Every serializable history $h$ of $P$ has a corresponding
  history $h'$ in $P'$  such that
  $h'$ is a refinement of $h$.
\end{enumerate}
\noindent The first condition ensures that $P'$ does not introduce any new
behaviors over $P$, while the second ensures that $P'$ does not remove
any desirable behavior exhibited by $P$.

\subsection{Refactoring Rules}
%%%%%%%%%%%%%%%%%%%%%%%%%%%%%%%%%%%%%%%%%%%%%%%%%%%
\label{sec:refactoring-semantics}
\begin{comment}
The notion of refinement presented above is defined extensionally, in
terms of execution histories and transaction final states.  But, our
refactoring procedure is intended to be applied statically without the
benefit of having concrete table instances to guide the process.

In
this section, we first describe the constraints necessary to reason
about refactorings \emph{intensionally} and entirely in terms of modifications
to database schema and the programs that use them. After providing the
intuitions behind these constraints on the rules employed by \tool{},
we formally state a theorem establishing they are sufficient to
preserve refinement between the original and the refactored

We
end the section by describing the transformations used by \tool{} in
more detail; while not explicit in the rules, their shape is
influenced by our interest in eliminating serializability anomalies.
\end{comment}

% and thus guides the shape of refactorings our implementation
% considers. We defer the

%A safe realization of
% these rules is presented in \autoref{subsec:ins}.

We describe \tool{}'s refactorings using a relation
$\hstepl{}{}{}\ \subseteq \set{V}\times P\times \set{V}\times P$,
between programs and sets of value correspondences. The rules in
\autoref{fig:ref_rules} are templates of the three categories of transformations
employed by \tool{}. % Together, they cover many
% refactoring techniques studied in the
% literature~\cite{Ambler:2006:Refactoring,Faroult:2008:Refactoring}.
These categories are: (1) adding a new schema to the program, captured by
the rule (\textsc{intro $\rho$}); (2) adding a new field to an existing
schema $\rho$, captured by rule (\textsc{intro $\rho.f$}); and, (3)
relocating certain data from one table to another while modifying the
way it is accessed by the program, captured by the rule (\textsc{intro
  $v$}).
%The changes described by the first two rules are straightforward -
%they simply update the program's schemas, without modifying its
%transactions.  \input{figures/schema_ref_example} For example,
%\autoref{fig:schema_ref_example} shows how the $\func{STUDENT}$ schema
%from the example discussed in Section~\ref{sec:over} is transformed
%into its final version in two refactoring steps; first by
%(\textsc{intro $\func{STUDENT.st\_em\_addr}$}) and then by
%(\textsc{intro $\func{STUDENT.st\_co\_avail}$}), both of these being
%applications of the ({\sc intro} $\rho.f$) rule.

The refactorings represented by (\textsc{intro $v$}) introduce a new
value correspondence $v$, and modify the body and return expressions
of a programs transactions via a rewrite function, $\eval{.}_v$.  A
particular instantiation of $\eval{.}_v$ must ensure the same data is
accessed and modified by the resulting program, in order to guarantee
that the refactored program refines the original. At a high-level, it
is sufficient for $\eval{\cdot}_v$ to ensure the following relationship
between the original ($P$) and refactored programs ($P'$) :
\begin{itemize}
  \item[(R1)] $P'$ accesses the same data as $P$,
    which may be maintained by different schemas;
  \item [(R2)] $P'$ returns the same final value as $P$;
  \item [(R3)] and, $P'$ properly updates all data
    maintained by $P$.
\end{itemize}

%%%%%%%%%%%%%%%%%%%%%%%%%
\def\horizontalSpace{3mm}
\begin{figure}[t]
\begin{minipage}[b]{0.49\textwidth}
    \begin{mdframed}[backgroundcolor=white,linecolor=white]
%%%%%%%%%%%%%%%%%%%%%%%%%%%%%%%%%%%%
      \centering
      \hspace{-5mm}
\begin{minipage}{0.47\textwidth}
  \ruleLabel{intro $\rho$}
\begin{mathpar}
  \RULE{
    \rho\not\in\set{R}_\func{RelNames}
  }{\hstepl{\set{V},(\set{R},\set{T})}{}{\set{V},(\set{R}\cup\{\rho:\emptyset\},\set{T})}}
\end{mathpar}
\end{minipage}
\hfill
%%%%%%%%%%%%%%%%%%%%%%%%%%%%%%%%%%%%%%%%%%%%%%%%%%%%%%%%%%%%%%%%%%%%%%%%
\begin{minipage}{0.52\textwidth}
  \ruleLabel{intro $\rho.f$}
\begin{mathpar}
  \RULE{
    R = \rho:\set{f}
    \quad
    f\not\in \set{f}
    \quad
    R' = \rho:\set{f}\cup\{f\}
    }{\hstepl{\set{V},(\{R\}\cup
  \set{R},\set{T})}{}{\set{V},(\{R'\}\cup \set{R},\set{T})}}
\end{mathpar}
\end{minipage}
%%%%%%%%%%%%%%%%%%%%%%%%%%%%%%%%%%%%%%%%%%%%%%%%%%%%%%%%%%%%%%%%%%%%%%%%
\\[\horizontalSpace]
\hspace{-12mm}
\begin{minipage}{0.93\textwidth}
  \ruleLabel{intro $v$}
\begin{mathpar}
  \RULE{
    v\not\in\set{V}
    \qquad
    \set{T'} = \{ t(\set{a})\{\eval{c}_v;\func{return}\ \eval{e}_v\}   \ALT
      t(\set{a})\{c;\func{return}\ {e}\}  \in \set{T}
    \}
  }{\hstepl{\set{V},(\set{R},\set{T})}{}{\set{V}\cup\{v\},(\set{R},\set{T'})}}
\end{mathpar}
\end{minipage}

%%%%%%%%%%%%%%%%%%%%%%%%%%%%%%%%%%%%%%
\end{mdframed}
\end{minipage}
\vspace{-5mm}
\caption{Refactoring Rules}
\label{fig:ref_rules}
\vspace{-5mm}
\end{figure}

\begin{figure*}[t]
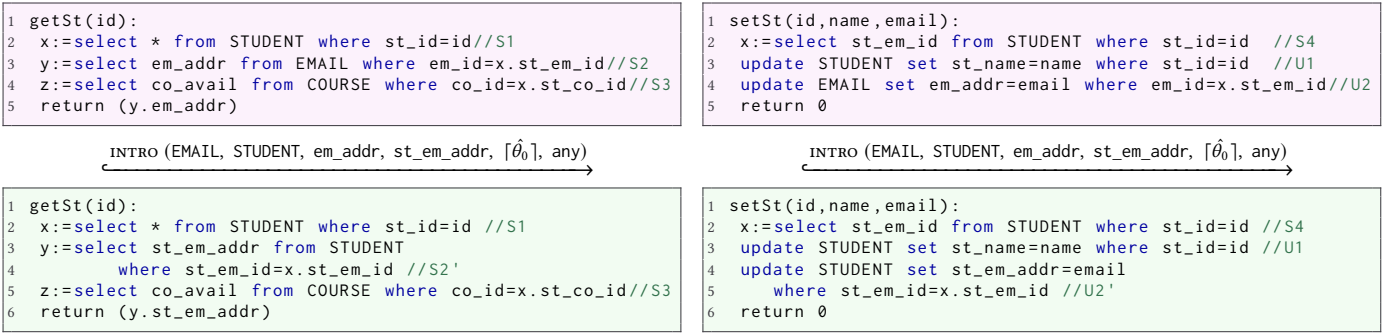


\definecolor{pgrey}{rgb}{0.26,0.25,0.28}
\definecolor{javared}{rgb}{0.2,0.2,0.7} % for strings
\definecolor{javagreen}{rgb}{0.2,0.45,0.3} % comments
\definecolor{javapurple}{rgb}{0.5,0,0.35} % keywords
\definecolor{javadocblue}{rgb}{0.25,0.35,0.75} % javadoc
\definecolor{weborange}{RGB}{0,75,0}

\lstset{language=Java,
  basicstyle=\ttfamily\scriptsize,
  breaklines=true,
  backgroundcolor=\color{Main-6},
  frame=single,
  rulecolor=\color{pgrey},
  keywordstyle=\color{black},
  stringstyle=\color{javared},
  commentstyle=\color{javagreen},
  morecomment=[s][\color{javadocblue}]{/**}{*/},
  numbers=left,
  %title=\footnotesize{getStudent(id)},
  xleftmargin=0.3em,
  xrightmargin=0.3em,
  framexleftmargin=0.9em,
  numberstyle=\tiny\color{pgrey},
  stepnumber=1,
  numbersep=5pt,
  tabsize=1,
  captionpos=b
  showspaces=false,
  showstringspaces=false,
  classoffset=2, % starting new class
  morekeywords={update,select,set,from,where},
  keywordstyle=\color{sql},
  moredelim=[is][\textcolor{red}]{\%\%}{\%\%},
}
\hspace{-2mm}
\begin{subfigure}[t]{0.49\textwidth}
\begin{minipage}[t]{\textwidth}
\begin{lstlisting}[]
getSt(id):
 x:=select * from STUDENT where st_id=id//S1
 y:=select em_addr from EMAIL where em_id=x.st_em_id//S2
 z:=select co_avail from COURSE where co_id=x.st_co_id//S3
 return (y.em_addr)
\end{lstlisting}
\vspace{-0.5mm}
\end{minipage}
\centering
$\hstepl{}{\textsc{intro $(\func{EMAIL}, \;
\func{STUDENT}, \;
\func{em\_addr},\;
\func{st\_em\_addr},\;
\lceil\hat{\theta_0}\rceil,\;
\func{any})$}}{}$
\lstset{backgroundcolor=\color{Main-7}}
\begin{minipage}[t]{\textwidth}
\vspace{-3mm}
\begin{lstlisting}[]
getSt(id):
 x:=select * from STUDENT where st_id=id //S1
 y:=select st_em_addr from STUDENT
        where st_em_id=x.st_em_id //S2'
 z:=select co_avail from COURSE where co_id=x.st_co_id//S3
 return (y.st_em_addr)
\end{lstlisting}
\end{minipage}
%\caption{}
\end{subfigure}
\hfill
\begin{subfigure}[t]{0.49\textwidth}
\lstset{backgroundcolor=\color{Main-6}}
\begin{minipage}[t]{\textwidth}
\begin{lstlisting}[]
setSt(id,name,email):
 x:=select st_em_id from STUDENT where st_id=id  //S4
 update STUDENT set st_name=name where st_id=id  //U1
 update EMAIL set em_addr=email where em_id=x.st_em_id//U2
 return 0
\end{lstlisting}
\vspace{-0.5mm}
\end{minipage}
\centering
$\hstepl{}{\textsc{intro $(\func{EMAIL}, \;
\func{STUDENT}, \;
\func{em\_addr},\;
\func{st\_em\_addr},\;
\lceil\hat{\theta_0}\rceil,\;
\func{any})$}}{}$
\lstset{backgroundcolor=\color{Main-7}}
\begin{minipage}[t]{\textwidth}
\vspace{-3mm}
\begin{lstlisting}[]
setSt(id,name,email):
 x:=select st_em_id from STUDENT where st_id=id //S4
 update STUDENT set st_name=name where st_id=id //U1
 update STUDENT set st_em_addr=email
    where st_em_id=x.st_em_id //U2'
 return 0
\end{lstlisting}
\end{minipage}
%\caption{}
\end{subfigure}
\hspace{-3mm}
\vspace{-1mm}
\caption{A single program refactoring step, where $\hat{\theta_0}(\func{EMAIL.em\_addr})=\func{STUDENT.st\_em\_addr}$}
\label{fig:prog_ref_example}
\end{figure*}

To see how a rewrite function might ensure R1 to R3, consider
the original (top) and refactored (bottom) programs presented in \autoref{fig:prog_ref_example}.
% which considers commands that do not access the
% source schema of the newly added value correspondence
This example depicts a refactoring of transactions $\func{getSt}$ and
$\func{setSt}$ to utilize a value correspondence from
$\func{em\_addr}$ to $\func{st\_em\_addr}$,
moving email addresses to the $\func{STUDENT}$ table, as described in \autoref{sec:over}.
The select
commands $\func{S1}$ and $\func{S3}$ in $\func{getS}$
remain unchanged after the refactoring, as they do not access the
affected table. However, the query $\func{S2}$, which
originally accessed the $\func{EMAIL}$ table is redirected to the
$\func{STUDENT}$ table.

More generally, in order to take advantage of
a newly added value correspondence $v$, $\eval{.}_v$ must alter every
query on the source table and
field in $v$ to use the target table of $v$ instead, so that the
new query accesses the same data as the original. This rewrite has the general form:
\[
\begin{array}{l}
[\![x\!:=\sql{select} f\; \sql{from} {R}
      \sql{where} \phi ]\!]_v \equiv \\
  x\!:=\sql{select} f'\; \sql{from} {R'}\ \sql{where}
  \func{redirect}({\phi},v.\theta)
\end{array}
\]
% where  the function  $\func{redirect}: \set{V}\times \phi\times
% \theta\rightarrow \phi$
% satisfies the following constraint:
% \begin{equation}
%       \RULE{
%         r\in R_\func{id}
% \qquad
%   \Sigma\sqsubseteq_V \Sigma'
%   \qquad
%   \Delta\sqsubseteq_V\Delta'
%   \qquad
%   \Delta,\phi(\Sigma(r))\Downarrow \func{true}
%       }{
%   \theta(r) = \{r'\,|\, \Delta',{\func{redirect}(V,\phi,\theta)}(\Sigma'(r')) \Downarrow \func{true}\}
%     }
%     \tag{R1}
% \end{equation}
Intuitively, in order for this transformation to ensure R1,
the $\func{redirect}$ function must return a new where clause
on the target table which selects a set of records
corresponding to set selected by the original
clause.

In order to preserve R2,
program expressions also need to be rewritten to evaluate to the same value as in the original
program. For example, observe that the
$\func{return}$ expression in $\func{getSt}$ is  updated to
reflect that the records held in the variable $\func{y}$ now adhere to
a different schema.
% \begin{equation}
%       \RULE{
%         v=(R,R',f,f',\theta,\alpha)
%         \qquad
%       \Delta \sqsubseteq_{V\cup\{v\}} \Delta'
%         \qquad
%         \Delta, e \Downarrow n
%       }{
%       \Delta',\eval{V,e}_v \Downarrow n
%     }
%       \tag{R2}
%     \end{equation}

%     For example, in \autoref{fig:prog_ref_example_getSet}
%     the expression $\func{y.em\_addr}$ used in the return statement is properly
%     updated in the refactored program to return the same intended information
%     regarding the student's email address.

The transformation performed in \autoref{fig:prog_ref_example} also rewrites the
update ($\func{U2}$) of transaction $\func{setSt}$.  In this case, the update
is rewritten using the same redirection strategy as ($\func{S2}$), so that it
correctly reflects the updates that would be performed by the original
program to the $\func{EMAIL}$ record.

Taken together, $R1-R3$ are sufficient to ensure that a particular
instance of \textsc{intro $v$} is sound\footnote{A complete formalization of all three
refactoring rules, their correctness criteria, and proofs of
soundness is presented in the supplementary materials of this submission.}:
\begin{theorem}
  Any instance of \textsc{intro $v$} whose instantiation of
  $\eval{\cdot}_v$ satisfies $R1-R3$ is always produces a refactored
  program that is a refinement of the original.
  %% \begin{proof}
  %%   \autoref{app:proofs}
  %% \end{proof}
  \label{th2}
\end{theorem}

\noindent Although our focus has been on preserving the semantics of
refactored programs, note that as a direct consequence of our
definition of program refinement, this theorem implies that sound
transformations do not introduce any new anomalies.

% In the next subsection, we will present one instantiation of rewriting rules
% which satisfy the properties defined in this section.

%%%%%%%%%%%%%%%%%%%%%%%%%%%%%%%%%%%%%%%%%%%%%%%%%%
% \subsection{Instantiation}
% \label{subsec:ins}
%%%%%%%%%%%%%%%%%%%%%%%%%%%%%%%%%%%%%%%%%%%%%%%%%%

We now present the instantiations of \textsc{intro $v$} used by
\tool{}, explaining along the way how they ensure R1-R3.
% In the previous subsection we discussed the high-level properties of a sound
% refactoring semantics, but did not discuss any specific implementation.
% We now turn to the question of how to actually implement the rewriting function for a real
% database program.
\begin{comment}
While we acknowledge that there may be other possible
refactorings, our current set supports all of the prevalent patterns
found in many real-world and benchmark database programs, and more
importantly, has demonstrated particular utility in eliminating
serializability anomalies.
\end{comment}

\subsubsection{The redirect rule}   % \textsc {$(R, R', f, f', \lceil\hat{\theta}\rceil, \func{any})$}}
\label{subsub:redirect}
Our first refactoring rule is parameterized over the choice of
schemas and fields and uses the aggregator $\func{any}$.
Given data store states $\Sigma$ and $\Sigma'$,
the record
correspondence is defined as:
\(\lceil\hat{\theta}\rceil(r) = \{r' \ALT r'\in R'_\func{id} \wedge \forall_{f\in R_\func{id}}\forall_{n}.\;
  \Sigma(r.f)\Downarrow n \Rightarrow
  \Sigma'(r'.\hat{\theta}(f))\Downarrow n\}\). \normalsize In essence,
  the lifted function $\hat{\theta}$ identifies how the value of the primary key
$f$ of a record $r$ can
be used to constrain the value of field $\hat{\theta}(f)$ in the target schema
to recover the set of records corresponding to $r$, i.e. $\theta(r)$.
The record correspondences from \autoref{sec:DBContain} were
defined in this manner, where
\vspace{-1mm}
\begin{flalign*}
&\hat{\theta_1}(\func{COURSE.co\_id})=\func{STUDENT.st\_co\_id}\,,\;\;\; \text{and}
\\[-1mm]
&\hat{\theta_2}(\func{COURSE.co\_id})=\func{COURSE\_CO\_ST\_CNT\_LOG.co\_id}.
\end{flalign*}
\noindent Defining the record correspondence this way ensures that if a record
$r$ is selected in $\Sigma$, the corresponding set of records in
$\Sigma'$ can be determined by identifying the values that were used
to select $r$, without depending on any particular instance of the
tables.

The definition of
$\func{redirect}$ for this rule is straightforward:
\begin{mathpar}
\func{redirect}(\phi,\lceil\hat{\theta}\rceil) =
\bigwedge_{f\in\phi_\func{fld}}
  \func{this}.\hat{\theta}(f)=\phi[f]_\func{exp}
  \vspace{-3mm}
\end{mathpar}
\normalsize
\begin{comment}
\small
\begin{align*}
  \func{redirect}(\phi,\lceil\hat{\theta}\rceil) \equiv &  \bigwedge_{f\in\phi_\func{fld}}
  \func{this}.\hat{\theta}(f)=\phi[f]_\func{exp}
                                     & \phi\ \textsf{is well-formed}\\
  \func{redirect}(\phi,\lceil\hat{\theta}\rceil) = \;& \bot & \mathsf{otherwise}.
\end{align*}
\normalsize
\end{comment}
The one wrinkle in this
definition of $\func{redirect}$  is that it is only defined when the
where clause $\phi$ is \emph{well-formed}, i.e. $\phi$ only consists
of conjunctions of equality constraints on primary key
fields. The expression used in such a constraint is denoted by
$\phi[f]_\func{exp}$.  As an example, the where clause of command
($\func{S2}$) in \autoref{fig:prog_ref_example} (left) is well-formed, where
$\phi[\func{em\_id}]_\func{exp} = \func{x.\st\_e\_id}$.  However, the
where clause in ($\func{S2}'$) after the refactoring step is not
well-formed, since it does not constrain the primary key of the
schema. This restriction ensures only select clauses accessing a
single record of the original table will be rewritten. Expressions
using variables containing the results of those queries are rewritten by substituting the source field name with the target field name, e.g.
$\eval{\func{at}^1(x.f)}_{v} \equiv  \func{at}^1(x.f')$.
\begin{comment}
(and thus the
variable containing the results of those queries will hold a single
record), allowing us to straightforwardly define the rewrite rule for
aggregations and accesses of $f$:
\small
\begin{equation}
  \label{I2-1}
  \begin{array}{lclr}
    \eval{\func{at}^e(x.f)}_{v} & := & \func{at}^e(x.f')  & \qquad e=1\\
    \eval{\func{at}^e(x.f)}_{v} & := & \bot  & \qquad e \neq 1 \\
    \eval{\func{agg}(x.f)}_{v} & := & \func{agg}(x.f')  &
    % \eval{V,\func{at}^1(f,x)}_{v} & := &
    % v.\alpha(f,x) &\quad (if\; v.\alpha\not\eq \func{any})
  \end{array}
\tag{I2-1}
\end{equation}
\normalsize

Finally, the transformation this refactoring applies to on updates can be
straightforwardly defined by reusing the definition of
$\func{redirect}(\phi,\theta)$ from above:
\small
\begin{align}
\begin{split}
  \label{I3-1}
  &\eval{\sql{update} R\; \sql{set} f=e \;\sql{where} \phi}_v \equiv \\
  &\sql{update} R'\; \sql{set} (f'=\eval{e}_v) \;\sql{where}
  \func{redirect}({\phi},v.\theta)
  \end{split}
    \tag{I3-1}
\end{align}
\normalsize
\end{comment}

\subsubsection{The logger rule}% \textsc{ $(R, LogR, f, f', \lceil\hat{\theta}\rceil, \func{sum})$}}
\label{subsub:logging}
Unfortunately, instantiating \textsc{intro $v$} is not so straightforward
when we want to utilize value correspondences with more
complicated aggregation functions than $\func{any}$. To
see why, consider how we would need to modify an $\sql{update}$ when
$\alpha=\func{sum}$ is used.  In this case, our rule transforms the program to
insert \emph{a new record} corresponding to each update performed by
the original program. Hence, the set of corresponding records in the
target table always grows and cannot be statically identified.

We enable these sorts of transformations by using \emph{logging}
schema for the target schema. A logging schema for source schema $R$
and the field $f$ is defined as follows:
\begin{enumerate*}[label=(\roman*)]
\item the target schema ($LogR$) has a primary key field, corresponding to every primary key
  field of the original schema ($R$);
\item the schema has one additional primary key field, denoted by
  $LogR.\func{log\_id}$, which allows a \emph{set} of records in $LogR$ to
  represent each record in $R$; and
\item the schema $LogR$ has a single field corresponding to the
  original field $R.f$, denoted by $LogR.f'$.
\end{enumerate*}

Intuitively, a logging schema captures the \emph{history} of updates
performed on a record, instead of simply replacing old values with new
ones.  Program-level aggregators can then be utilized to determine the
final value of each record, by observing all corresponding entries in
the logging schema. The schema $\func{COURSE\_CO\_ST\_CNT\_LOG}$ from
\autoref{sec:over} is an example of a logging schema for the source
schema and field $\func{COURSE}.\func{co\_st\_cnt}$.

Under these restrictions, we can define an implementation of
$\eval{\cdot}$ for the logger rule using $\func{sum}$ as an
aggregator. This refactoring also uses a lifted function
$\lceil\hat{\theta}\rceil$ for its value correspondence, which allows
$\eval{\cdot}$ to reuse our earlier definition of
$\func{redirect}$. We define $\eval{\cdot}$ on accesses to $f$ to use
program-level aggregators, e.g.
$\eval{\func{at}^1(x.f)}_{v} := \func{sum}(x.f')$.

\begin{comment}
\small
\begin{equation}
  \label{I2-2}
  \begin{array}{lclr}
    \eval{\func{at}^e(x.f)}_{v} & := & \func{sum}(x.f') &\qquad e=1 \\
    \eval{\func{at}^e(x.f)}_{v} & := & \bot  & \qquad e \neq 1 \\
    \eval{\func{agg}(x.f)}_{v} & := & \bot &
  \end{array}
\tag{I2-2}
\end{equation}
\normalsize
\end{comment}
Finally, the rewritten $\sql{update}$commands simply need to log
any updates to the field $f$, so its original value can be recovered in
the transformed program, e.g.
\vspace{-0.6mm}
\small
\begin{equation*}
  \begin{array}{l}
  \!\!\!\!\eval{\sql{update} R\; \sql{set} f = e + \func{at}^1(x.f) \;\sql{where}
  \phi}_v  \;\equiv \;
  \sql{update} R' \\ \;\;\; \sql{set} f'=\eval{e}_v
           \sql{where}
  \func{redirect}(\phi,v.\theta) \wedge
  R'.\mathtt{log\_id} = \func{uuid()}.
\end{array}
\end{equation*}
\normalsize

\begin{comment}
Note that the above approach can be expanded into much richer programming
models. For example, we can introduce various types of
refactorings in a program maintaining a set of objects, in order to eliminate
anomalies cause by concurrent addition or deletion of elements from the set.
\end{comment}

 Having introduced the particular refactoring rules instantiated in \tool, we are now ready to establish the soundness of those refactorings:
\begin{theorem}
(Soundness) Any sequence of refactorings using the
rewrite rules described in this section
  satisfy the correctness properties R1-R3.
  %%   \begin{proof}
  %%   \autoref{app:proofs}.
  %% \end{proof}
  \label{th1}
\end{theorem}

\def\horizontalSpace{3mm}

\begin{figure}[h]
%\begin{figure}[t]
\centering
%
%
%
% repair
%
% try_repai%
%\\[\horizontalSpace]
%
%
%
%
\begin{minipage}[t]{0.48\textwidth}
 \begin{mdframed}[backgroundcolor=grey12,linecolor=grey12]
\footnotesize
\begin{myalgorithm}
\SetKwInOut{Function}{Function}
  \DontPrintSemicolon
  \Function{$\func{repair}(P)$}
  $\set{\dap} \leftarrow \orac(P)$;\quad
$P\leftarrow \func{pre\_process}(P,\set{\dap})$

  \For{$\dap \in \set{\dap}$}{
   \lIf{$\func{try\_repair}(P,\dap)=\; P'$}{
        $P\leftarrow P'$
    }
 }
 {\bf return} $\func{post\_process}(P)$
\end{myalgorithm}

\footnotesize
\SetAlgoNoLine
\IncMargin{5em}
\SetKwInOut{Function}{Function}
\begin{myalgorithm}
  \DontPrintSemicolon
  \Function{$\func{try\_repair}(P,\dap)$}
  $ c_1 \leftarrow \dap.c_1$; \quad $c_2 \leftarrow \dap.c_2$
%  $ k_1 \leftarrow \func{kind}(c_1); \; \; k_2 \leftarrow \func{kind}(c_2)$\\
%  $ r_1 \leftarrow \func{rel}(c_1); \; \; r_2 \leftarrow \func{rel}(c_2)$\\
  
  \If{$\func{same\_kind}(c_1,c_2)$}{
    \If{$\func{same\_schema}(c_1,c_2)$}{
      {\bf return} $\func{try\_merging}(P,c_1,c_2)$
    }\uElseIf{$\func{try\_redirect}(P,c_1,c_2)=\;P'$}{
        {{\bf return} $\func{try\_merging}(P',c_1,c_2)$}
      }
    }
    {\bf return} $\func{try\_logging}(P,c_1,c_2)$
\end{myalgorithm}
\end{mdframed}
\end{minipage}
%
% try_redirect
%
%
%
%
%\begin{minipage}{0.88\textwidth}
%\centering
%\footnotesize
%\SetKwInOut{Function}{Function}
%\begin{algorithm}[H]
%  \SetAlgoLined
%  \DontPrintSemicolon
%  \Function{$\func{try\_redirect}(P,q_1,q_2)$}
%  \KwOut{$option \;\mathbb{P}$}
%  determine key correspondence (return none if not possible)\\
%  optional: determine the best split \\
%  add fields to R1\\
%  add new value correspondence \\
%  redirect q2
%\end{algorithm}
%\end{minipage}
%%
%%
%% try_delta
%%
%%
%%
%%
%\begin{minipage}{0.88\textwidth}
%\centering
%\footnotesize
%\SetKwInOut{Function}{Function}
%\begin{algorithm}[H]
%  \SetAlgoLined
%  \DontPrintSemicolon
%  \Function{$\func{try\_delta}(P)$}
%  \KwOut{$option \;\mathbb{P}$}
%  determine key and value correspondence (return none if not possible)\\
%  add a new relation \\
%  add new fields to it \\
%  add the determined key and val. corr.\\
%  call redirect all (return none if not possible) \\
%  call trim \\
%\end{algorithm}
%\end{minipage}
%
%
\caption{The repair algorithm}
\label{fig:algorithm}
\end{figure}

\section{Repair Procedure}
\label{sec:fix}

\autoref{fig:algorithm} presents our algorithm for eliminating
serializability anomalies using the refactoring rules from the
previous section.  The algorithm ($\func{repair}$) begins by applying
an anomaly detector $\mathcal{O}$ to a program to identify a set of
anomalous access pairs.  As an example, consider $\func{regSt}$ from
our running example. For this transaction, the anomaly oracle
identifies two anomalous access pairs:
\begin{align}
  &\func{(U3,\{st\_co\_id,st\_reg\},U4,\{co\_avail\})} \tag{$\dap_1$} \\
  &\func{(S5, \func{\{co\_st\_cnt\}}, \func{U4},\func{\{co\_st\_cnt\}})} \tag{$\dap_2$}
\end{align}

\noindent The first of these is involved in the dirty read anomaly from
\autoref{sec:over}, while the second is involved in the lost update
anomaly.

The repair procedure next performs a preprocessing phase, where database
commands are split into multiple commands such that each command is
involved in at most one anomalous access pair.  For example, the first
step of repairing the $\func{regSt}$ transaction is to split command
$\func{U4}$ into two update commands, as shown in
\autoref{fig:fix_example_regSt} (top).  Note that we only perform this
step if the split fields are not accessed together in other parts of
the program; this is to ensure that the splitting does not introduce
new unwanted serializability anomalies.

After preprocessing, the algorithm iterates over all detected
anomalous access pairs~($\set{\dap}$) and attempts to repair them one
by one using $\func{try\_repair}$. This function attempts to eliminate
a given anomaly in two different ways; either by merging anomalous
database commands into a single command, and/or by removing one of
them by making it obsolete.
%The former is feasible only
%    if both commands are of the same kind (e.g. both are selects); this is captured
%    by the condition in line 3.
In the remainder of this section, we present these two strategies in
more detail, using the running example from
\autoref{fig:fix_example_regSt}.

\begin{figure}[t]
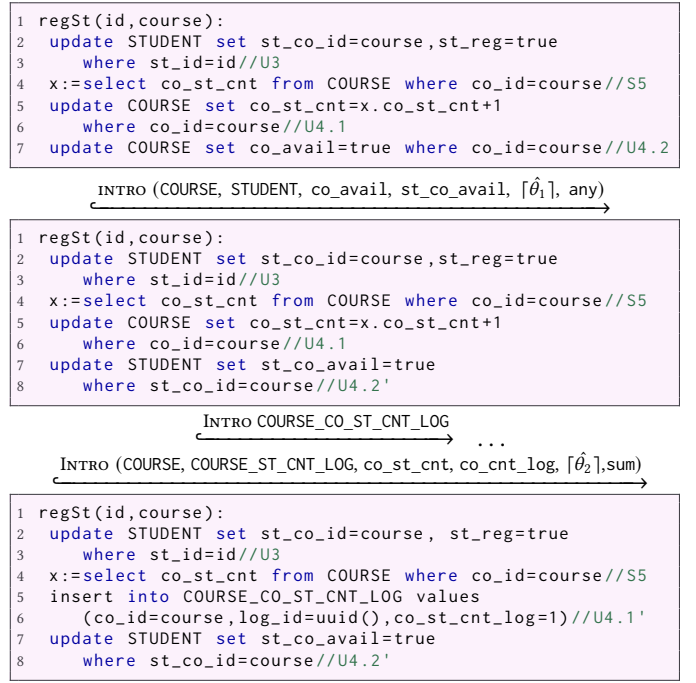


\definecolor{pgrey}{rgb}{0.26,0.25,0.28}
\definecolor{javared}{rgb}{0.2,0.2,0.7} % for strings
\definecolor{javagreen}{rgb}{0.2,0.45,0.3} % comments
\definecolor{javapurple}{rgb}{0.5,0,0.35} % keywords
\definecolor{javadocblue}{rgb}{0.25,0.35,0.75} % javadoc
\definecolor{weborange}{RGB}{0,75,0}

\lstset{language=Java,
  basicstyle=\ttfamily\scriptsize,
  breaklines=true,
  backgroundcolor=\color{Main-6},
  frame=single,
  rulecolor=\color{pgrey},
  keywordstyle=\color{black},
  stringstyle=\color{javared},
  commentstyle=\color{javagreen},
  morecomment=[s][\color{javadocblue}]{/**}{*/},
  numbers=left,
  %title=\footnotesize{getStudent(id)},
  xleftmargin=0.3em,
  framexleftmargin=1em,
  numberstyle=\tiny\color{pgrey},
  stepnumber=1,
  numbersep=5pt,
  tabsize=1,
  captionpos=b
  showspaces=false,
  showstringspaces=false,
  classoffset=2, % starting new class
  morekeywords={return,update,into,select,set,from,where},
  keywordstyle=\color{sql},
  moredelim=[is][\textcolor{red}]{\%\%}{\%\%},
}

%
\begin{comment}
\begin{minipage}[b]{0.48\textwidth}
\begin{lstlisting}[]
regSt(id,course):
 update STUDENT set st_co_id=course, st_reg=true
    where st_id=id //U3
 x:=select co_st_cnt from COURSE where co_id=course //S5
 update COURSE set co_st_cnt=x.co_st_cnt+1, co_avail=true
    where co_id=course //U4
\end{lstlisting}
\vspace{-2mm}
\end{minipage}
\footnotesize[\texttt{pre-process: split U4 to U4.1 and U4.2}]\normalsize
\end{comment}
\begin{minipage}[b]{0.478\textwidth}
  \vspace{0.5mm}
\begin{lstlisting}[]
regSt(id,course):
 update STUDENT set st_co_id=course,st_reg=true
    where st_id=id//U3
 x:=select co_st_cnt from COURSE where co_id=course//S5
 update COURSE set co_st_cnt=x.co_st_cnt+1
    where co_id=course//U4.1
 update COURSE set co_avail=true where co_id=course//U4.2
\end{lstlisting}
\vspace{-1mm}
\end{minipage}

$\hstepl{}{\textsc{intro $(\func{COURSE}, \;
\func{STUDENT}, \;
\func{co\_avail},\;
\func{st\_co\_avail},\;
\lceil\hat{\theta_1}\rceil,\;
\func{any})$}}{}$

\begin{minipage}[b]{0.478\textwidth}
\begin{lstlisting}[]
regSt(id,course):
 update STUDENT set st_co_id=course,st_reg=true
    where st_id=id//U3
 x:=select co_st_cnt from COURSE where co_id=course//S5
 update COURSE set co_st_cnt=x.co_st_cnt+1
    where co_id=course//U4.1
 update STUDENT set st_co_avail=true
    where st_co_id=course//U4.2'
\end{lstlisting}
\vspace{-1mm}
\end{minipage}

$\hstepl{}{\textsc{Intro}\;\func{COURSE\_CO\_ST\_CNT\_LOG}}{}\;\;\ldots$
$\hstepl{}{\textsc{Intro}\;(\func{COURSE},\; \func{COURSE\_ST\_CNT\_LOG},\;
\func{co\_st\_cnt},\; \func{co\_cnt\_log},\; \lceil\hat{\theta_2}\rceil,\func{sum})}{}$

\begin{minipage}[b]{0.478\textwidth}
\begin{lstlisting}[]
regSt(id,course):
 update STUDENT set st_co_id=course, st_reg=true
    where st_id=id//U3
 x:=select co_st_cnt from COURSE where co_id=course//S5
 insert into COURSE_CO_ST_CNT_LOG values
    (co_id=course,log_id=uuid(),co_st_cnt_log=1)//U4.1'
 update STUDENT set st_co_avail=true
    where st_co_id=course//U4.2'
\end{lstlisting}
\vspace{-2mm}
\end{minipage}

%\footnotesize[\texttt{post-process: merge U4.2' with U3 + remove S5}]\normalsize

\lstset{language=Java,
backgroundcolor=\color{Main-7}}

\begin{comment}
\begin{minipage}[b]{0.47\textwidth}
  \vspace{0.5mm}
\begin{lstlisting}[]
regSt(id,course):
 update STUDENT set st_co_id=course, st_co_avail=true,
    st_reg=true where st_id=id
 insert into COURSE_CO_ST_CNT_LOG values
    (co_id=course,log_id=uuid(),co_st_cnt_log=1)
\end{lstlisting}
\end{minipage}
\end{comment}

\lstset{language=Java,
  backgroundcolor=\color{Main-7},
}

%\begin{minipage}[b]{\textwidth}
%\begin{lstlisting}[]
%regSt(id,course):
%  update STUDENT set st_co_id=course, st_co_avail=true, st_reg=true where st_id=id
%  insert into COURSE_ST_CNT_LOG values (co_id=course,log_id=uuid(),co_st_cnt_log=1)
%\end{lstlisting}
%\end{minipage}

\caption{Repair steps of transaction $\func{regSt}$}
\label{fig:fix_example_regSt}
\end{figure}

We first explain the merging approach. Two database commands can only
be merged if they are of the same kind (e.g. both are
$\func{selects}$) and if they both access the same schema.  These
conditions are checked in lines 2-3.  Function $\func{try\_merge}$
attempts to merge the commands if it can establish that their where
clauses always select the exact same set of records, i.e. condition
(R1) described in \autoref{sec:refactoring-semantics}.

Unfortunately, database commands involved in anomalies are rarely on
the same schema and cannot be merged as they originally are. Using the
refactoring rules discussed earlier, \tool{} attempts to introduce
value correspondences so that the anomalous commands are redirected to
the same table in the refactored program and thus mergeable. This is
captured by the call to the procedure $\func{try\_redirect}$. This
procedure first introduces a set of fields into the schema accessed by
$c_1$, each corresponding a field accessed by $c_2$.  Next, it
attempts to introduce a sequence of value correspondences between the
two schemas using the redirect rule, such that $c_2$ is redirected to
the same table as $c_1$. The record correspondence is constructed by
analyzing the commands' where clauses and identifying equivalent
expressions used in their constraints.  If redirection is successful,
$\func{try\_merge}$ is invoked on the commands and the result is
returned (line~6).

For example, consider commands $\func{U3}$ and $\func{U4.2}$ in
\autoref{fig:fix_example_regSt} (top), which are involved in the
anomaly $\dap_1$. By introducing a value correspondence from $\func{COURSE}$ to
$\func{STUDENT}$, \tool{} refactors the program into a refined version
where $\func{U4.2}$ is transformed into $\func{U4.2'}$ and is
mergeable with $\func{U3}$.

Merging is sufficient to fix $\dap_1$, but fails to eliminate
$\dap_2$. The repair algorithm next tries to translate database
updates into an equivalent insert into a logging table using the
$\func{try\_logging}$ procedure.This procedure first introduces a new
logging schema (using the \textsc{intro $\rho$} rule) and then
introduces fields into that schema (using \textsc{intro $\rho.f$}).
It then attempts to introduce a value correspondence from the schema
involved in the anomaly to the newly introduced schema using the
logger rule.  The function returns successfully if such a translation
exists and if the select command involved in the anomaly becomes
obsolete, i.e., the command is dead-code.  For example, in
\autoref{fig:fix_example_regSt}, a value correspondence from
$\func{COURSE}$ to the logger table $\func{COURSE\_CO\_ST\_CNT}$ is
introduced, which translates the $\func{update}$ command involved in
the anomaly to an $\func{insert}$ command. The select command is
obsolete in the final version, since variable $x$ is never used.

Once all anomalies have been iterated over, \tool{} performs a
post-processing phase on the program to remove any remaining dead code
and merge commands whenever possible.  For example, the transaction
$\func{regSt}$ is refactored into its final version depicted in
\autoref{fig:over_example_refactored} after post-processing. Both
anomalous accesses ($\dap_1$ and $\dap_2$) are eliminated in the final
version of the transaction.

    %% Although the order in which the procedure treats anomalies is
    %% chosen non-deterministically, we compared different orderings in
    %% our experiments and did not observe any differences between the
    %% output programs.

%\BD{It's a little unclear to me where dead-code
%      elimination happens-- isn't it only needed after a logger
%      refactoring? Why do we do it in post processing too?}

    %Although the order in which the procedure treats anomalies is
    %chosen non-deterministically, we compared different orderings in
    %our experiments and did not observe any differences between the
    %output programs.

\section{Implementation}
\label{sec:impl}

\tool{} is a fully automated static analyzer and program repair tool
implemented in Java. Its input programs are written in a DSL similar
to the one described in~\autoref{fig:syntax}, but it would be
straightforward to extend the front-end to support popular database
programming APIs, e.g. JDBC or Python's DB-API.  \tool{} consists of a
static anomaly detection engine and a program refactoring engine and
outputs the repaired program.  The static anomaly detector in \tool{}
adapts existing techniques to reason about serializability violations
over abstract executions of a database application~\cite{BR18,KR18}.
In this approach, detecting a serializability violation is reduced to
checking the satisfiability of an FOL formula constructed from the
input program. This formula includes variables for each of the
transactional dependencies, as well as the visibility and global
time-stamps that can appear during a program's execution. The
assignments to these variables in any satisfying model can be used to
reconstruct an anomalous execution.  We use an off-the-shelf SMT
solver, Z3~\cite{Moura:2008:Z3}, to check for anomalies in the input
program and identify a set of anomalous access pairs. These access
pairs are then used by an implementation of the \textsf{repair}
algorithm build a repaired version of the input program.

\section{Evaluation}
\label{sec:eval}
%\BD{Correctness doesn't seem quite right... maybe ``safety'' or
%  ``reliability''? }

This section evaluates \tool{} along two dimensions:
\begin{enumerate}
    \item {\normalfont \bf Effectiveness:}
    Does schema refactoring eliminate serializability
  anomalies in real-world database applications? Is \tool{}  capable of
  repairing meaningful concurrency bugs?
  \item {\normalfont \bf Performance:}
  What impact does \tool{} have on the performance of
  refactored programs? How does \tool{} compare to other solutions to
  eliminating serializability anomalies, in particular by relying on
  stronger database-provided consistency guarantees?
\end{enumerate}

\subsection{Effectiveness}
To assess \tool' effectiveness, we applied it to a corpus of
standard benchmarks from the database community, including TPC-C,
SEATS and SmallBank~\cite{BR18,
  Difallah:OLTP_Bench, pldi15, GO16,Rahmani:2019:Clotho}. \autoref{tab:main} presents the results for each
benchmark. The first four columns display the number of transactions
(\#Txns), the number of tables in the original and refactored schemas
(\#Tables), and the number of anomalies detected assuming eventually
consistent guarantees for the original (EC) and refactored (AT)
programs.  For each benchmark, \tool{} was able to repair at least
half the anomalies, and in many cases substantially more, suggesting
that many serializability bugs can be directly repaired by our schema
refactoring technique.
% The total time needed by \tool\ to analyze the original program and
% construct the refactored program is also presented.

In order to compare our approach to other means of anomaly elimination
-- namely, by merely strengthening the consistency guarantees provided
by the underlying database -- we modified \tool's anomaly oracle to
only consider executions permitted under causal consistency and
repeatable read; the former enforces causal ordering in the visibility
relation, while the latter prevents results of a newly committed
transaction $T$ becoming visible to an executing transaction that has
already read state that is written by $T$.
% \footnote{This can be done by
% strengthening the encoding according to each guarantee's
% specification.  For example, CC executions are captured by the
% following axiom which enforces that any effect visible to $\eta'$
% should also be visible to effects that witness $\eta'$:
% $\forall \eta, \eta', \eta''.\; \mathsf{vis}(\eta,\eta') \wedge
% \mathsf{vis}(\eta',\eta'') \Rightarrow \mathsf{vis}(\eta,\eta'')$ }
The next two columns of \autoref{tab:main}, (CC) and (RR), show the
result of this analysis: causal consistency was only able to reduce
the number of anomalies in one benchmark (by 12\%) and repeatable read
in three (by 5\%, 15\% and 16\%). This suggests that only relying on
isolation guarantees between eventual and sequential consistency is
not likely to significantly reduce the number of concurrency bugs
that manifest in an EC execution. % The final column
% presents the total time needed by \tool{} to analyze and repair each
% program.

As a final measure of \tool's impact on correctness, we carried out a
more in-depth analysis of the SmallBank benchmark, in order to
understand \tool's ability to repair meaningful concurrency bugs. This
benchmark maintains the details of customers and their accounts, with
dedicated tables holding checking and savings entries for each
customer. By analyzing this and similar banking applications from the
literature~\cite{GO16,KJ18,WB17}, we identified three invariants to be
preserved by each transaction\footnote{Detailed descriptions of each
  invariant can be found in the supplementary materials of our submission}:
\begin{comment}
  \begin{enumerate*}[label=(\roman*)]
  \item The balance of both accounts must always be non-negative,
  \item Each account must accurately reflect the history of deposits
    to that account, and
  \item Each client must always witness a consistent state of her
    checking and savings accounts. For example, when transferring
    money between accounts, users should not see a state where the
    money is deducted from the checking account but not yet deposited
    into savings.
  \end{enumerate*}
\end{comment}
\noindent Interestingly, we were able to detect violations of \emph{all
three} invariants in the original program under EC, while the repaired
program violated \emph{only one}. This is evidence that the statically
identified serializability anomalies eliminated by \tool{} are
meaningful proxies to the application-level invariants that developers
care about.

\begin{comment}
\BD{This analysis is a bit of an orphan, as it does not address either
  of the main evaluation questions... We should discuss what to do
  with it.}  Lastly, we investigated the utility of using the results
of our oracle to drive \tool's repair procedure.  For these
experiments we removed the initial phase of analysis and instead
introduced tables and fields randomly.  The complete description of
our experiments and the results can be found in \autoref{app:eval}.
In summary, the vast majority (97\%) of random refactorings did not
eliminate \emph{any} of the anomalies. Even those experiments that
managed to repair some anomalies still resulted in a program with many
more bugs than that returned by \tool{}'s oracle-guided repair
strategy.
\end{comment}
\begin{comment}
\vspace{2mm}
\noindent\fbox{%
    \parbox{0.465\textwidth}{%
    {\bf Summary. }
       Our experiments show that \tool{} is capable of analyzing real-world benchmarks and eliminates an average of 73\% of anomalies through schema refactoring. It is also shown to be more effective at eliminating anomalies compared to well-known weak consistency guarantees.
    }%
}
\end{comment}

\setlength{\tabcolsep}{2pt}
\begin{table}[t] \centering
\begin{footnotesize}
\begin{tabular}{@{}lclccccc@{}}
\toprule
%\vspace{-0.5mm}
\textbf{Benchmark} & \textbf{\#Txns} & \thead{ \textbf{\#Tables} } &
                                                                     \textbf{EC} & \textbf{AT} & \textbf{CC} & \textbf{RR} & \thead{\textbf{Time\,(s)}} \vspace{-0.4mm}\\
\midrule
\textbf{TPC-C}
%& Warehouse/Order Management System
~\cite{TPCC,K19} & 5 & \hspace{2mm} 9, 16 & 33 & 8 & 33 & 33
%& \hspace{1mm} 76\,\%
& 81.2  \\
\textbf{SEATS}
%& Flight Search and Ticketing System
~\cite{Difallah:OLTP_Bench, seats}& 6 &\hspace{2mm} 8, 12 & 35 & 10 & 35 & 33
%& \hspace{1mm} 71\%
& 61.5  \\
\textbf{Courseware}
%& Online Education Framework
~\cite{GO16, KA18}& 5 & \hspace{2mm} 3, 2 & 5 & 0 & 5 & 5
%& \hspace{1mm} 100\%
&  12.7 \\
\textbf{SmallBank}
%& Electronic Banking System
~\cite{Difallah:OLTP_Bench, pldi15} & 6 & \hspace{2mm} 3, 5 & 24 & 8 & 21 & 20
%& \hspace{1mm} 66\%
&  68.7 \\
\textbf{Twitter}
%& Simple Micro-blogging Service
~\cite{Difallah:OLTP_Bench} & 5 &\hspace{2mm} 4, 5 & 6 & 1 & 6 & 5
%& \hspace{1mm} 83\%
&  3.6  \\
\textbf{FMKe}
%& Health Care and Patient Management
~\cite{fmke} & 7 & \hspace{2mm}  7, 9 & 6 & 2 & 6 & 6
%& \hspace{1mm} 66\%
&  33.6  \\
\textbf{SIBench}
%& Snapshot Isolation Microbenchmark
~\cite{Difallah:OLTP_Bench}& 2 &\hspace{2mm} 1, 2 & 1 & 0 & 1 & 1
%& \hspace{1mm} 100\%
&  0.3 \\
\textbf{Wikipedia}
%& Simple Online Encyclopedia
~\cite{Difallah:OLTP_Bench} & 5 & \hspace{2mm}  12, 13 & 2 & 1 & 2 & 2
%&\hspace{1mm} 50\%
&  9.0 \\
\textbf{Killrchat}
%& Online Chatroom Management
~\cite{killrchat, BR18} & 5 & \hspace{2mm}  3, 4 & 6 & 3 & 6 & 6
%& \hspace{1mm} 50\%
& 42.9
\\
\bottomrule
\end{tabular}
\end{footnotesize}
\vspace{3mm}
\caption{Statically identified anomalous access pairs in the original
  and refactored benchmark programs. Time\,(s) holds the total time
  to analyze and repair each benchmark.}
\label{tab:main}
\vspace{-5mm}
\end{table}

\definecolor{darkgreen}{rgb}{0.0, 0.5, 0.0}
\definecolor{darkyello}{rgb}{0.8, 0.73, 0.031}
\definecolor{cobalt}{rgb}{0.0, 0.28, 0.67}
%%%%%%%%%%%%%%%%%%%%%%%%%%%%%%%%%%%%%%%%%%%%%%%%%%%%%%%%%%%%%%%%%%%%%%%%%%%%%%%%%%%%%%%%

\begin {figure*}[t]
\begin{minipage}{0.95\textwidth}
\begin{footnotesize}
% VA
\begin{subfigure}[b]{0.33\textwidth}
\centering
  \begin{minipage}{\textwidth}
  \vspace{-5mm}
  \includegraphics[width=\textwidth]{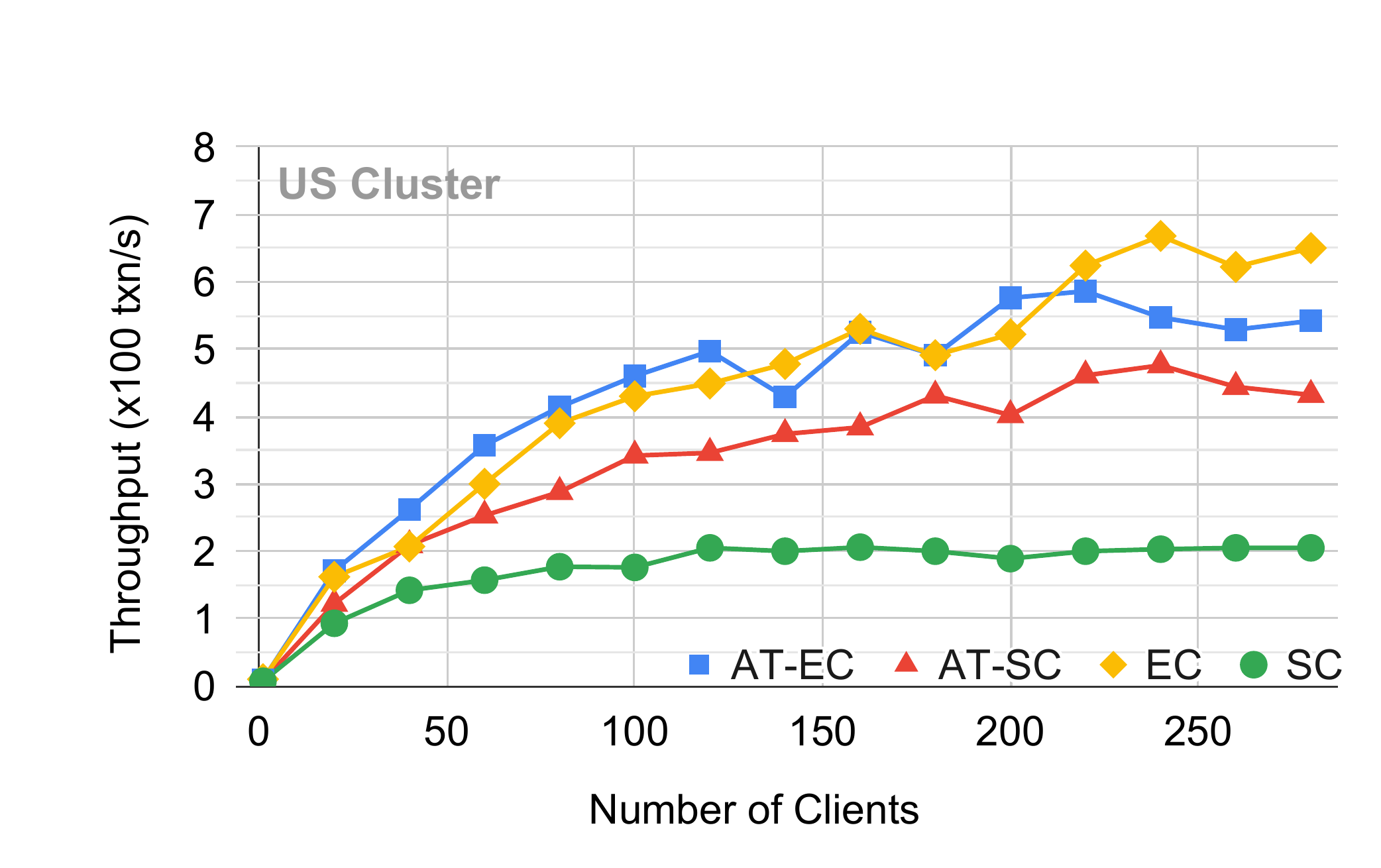}
  \end{minipage}
  \begin{minipage}{\textwidth}
    \vspace{-94mm}
  \includegraphics[width=\textwidth]{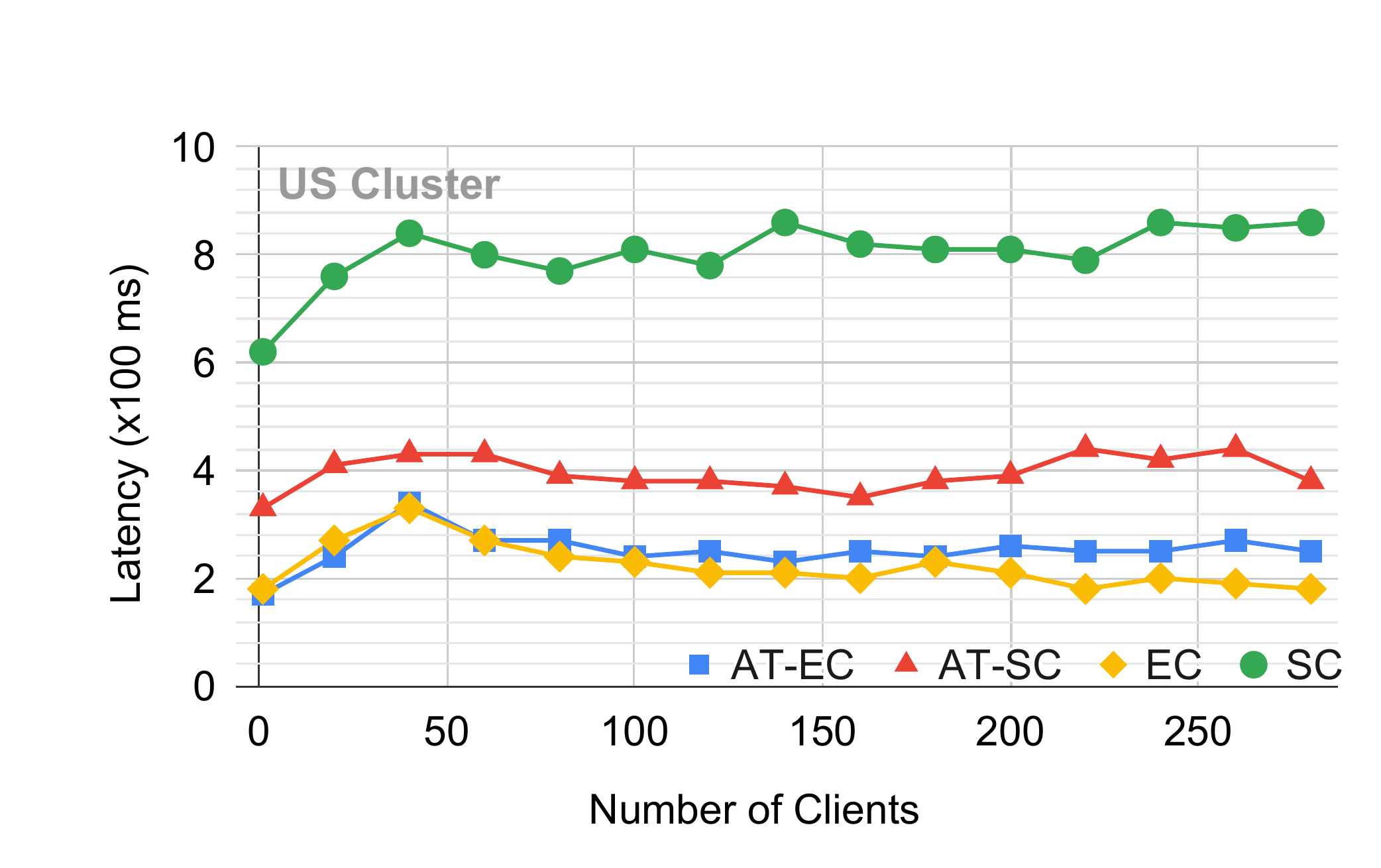}
  \end{minipage}
 \caption{SmallBank}
\end{subfigure}
%%
% US
\begin{subfigure}[b]{0.33\textwidth}
\begin{minipage}{\textwidth}
  \vspace{-5mm}
  \includegraphics[width=\textwidth]{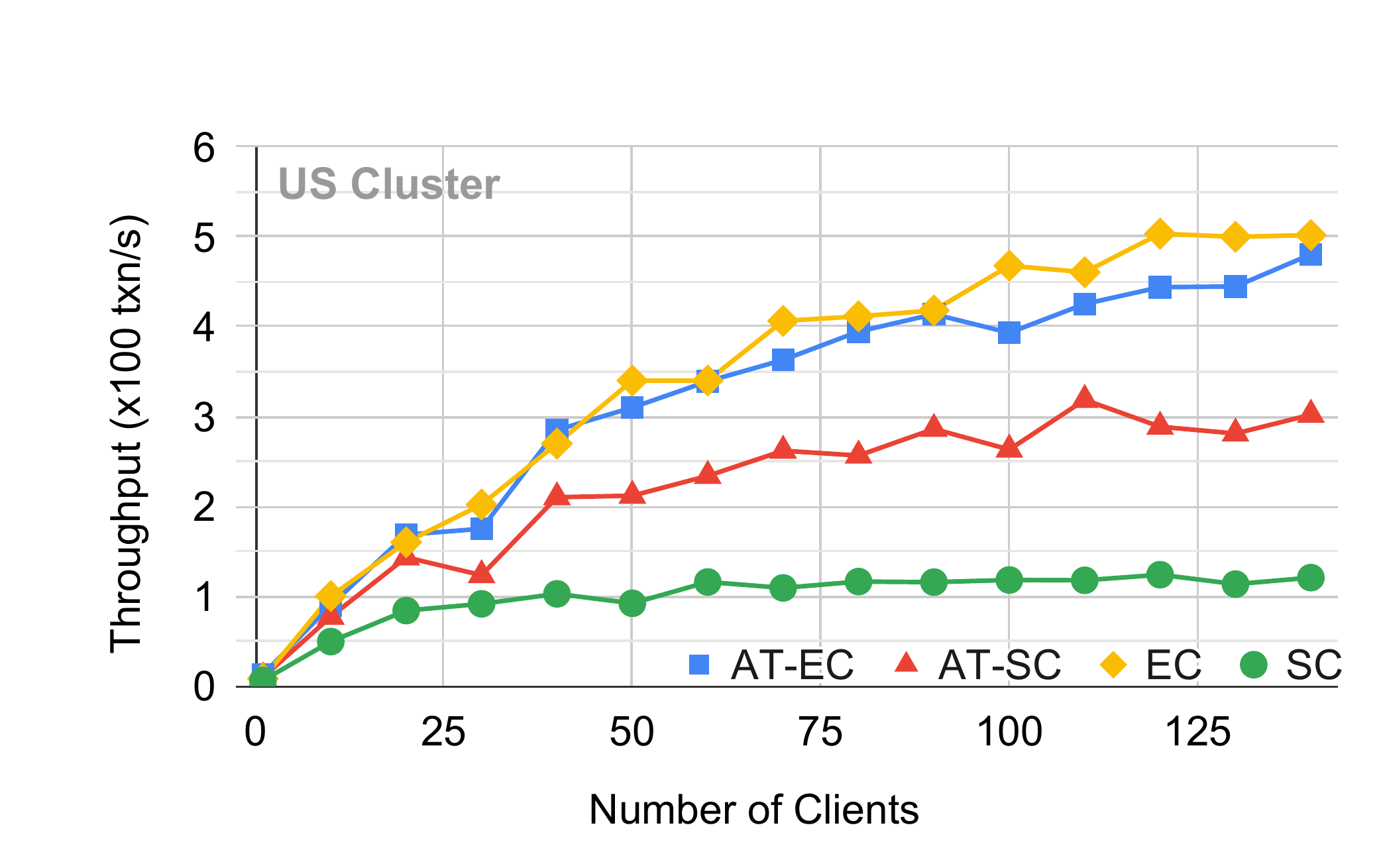}
  \end{minipage}
  \begin{minipage}{\textwidth}
    \vspace{-94mm}
  \includegraphics[width=\textwidth]{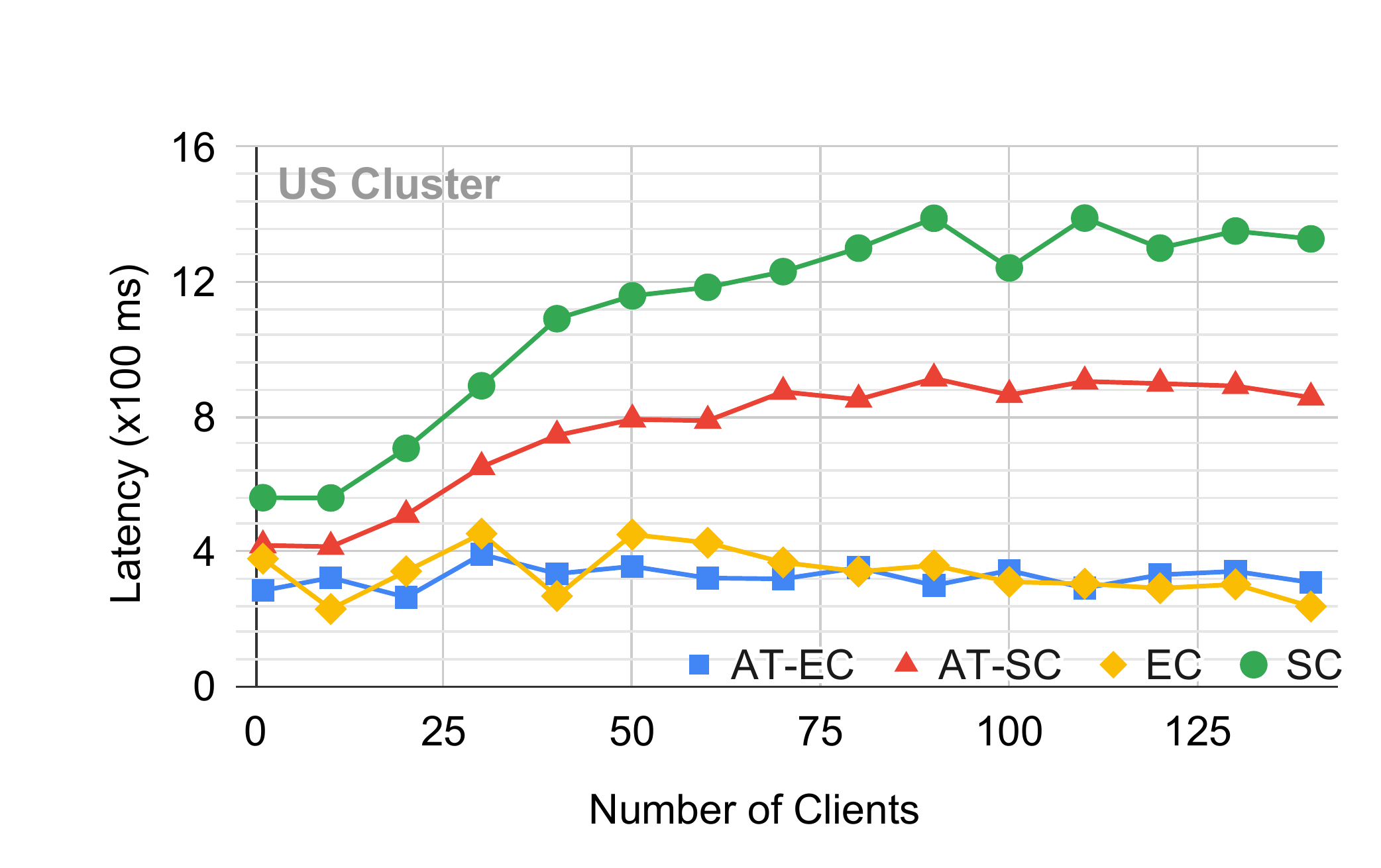}
  \end{minipage}
\caption{SEATS}
\end{subfigure}
%%
% GLOBAL
\begin{subfigure}[b]{0.33\textwidth}
\centering
\vspace{30mm}
\begin{minipage}{\textwidth}
  \vspace{-5mm}
  \includegraphics[width=\textwidth]{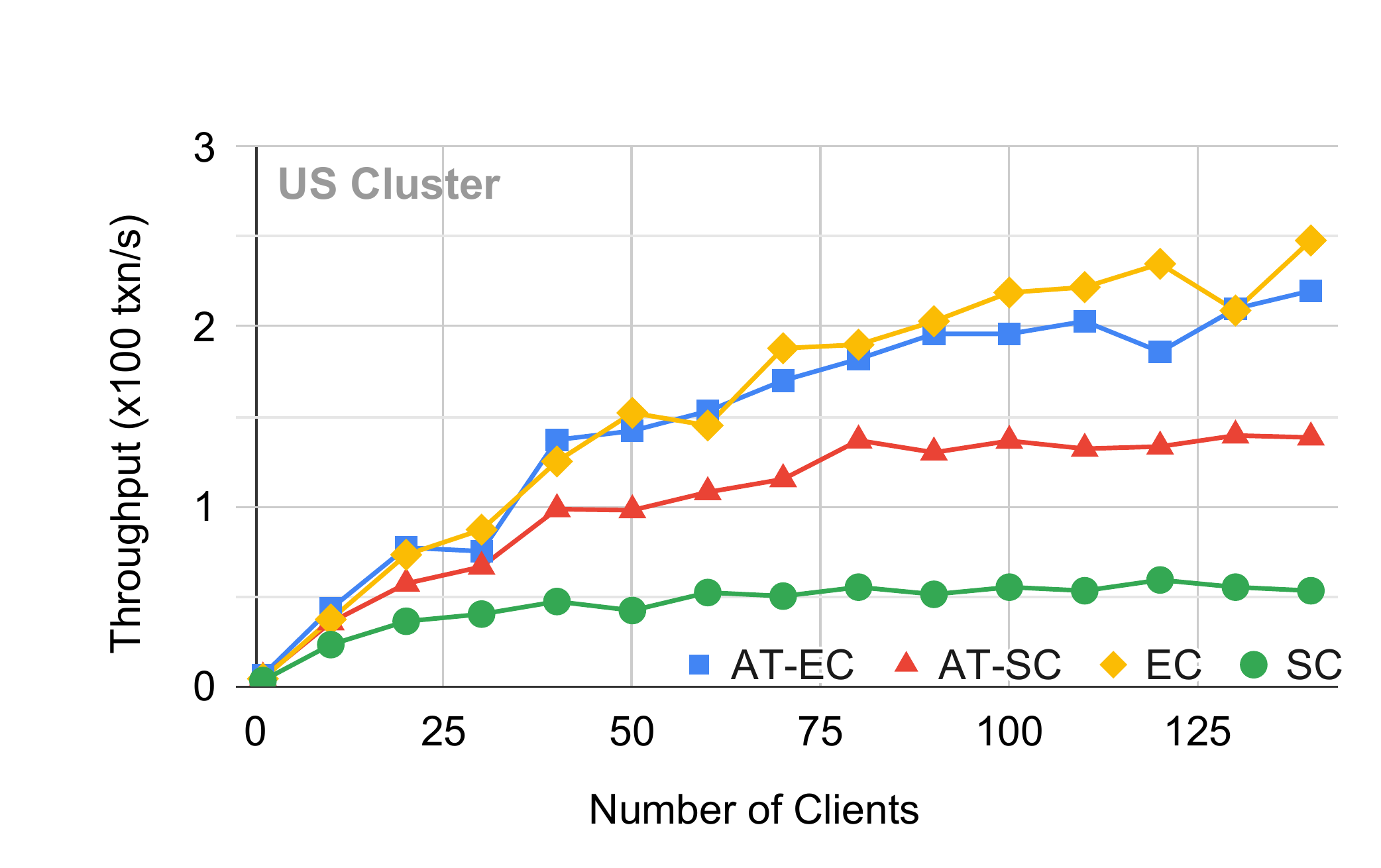}
   \end{minipage}
  \begin{minipage}{\textwidth}
    \vspace{-94mm}
  \includegraphics[width=\textwidth]{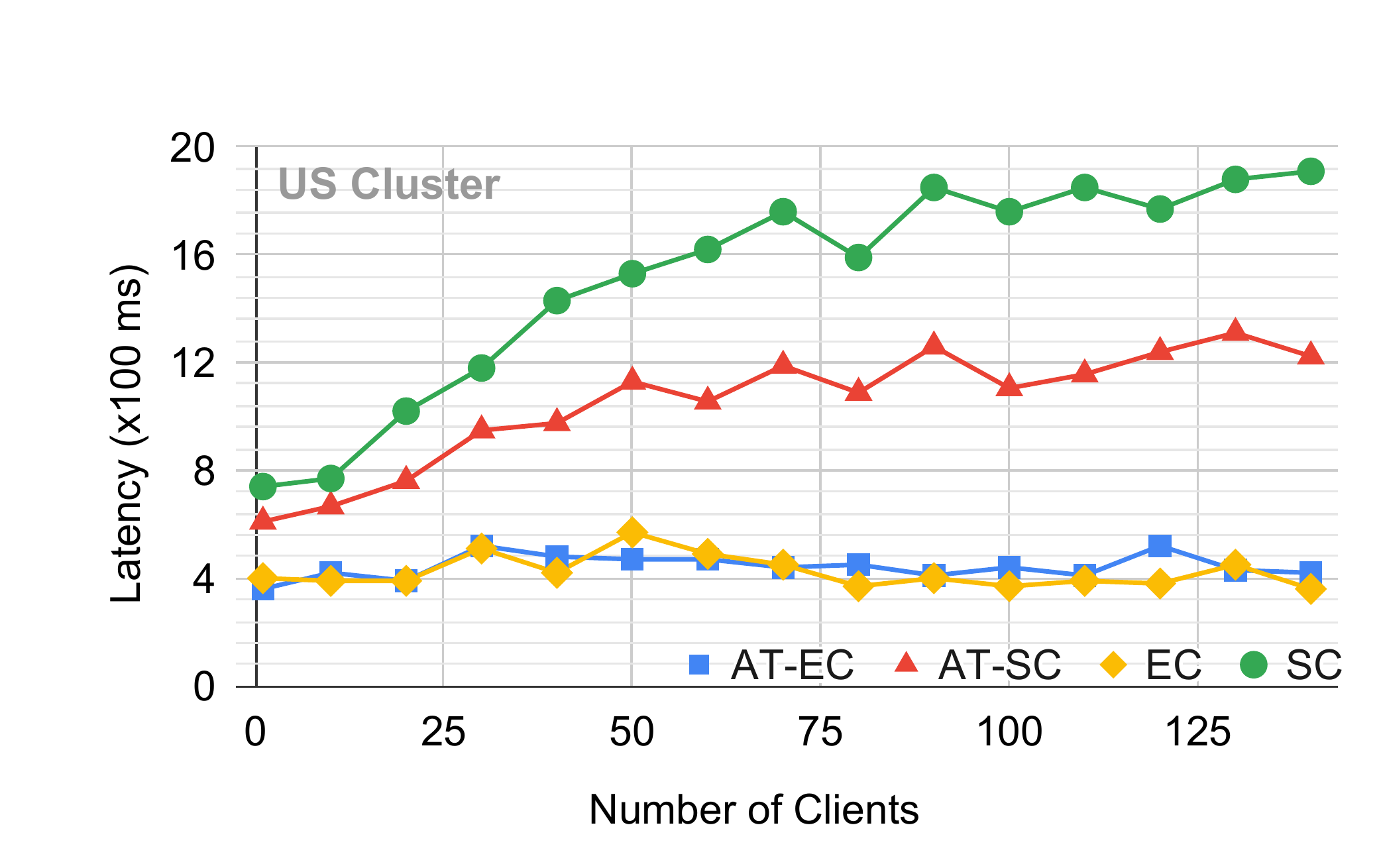}
    \end{minipage}
 \caption{TPC-C}
\end{subfigure}
\end{footnotesize}
\vspace{-5mm}
  \caption{Performance evaluation of SmallBank, SEATS and TPC-C benchmarks running on US cluster (see the supplementary materials for results of experiments on other two clusters). }
  \label{fig:perf}
  \end{minipage}
\end{figure*}

\subsection{Performance}
To evaluate the performance impact of schema refactoring, we conducted
further experiments on a real-world geo-replicated database cluster,
consisting of three AWS machines (M10 tier with 2 vCPUs and 2GB of
memory) located across US in N. Virginia, Ohio and Oregon. Similar
results were exhibited by experiments on a single data center and
globally distributed clusters; see supplementary materials
for more details.  Each node runs MongoDB (v.4.2.9), a modern document
database management system that supports a variety of data-model
design options and consistency enforcement levels. MongoDB documents
are equivalent to records and a collection of documents is equivalent
to a table instance, making all our techniques applicable to MongoDB
clients.

\autoref{fig:perf} presents the latency (top) and throughput (bottom)
of concurrent executions of SmallBank (left), SEATS (middle) and TPC-C
(right) benchmarks.  These benchmarks are representative of the kind
of OLTP applications best suited for our refactoring approach.
Horizontal axes show the number of clients, where each client
repeatedly submits transactions to the database according to each
benchmark's specification. Each experiment was run for 90 seconds and
the average performance results are presented. For each benchmark,
performance of four different versions of the program are compared:
\begin{enumerate*}[label=(\roman*)]
\item original version running under EC
  (\textcolor{darkyello}{$\blacklozenge$ EC}),
\item refactored version running under EC
  (\textcolor{cobalt}{$\blacksquare$ AT-EC}),
\item original version running under SC
  (\textcolor{darkgreen}{\CIRCLE\ SC}) and
\item refactored version where transactions with at least one anomaly
  are run under SC and the rest are run under EC
  (\textcolor{red}{$\blacktriangle$ AT-SC})
\end{enumerate*}.
Across all benchmarks, SC results in poor performance compared to EC,
due to lower concurrency and additional synchronization required
between the database nodes.  On the other hand, AT-EC programs show
negligible overhead with respect to their EC counterparts, despite
having fewer anomalies. Most interestingly, refactored programs show
an average of 120\% higher throughput and 45\% lower latency compared
to their counterparts under SC, while offering the \emph{same level of
  safety}. These results provide evidence that automated schema
refactoring can play an important role in improving both the
correctness and performance of modern database programs.

\begin{comment}
\vspace{2mm}
\noindent\fbox{%
    \parbox{0.465\textwidth}{%
    {\bf Summary. }
       Our experiments show that refactored programs incur negligible overhead while eliminating many concurrency bugs. Refactored programs are also shown to be faster (up to 100\%) when all problematic data-accesses are eliminated by enforcing strong consistency.
    }%
}
\end{comment}

\section{Related Work}
\label{sec:related}

\citet{Wang:2019:Synthesizing} describe a synthesis
procedure for generating programs consistent with a database
refactoring,
%Given a source database program $P$, a source schema
%$S$, and a target schema $S'$, their approach synthesizes a new
%program $P'$ that operates over $S'$ but which is semantically
%equivalent to $P$
as determined by a verification procedure that
establishes database program equivalence~\cite{Wang:2017:Verifying}.
Their synthesis procedure performs enumerative search over a template
whose structure is derived by value correspondences extracted by
reasoning over the structure of the original and refactored schemas.  Our approach
%, while superficially related,
has several important differences.  First, our
search for a target program is driven by anomalous access pairs that
identify serializability anomalies in the original program and does
not involve enumerative search over the space of \emph{all} equivalent
candidate programs.  This important distinction eliminates the need
for generating arbitrarily-complex templates or sketches.  Second,
because we \emph{simultaneously} search for a target schema and
program consistent with that schema given these access pairs, our
technique does not need to employ conflict-driven
learning~\cite{FM+18} or related mechanisms to guide a general
synthesis procedure as it recovers from a failed synthesis attempt.
Instead, value correspondences derived from anomalous access pairs
help define a restricted class of schema refactorings (e.g.,
aggregation and logging) that directly informs the structure of the
target program.  

Identifying serializability anomalies in database systems is a
well-studied topic that continues to garner
attention~\cite{BE95,BE87,LBL04,FE05b,JO07}, although the issue of
automated repair is comparatively less explored.  A common approach in
all these techniques is to model interactions among concurrently
executing database transactions as a graph, with edges connecting
transactions that have a data dependency with one another; cycles in
the graph indicate a possible serializability violation.  Both
dynamic~\cite{WB17,BDM+17} and
static~\cite{BR18,KR18,Rahmani:2019:Clotho} techniques have been
developed to discover these violations in various domains and
settings.

%Techniques have been
%developed to discover these violations dynamically.  For example,
%\citet{WB17} use program traces to identify potential vulnerabilitis in
%Web applications that exploit weak isolation while
%\citet{BDM+17} present a dynamic analysis
%technique for discovering serializability in an eventually consistent
%distributed setting.  Follow-on work~\cite{BR18} develops scalable
%static methods under stronger causally-consistent
%assumptions.  \citet{Rahmani:2019:Clotho} present a test generation
%tool for triggering serializability anomalies that builds upon a
%static detection framework described in~\cite{KR18}.

%The idea of exploring alterations to data explored here, rather than
%changes to control, is reminiscent of data-centric synchronization
%proposed by \citet{DH+12} that builds atomic sets with associated units
%of work.  The context of their investigation, concurrent Java
%programs, is quite different from ours; in particular, their solution
%does not consider sound schema refactorings, an integral part of our
%approach.
An alternative approach to eliminating serializability anomalies is to
develop correct-by-construction methods.  For example, to safely
develop applications for eventually-consistent distributed
environments, conflict-free replicated data-types (CRDTs)~\cite{crdt} 
have been proposed.  CRDTs are abstract data-types (e.g. sets,
counters) equipped with commutative operations whose semantics are
invariant with respect to the order in which operations are applied on
their state.  Alternatively, there have been recent efforts which
explore enriching specifications, rather than applications, with
mechanisms that characterize notions of correctness in the presence of
replication~\cite{HL19,pldi15}, using these specifications to guide
safe implementations.  These techniques, however, have not been
applied to reasoning on the correctness of concurrent relational
database programs which have highly-specialized structure and
semantics, centered on table-based operations over inter-related
schema definitions, rather than control- and data-flow operations over
a program heap.

\section{Conclusions}
\label{sec:conc}

This paper presents \tool, a database refactoring tool
intended to repair 
serializability violations. 
We have formalized the refactoring
procedure
%articulated how refactorings can be guided by a repair
%mechanism, 
and demonstrated experimental results that 
%justify the
%utility of our approach.  
%Our results show that
schema refactoring is a viable strategy for concurrency bug repair in
modern database applications.

%
%
%%%%%%%%%%%%%%%%%%%%%%%%%%%%%%%%%%%%%%%%%%%%%%%%%%%%%%

%% %% Acknowledgments
%% \begin{acks}                            %% acks environment is optional
%%                                         %% contents suppressed with 'anonymous'
%%   %% Commands \grantsponsor{<sponsorID>}{<name>}{<url>} and
%%   %% \grantnum[<url>]{<sponsorID>}{<number>} should be used to
%%   %% acknowledge financial support and will be used by metadata
%%   %% extraction tools.
%%   This material is based upon work supported by the
%%   \grantsponsor{GS100000001}{National Science
%%     Foundation}{http://dx.doi.org/10.13039/100000001} under Grant
%%   No.~\grantnum{GS100000001}{nnnnnnn} and Grant
%%   No.~\grantnum{GS100000001}{mmmmmmm}.  Any opinions, findings, and
%%   conclusions or recommendations expressed in this material are those
%%   of the author and do not necessarily reflect the views of the
%%   National Science Foundation.
%% \end{acks}

%% Bibliography
\bibliography{ref}

%% Appendix
%\newpage
\appendix
\section{Complete Evaluation Results}
\label{app:eval}
In this section we present additional evaluation results omitted from the paper.

\subsection{Performance on Global and Local Database Clusters}
Figures \ref{fig:perf_smallbank}, \ref{fig:perf_seats}
and \ref{fig:perf_tpcc} present our experimental results from running SmallBank, SEATS and TPC-C benchmark on three different distributed database clusters. Results for US cluster (middle) are also presented in the paper. 
The VA cluster (left) consists of three AWS nodes all located in the same data center in N. Virginia. The global cluster (right) consists of three nodes located in N. Virgina, London and Tokyo. As we mentioned in the paper, 
M10 tier machines with 2 vCPUs and 2GB 
memory are used and each node runs MongoDB (v.4.2.9).

\begin {figure*}[h]  
% VA
\begin{subfigure}[b]{0.32\textwidth}  
\centering
  \includegraphics[width=\textwidth]{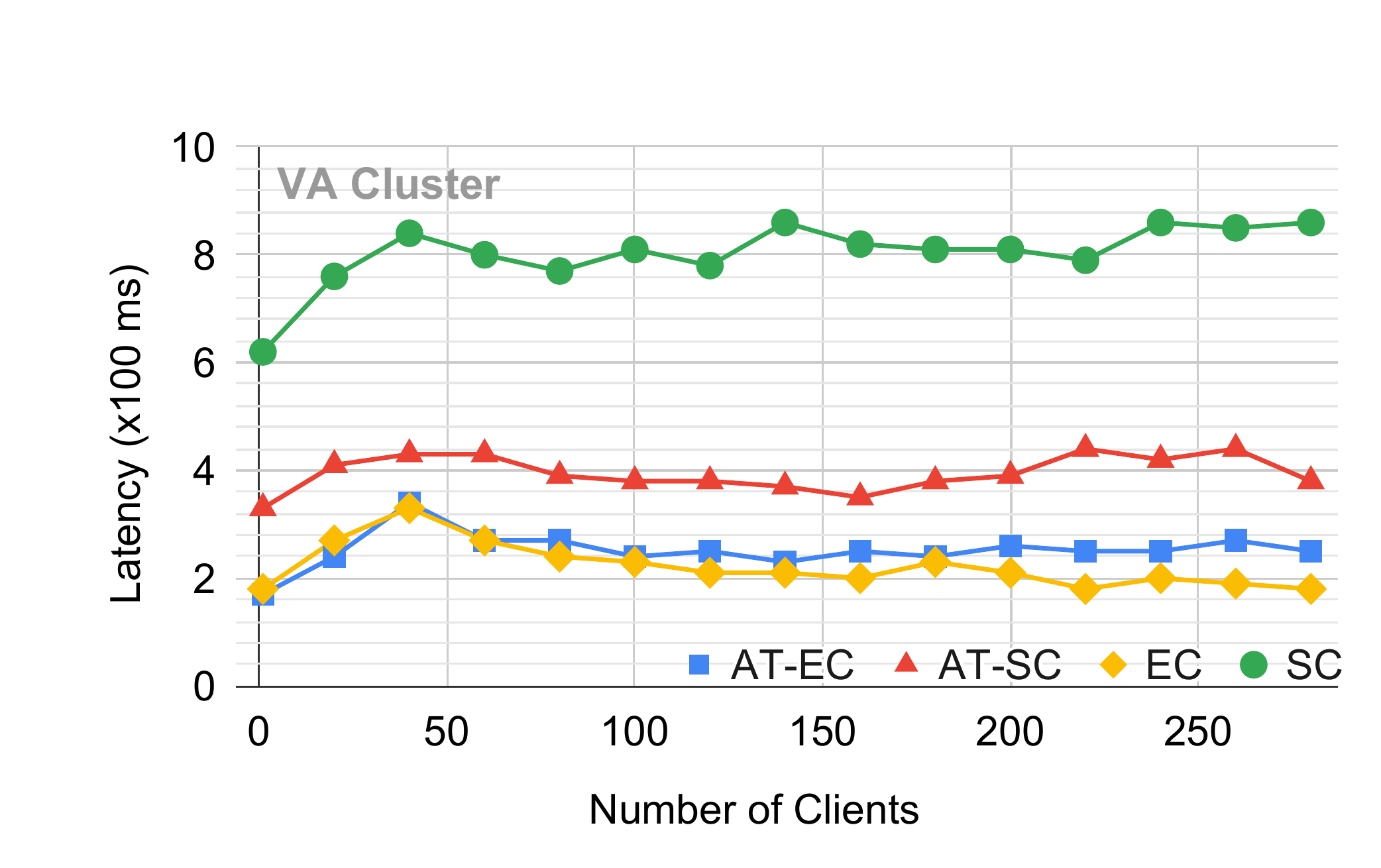}
\end{subfigure}
%%
% US
\begin{subfigure}[b]{0.32\textwidth}  
\includegraphics[width=\textwidth]{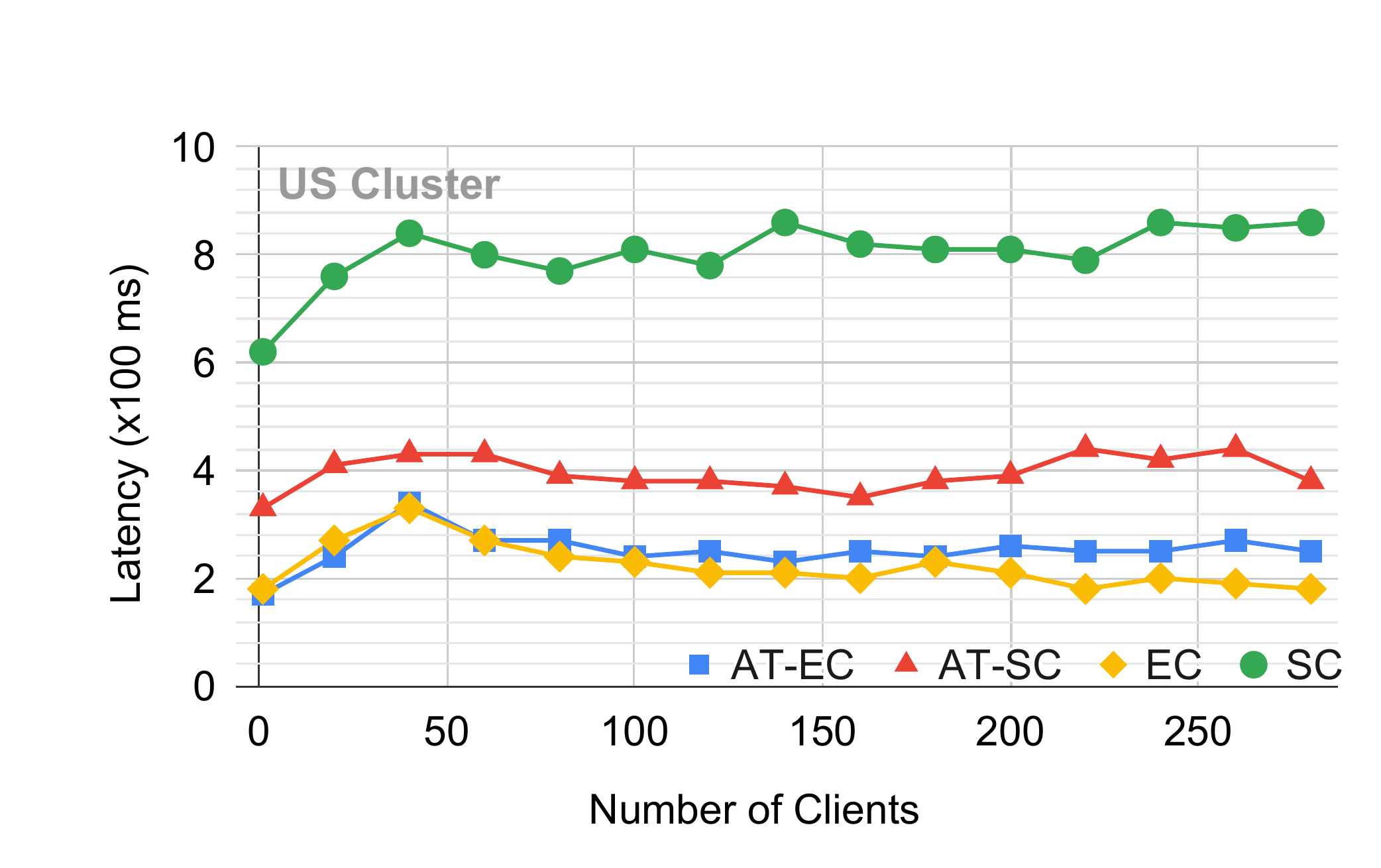}
\end{subfigure}
%%
% GLOBAL
\begin{subfigure}[b]{0.32\textwidth}  
\centering
  \includegraphics[width=\textwidth]{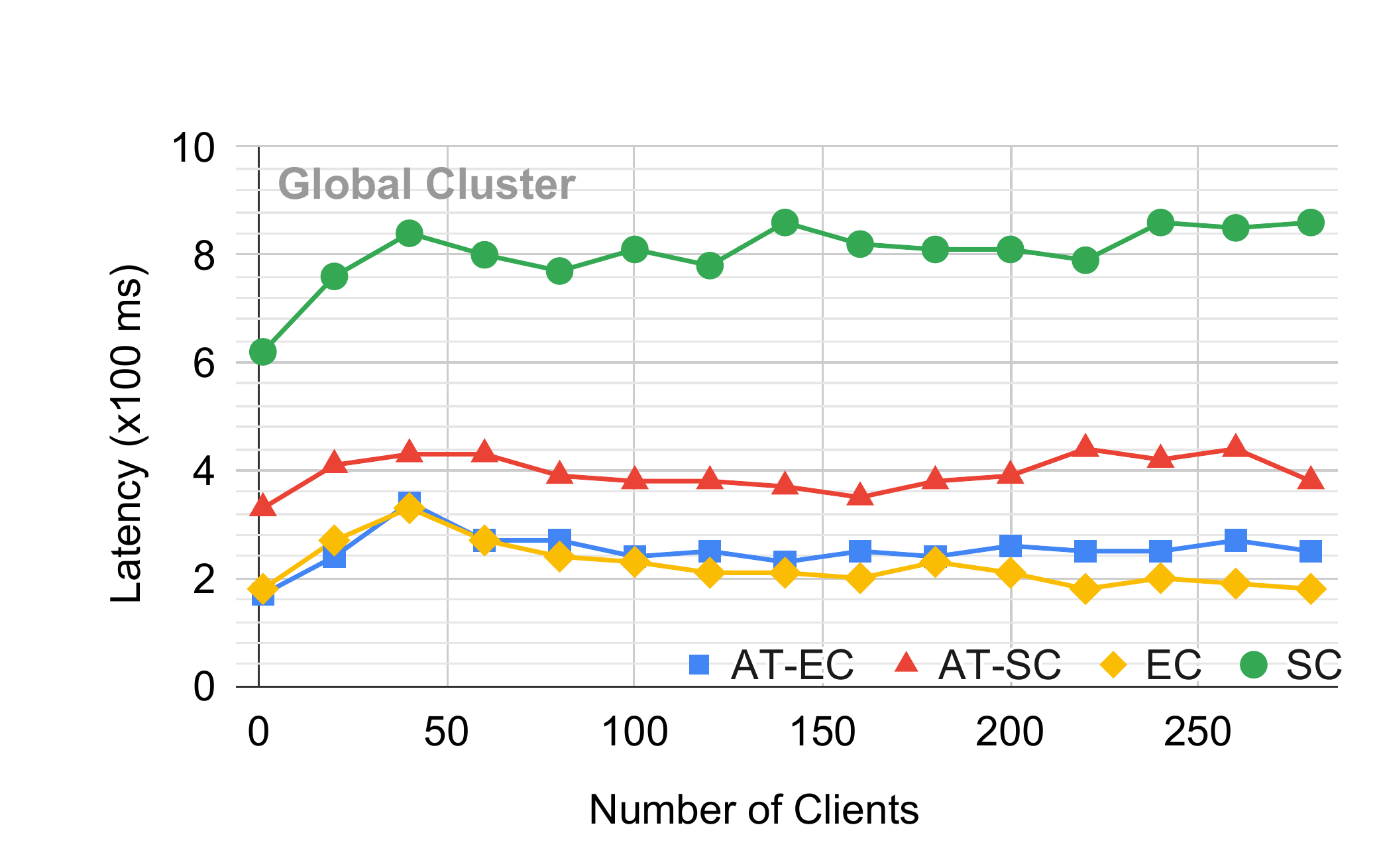}
\end{subfigure}
% VA
\begin{subfigure}[b]{0.32\textwidth}  
\centering
  \includegraphics[width=\textwidth]{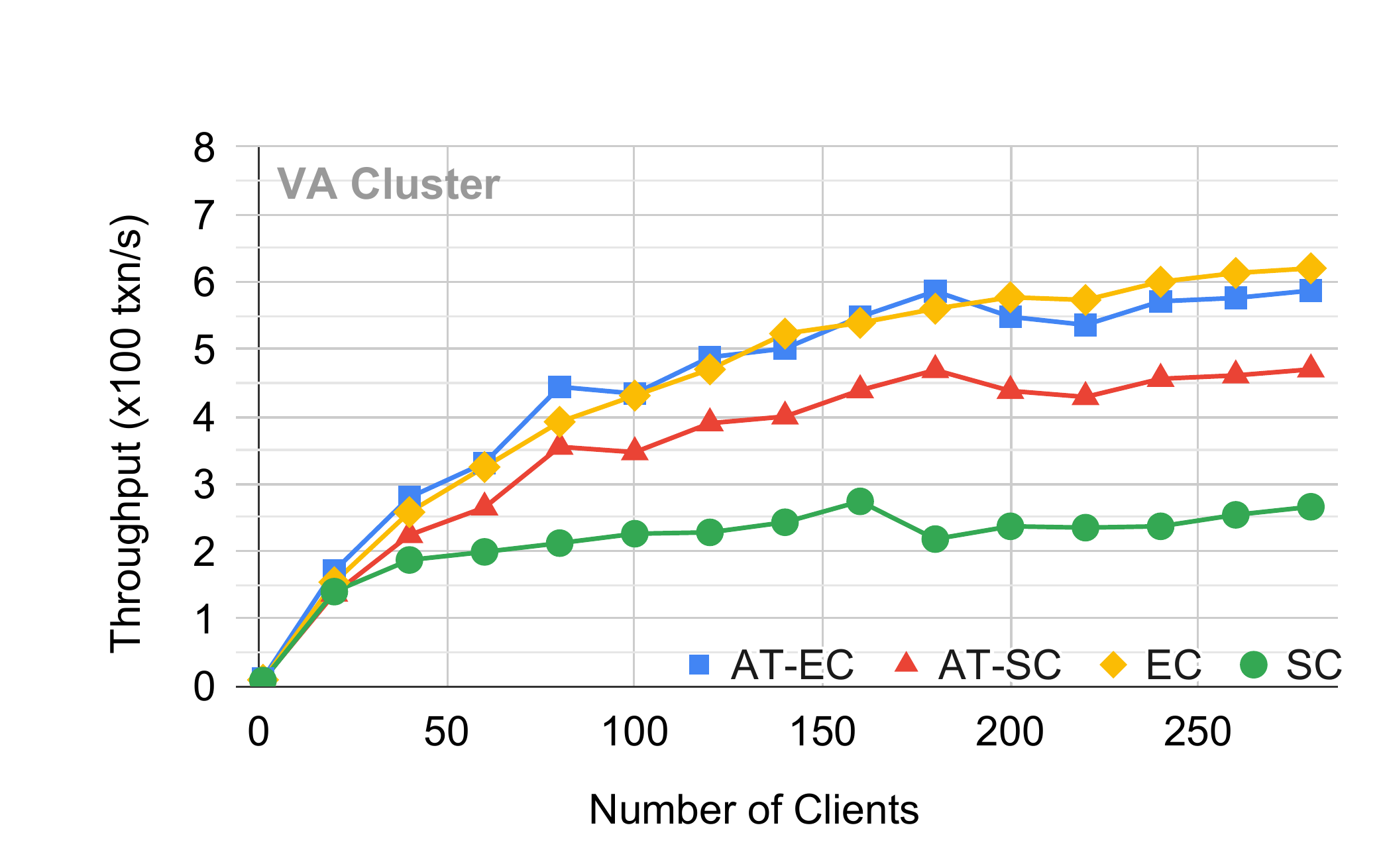}
\end{subfigure}
%%
% US
\begin{subfigure}[b]{0.32\textwidth}  
  \includegraphics[width=\textwidth]{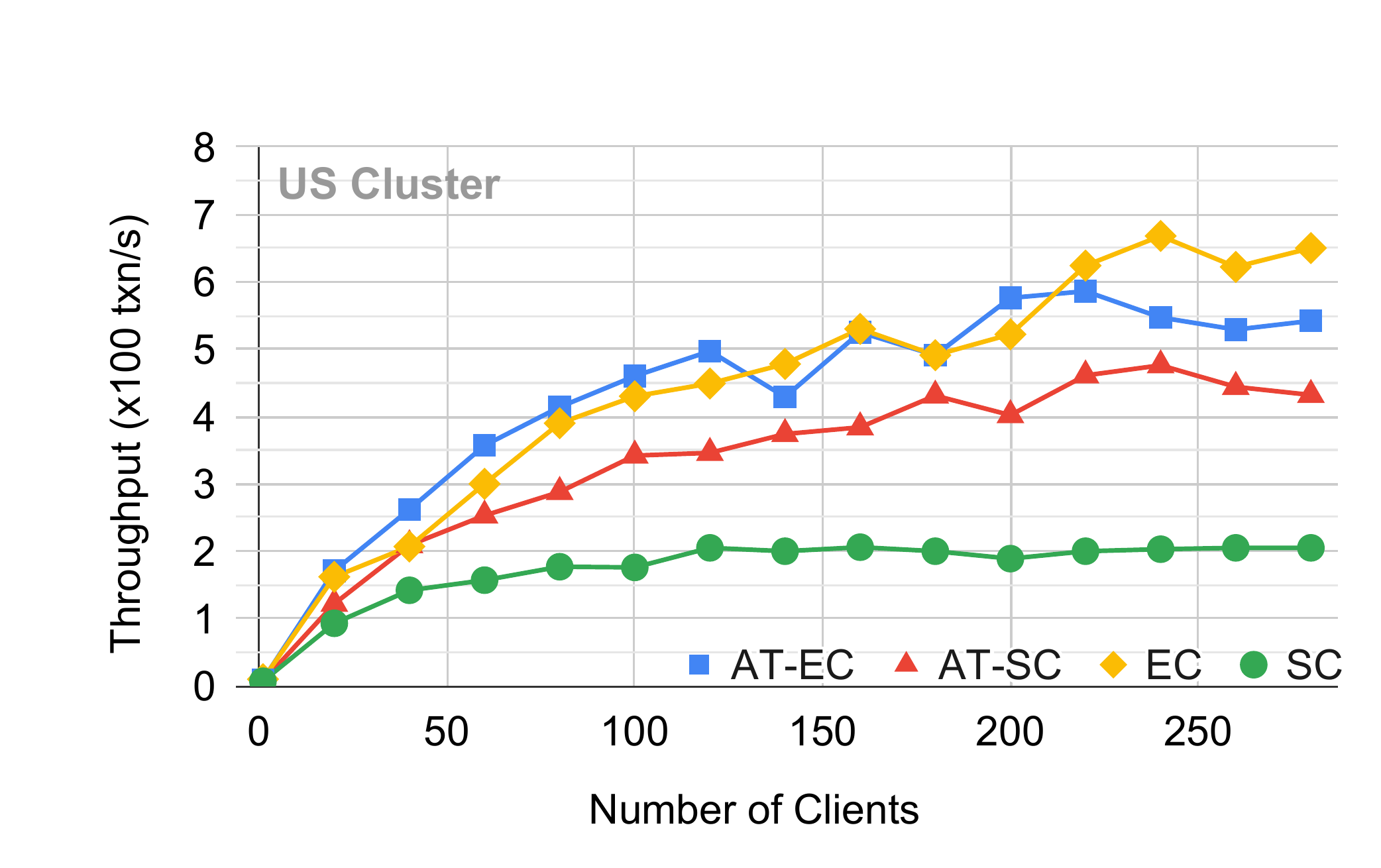}
\end{subfigure}
%%
% GLOBAL
\begin{subfigure}[b]{0.32\textwidth}  
\centering
  \includegraphics[width=\textwidth]{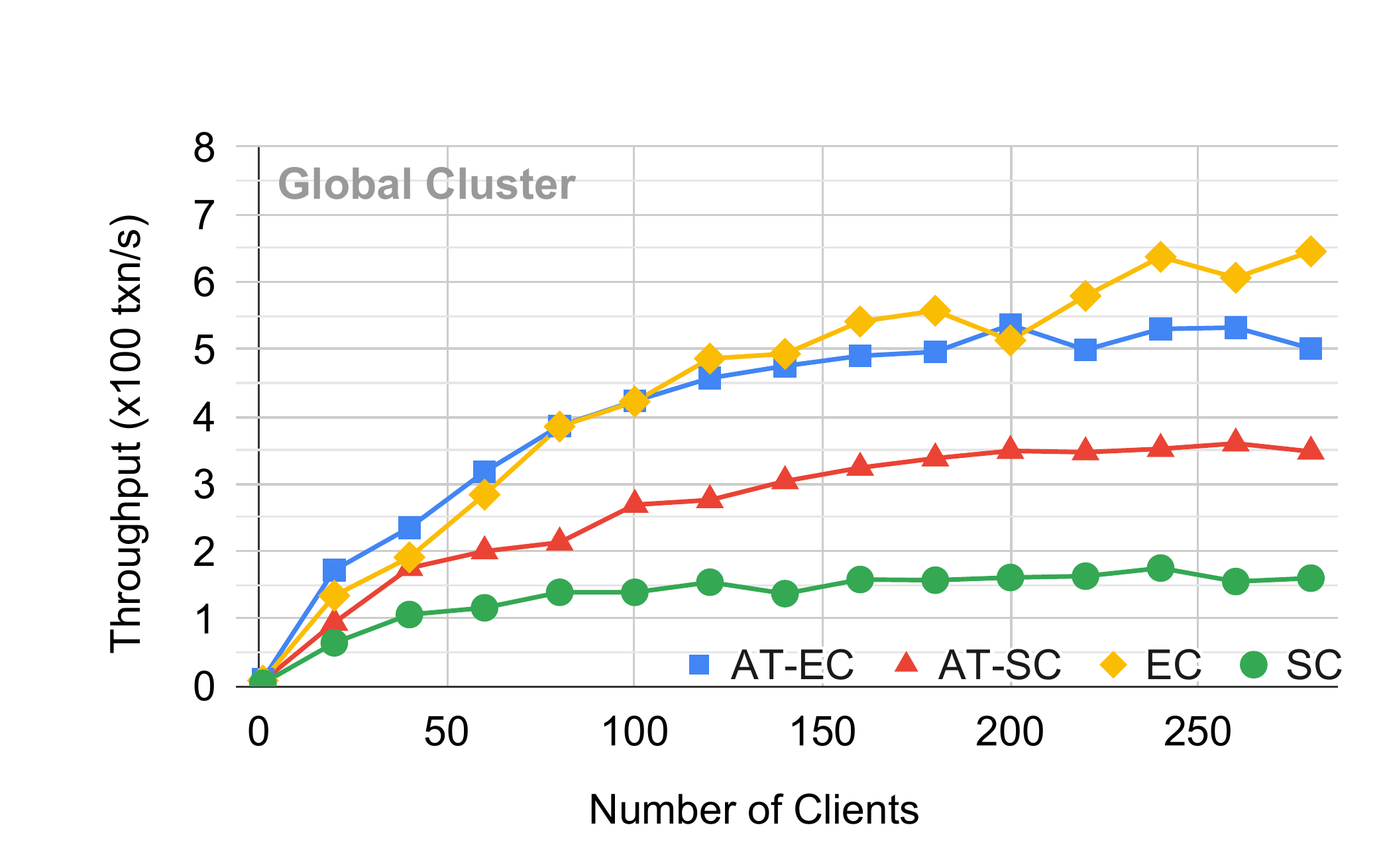}
\end{subfigure}
  \caption{Performance evaluation of SmallBank benchmark}
  \label{fig:perf_smallbank}
\end{figure*}

\begin {figure*}[h]  
% VA
\begin{subfigure}[b]{0.32\textwidth}  
\centering
  \includegraphics[width=\textwidth]{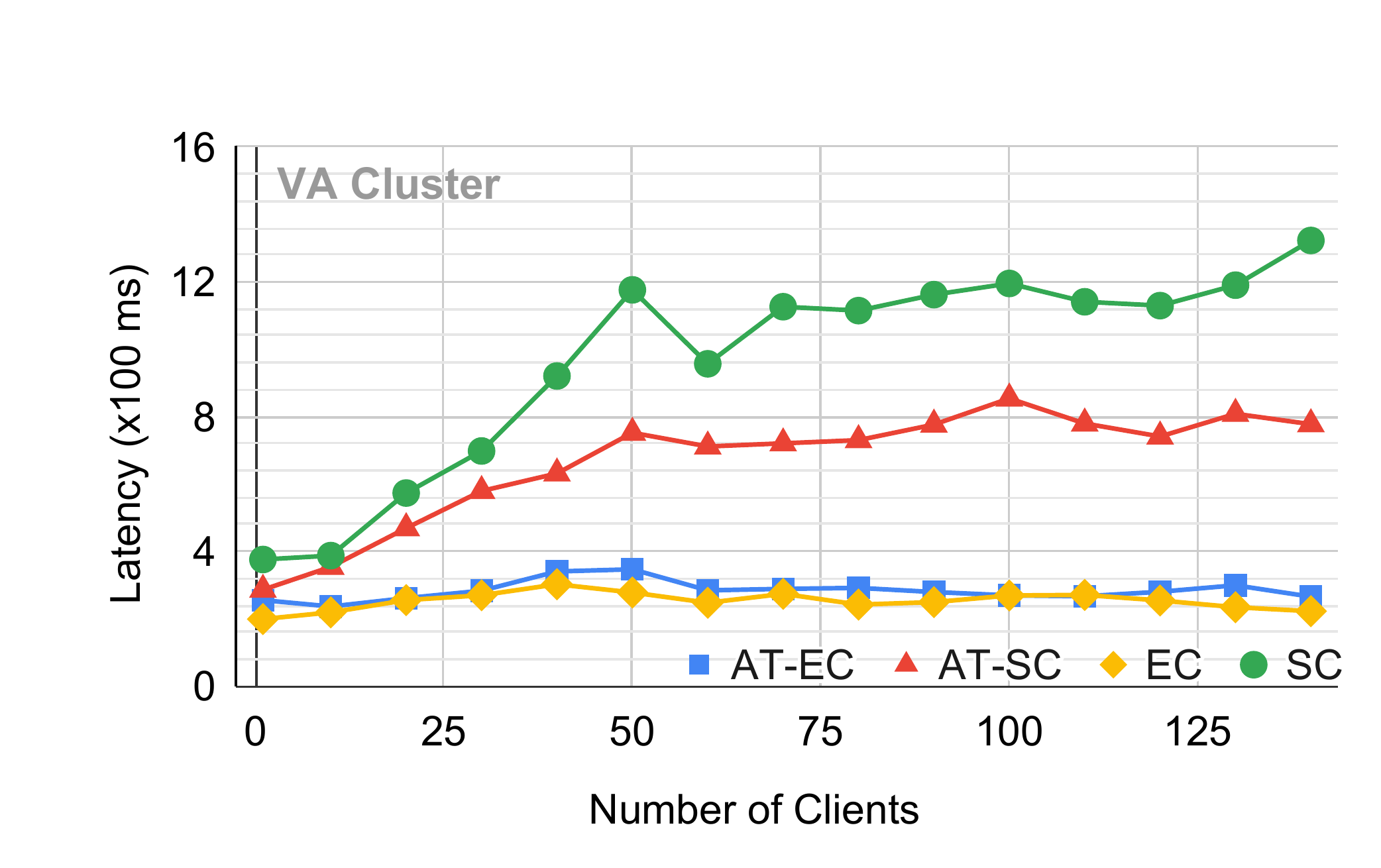}
\end{subfigure}
%%
% US
\begin{subfigure}[b]{0.32\textwidth}  
\includegraphics[width=\textwidth]{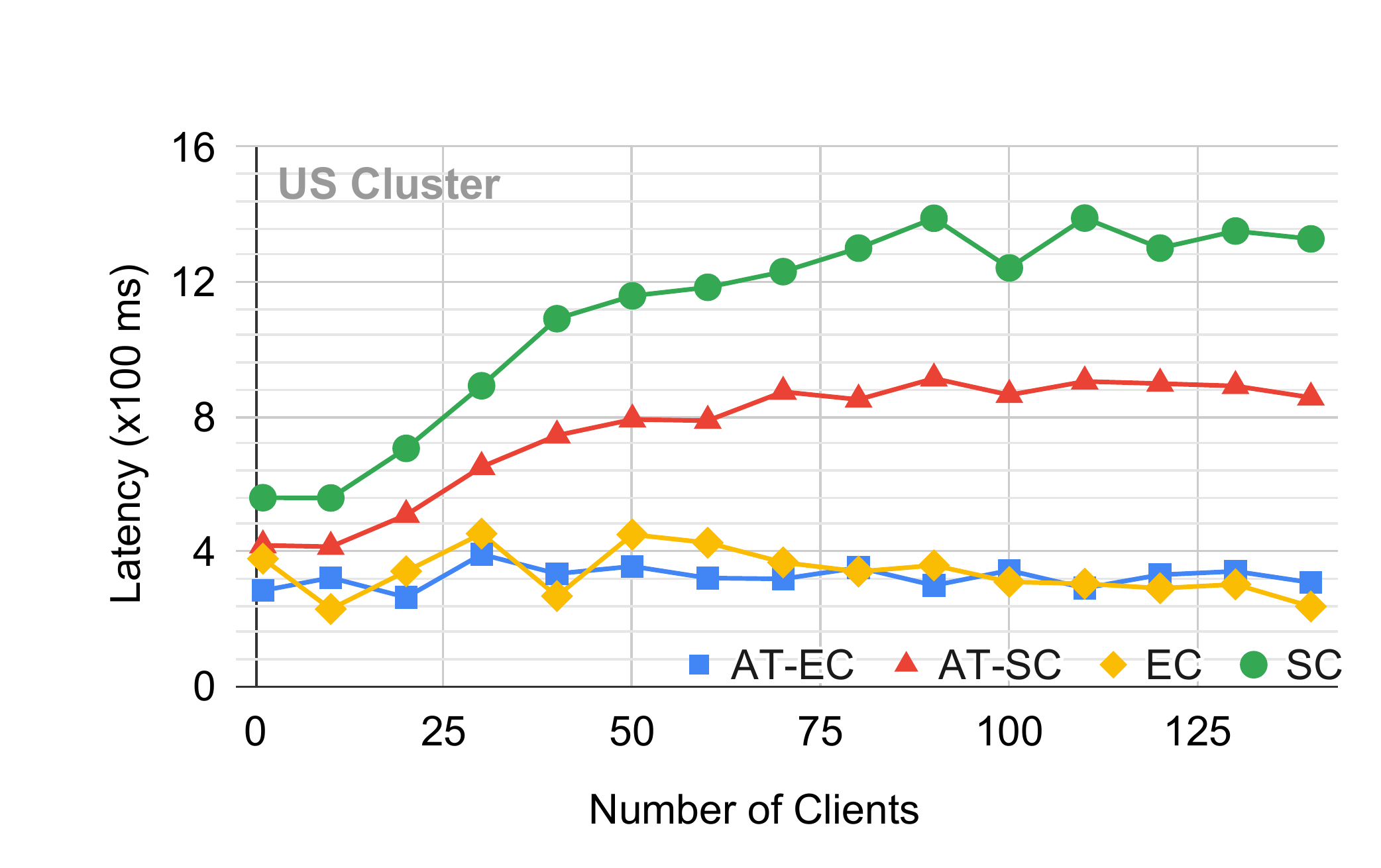}
\end{subfigure}
%%
% GLOBAL
\begin{subfigure}[b]{0.32\textwidth}  
\centering
  \includegraphics[width=\textwidth]{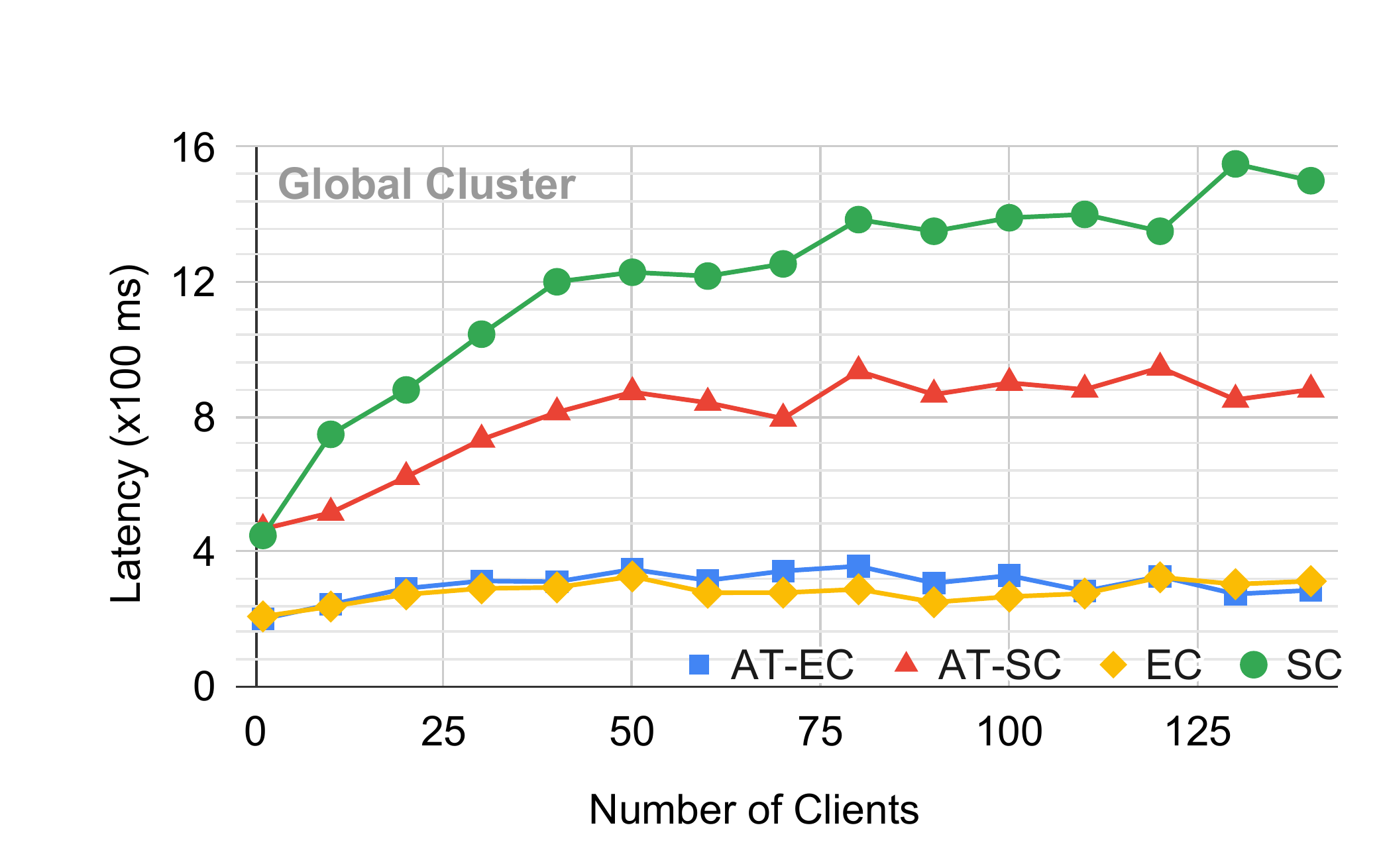}
\end{subfigure}
% VA
\begin{subfigure}[b]{0.32\textwidth}  
\centering
  \includegraphics[width=\textwidth]{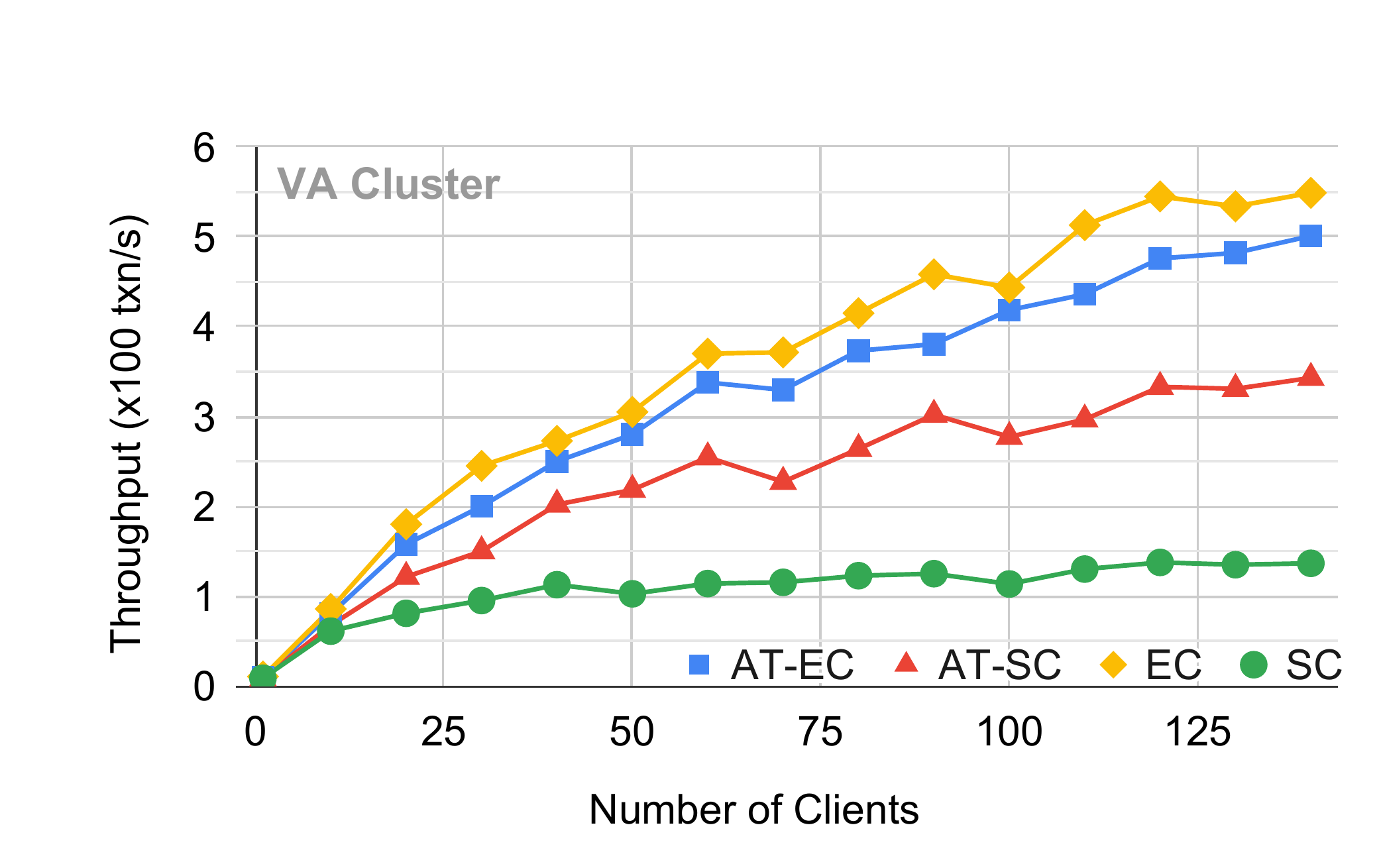}
\end{subfigure}
%%
% US
\begin{subfigure}[b]{0.32\textwidth}  
  \includegraphics[width=\textwidth]{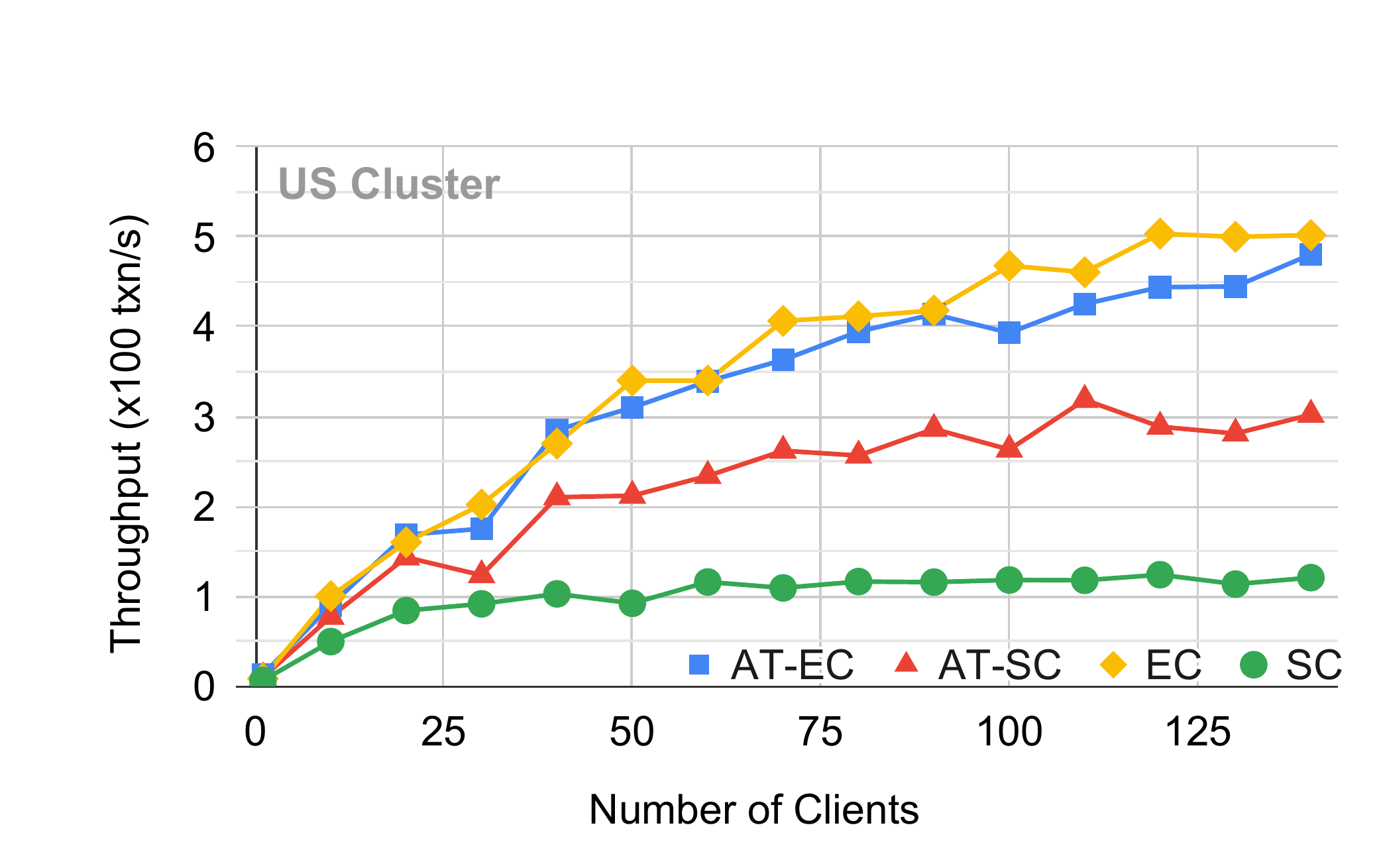} 
\end{subfigure}
%%
% GLOBAL
\begin{subfigure}[b]{0.32\textwidth}  
\centering
  \includegraphics[width=\textwidth]{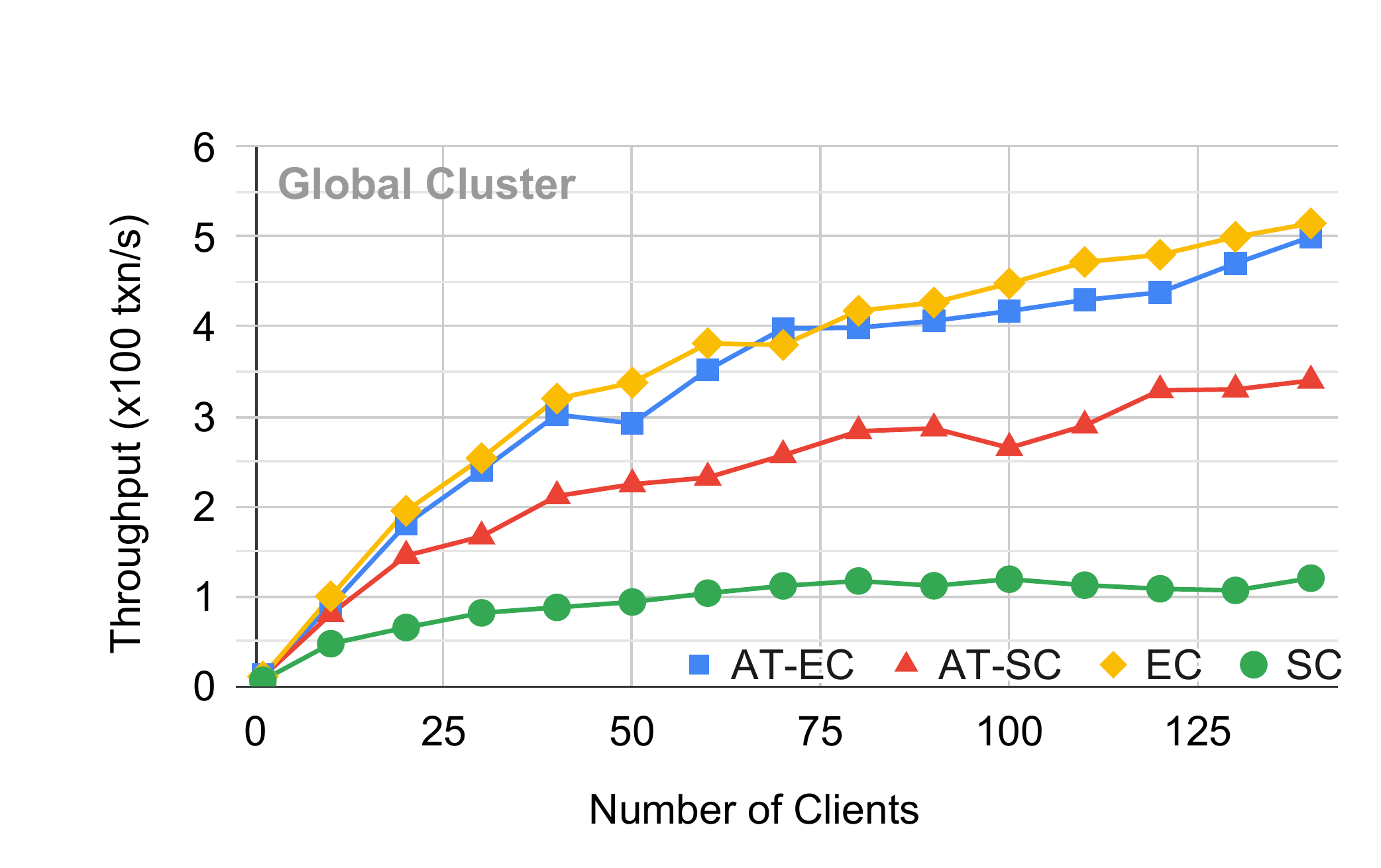}
\end{subfigure}
  \caption{Performance evaluation of SEATS benchmark}
  \label{fig:perf_seats}
\end{figure*}

\begin {figure*}[h]  
% VA
\begin{subfigure}[b]{0.32\textwidth}  
\centering
  \includegraphics[width=\textwidth]{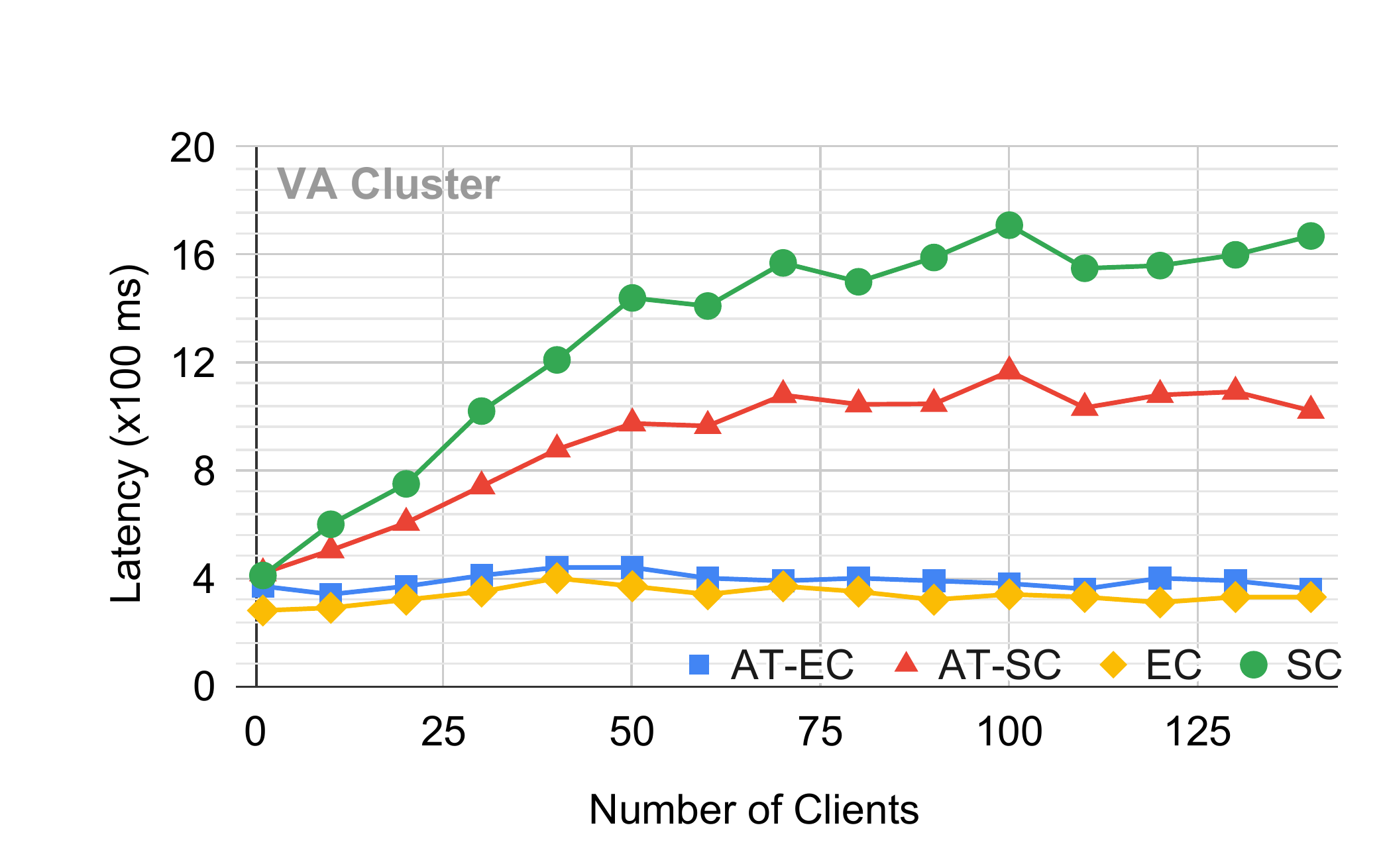}
\end{subfigure}
%%
% US
\begin{subfigure}[b]{0.32\textwidth}  
\includegraphics[width=\textwidth]{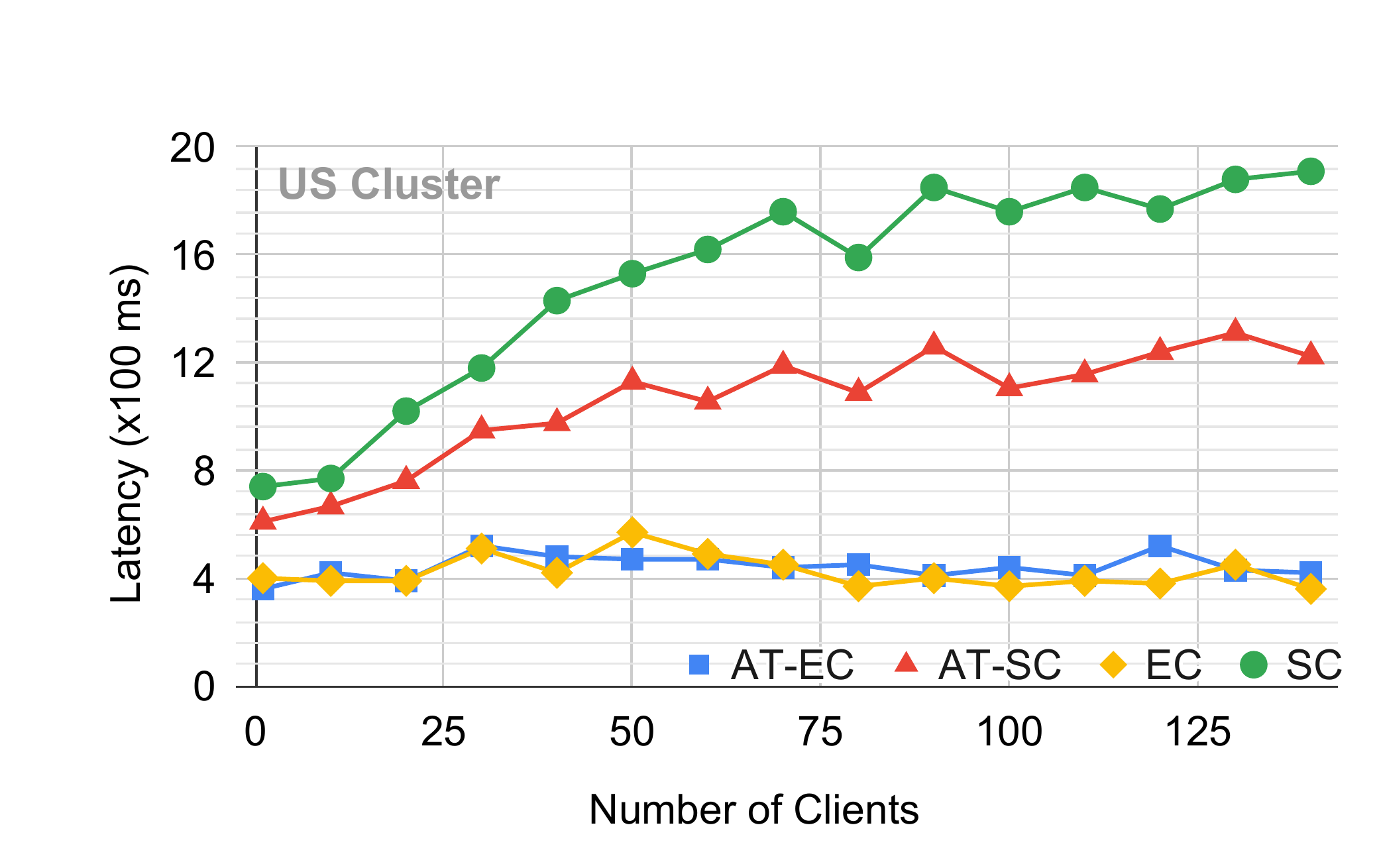}
\end{subfigure}
%%
% GLOBAL
\begin{subfigure}[b]{0.32\textwidth}  
\centering
  \includegraphics[width=\textwidth]{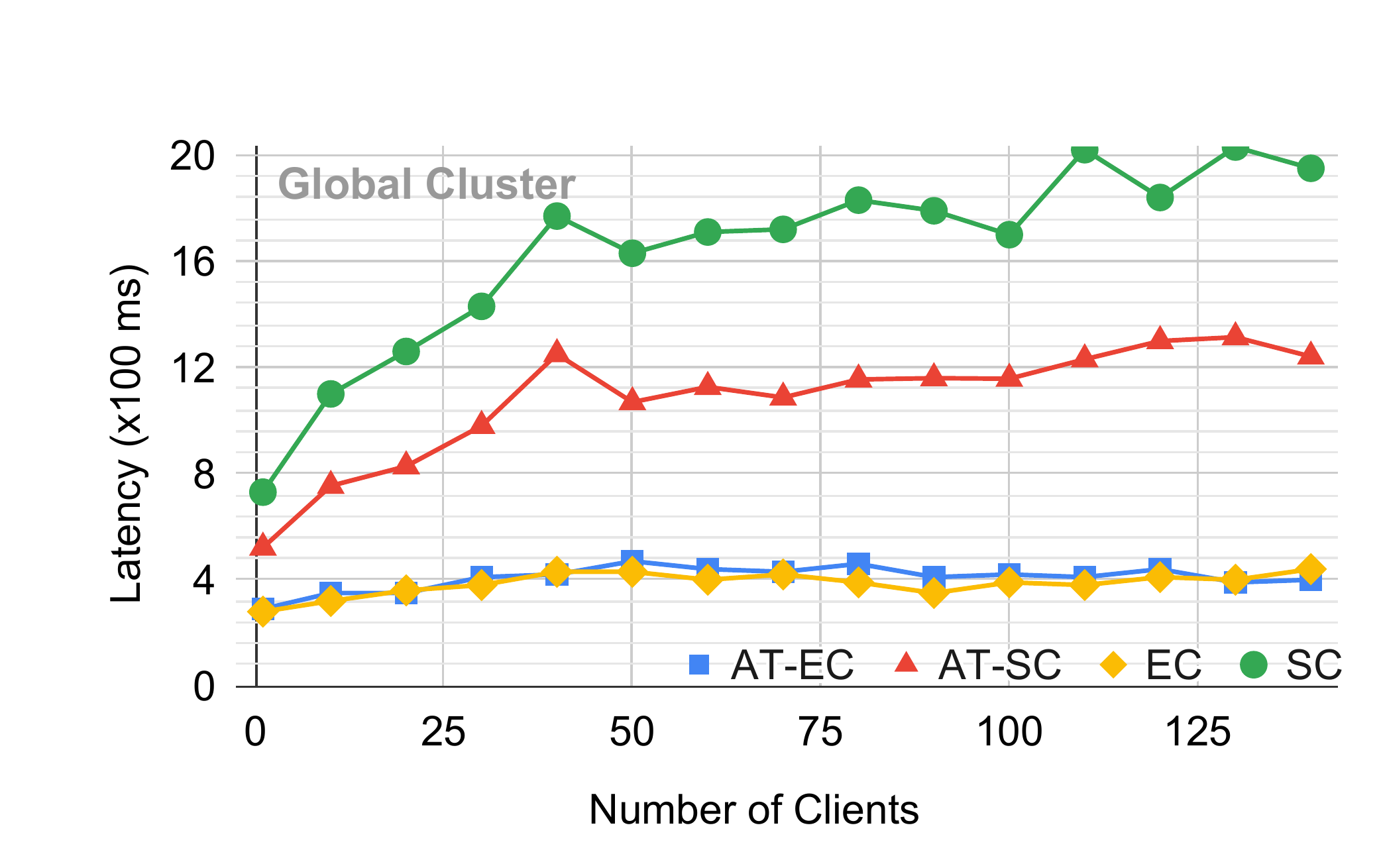}
\end{subfigure}
% VA
\begin{subfigure}[b]{0.32\textwidth}  
\centering
  \includegraphics[width=\textwidth]{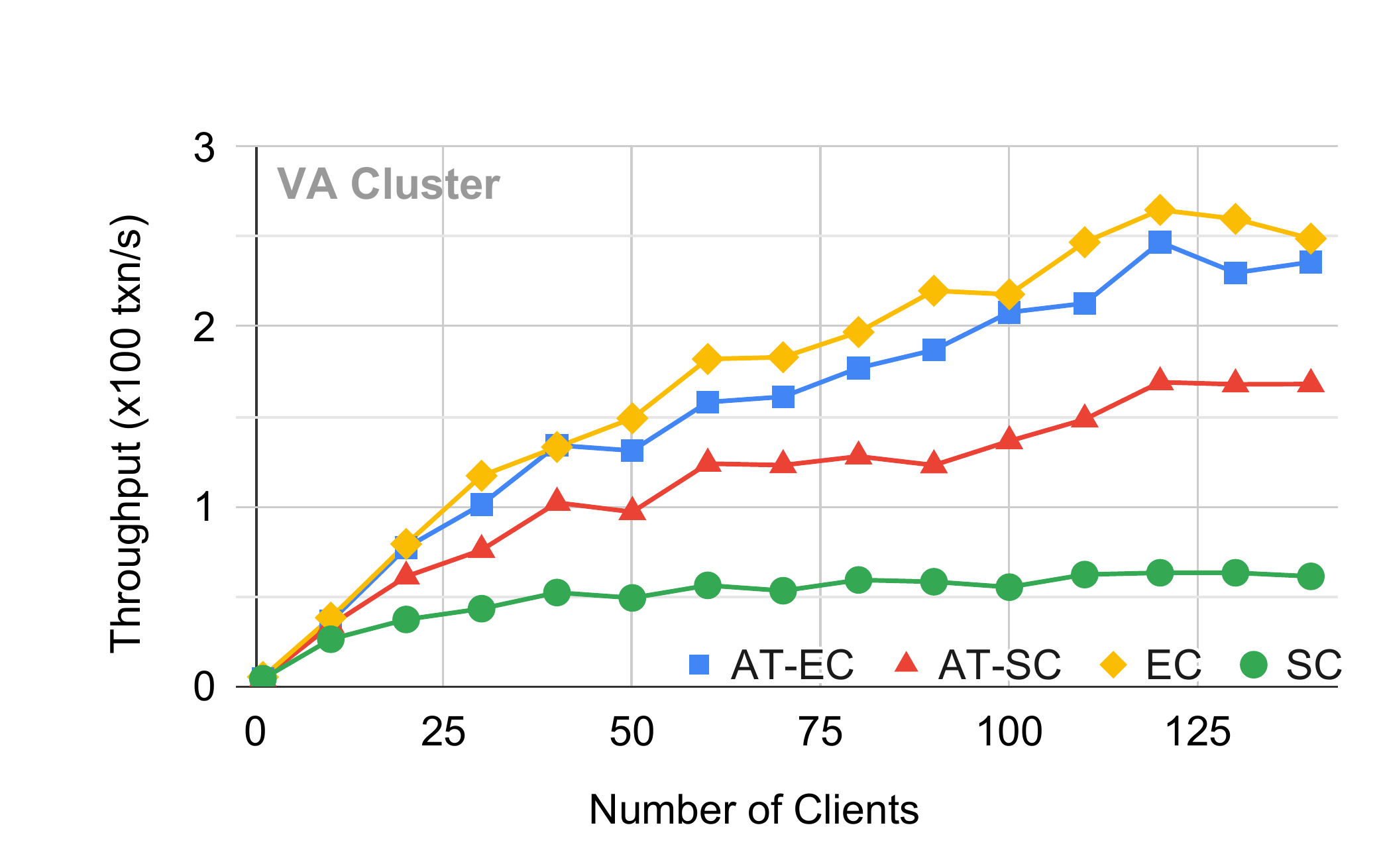}
\end{subfigure}
%%
% US
\begin{subfigure}[b]{0.32\textwidth}  
  \includegraphics[width=\textwidth]{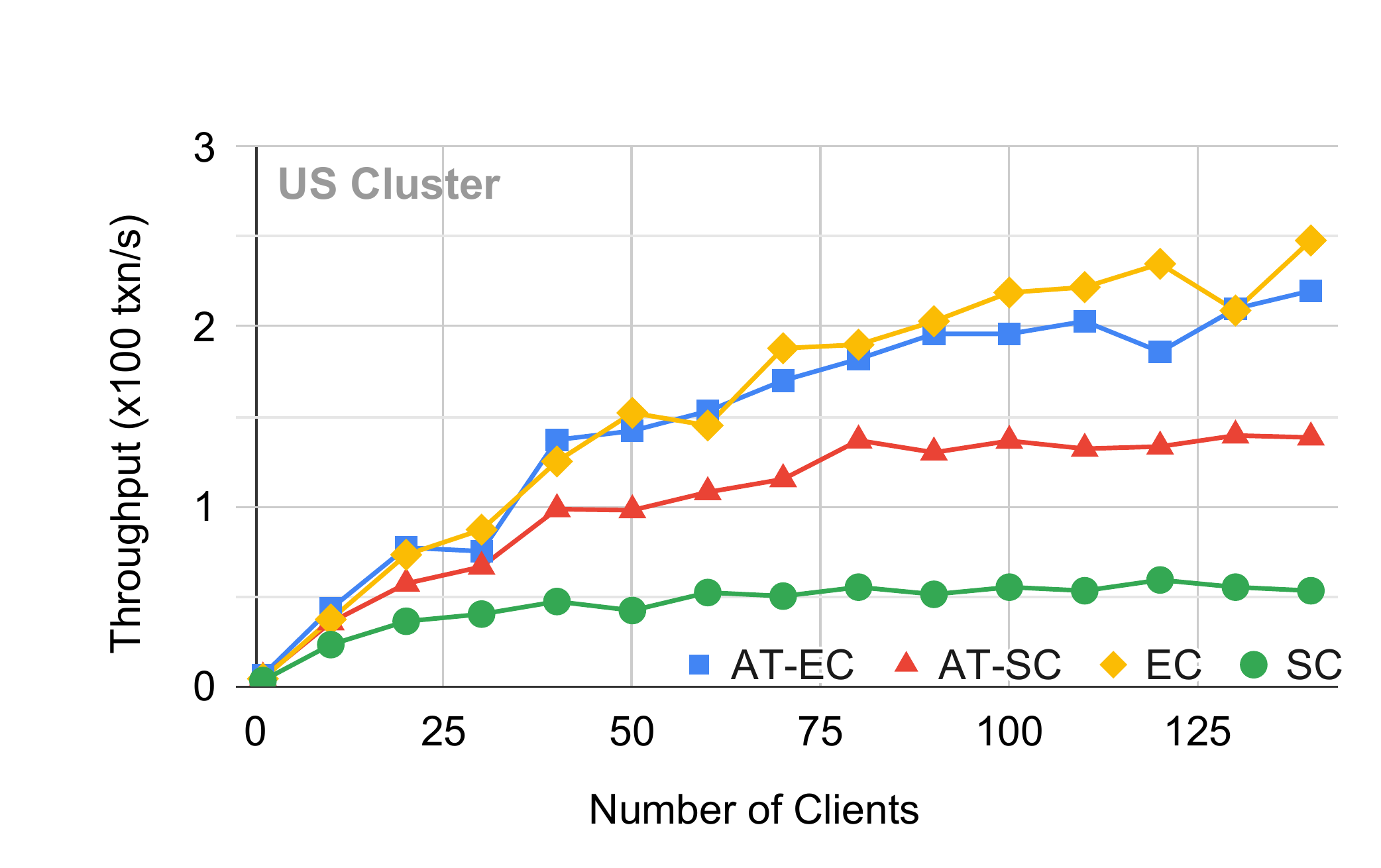}
\end{subfigure}
%%
% GLOBAL
\begin{subfigure}[b]{0.32\textwidth}  
\centering
  \includegraphics[width=\textwidth]{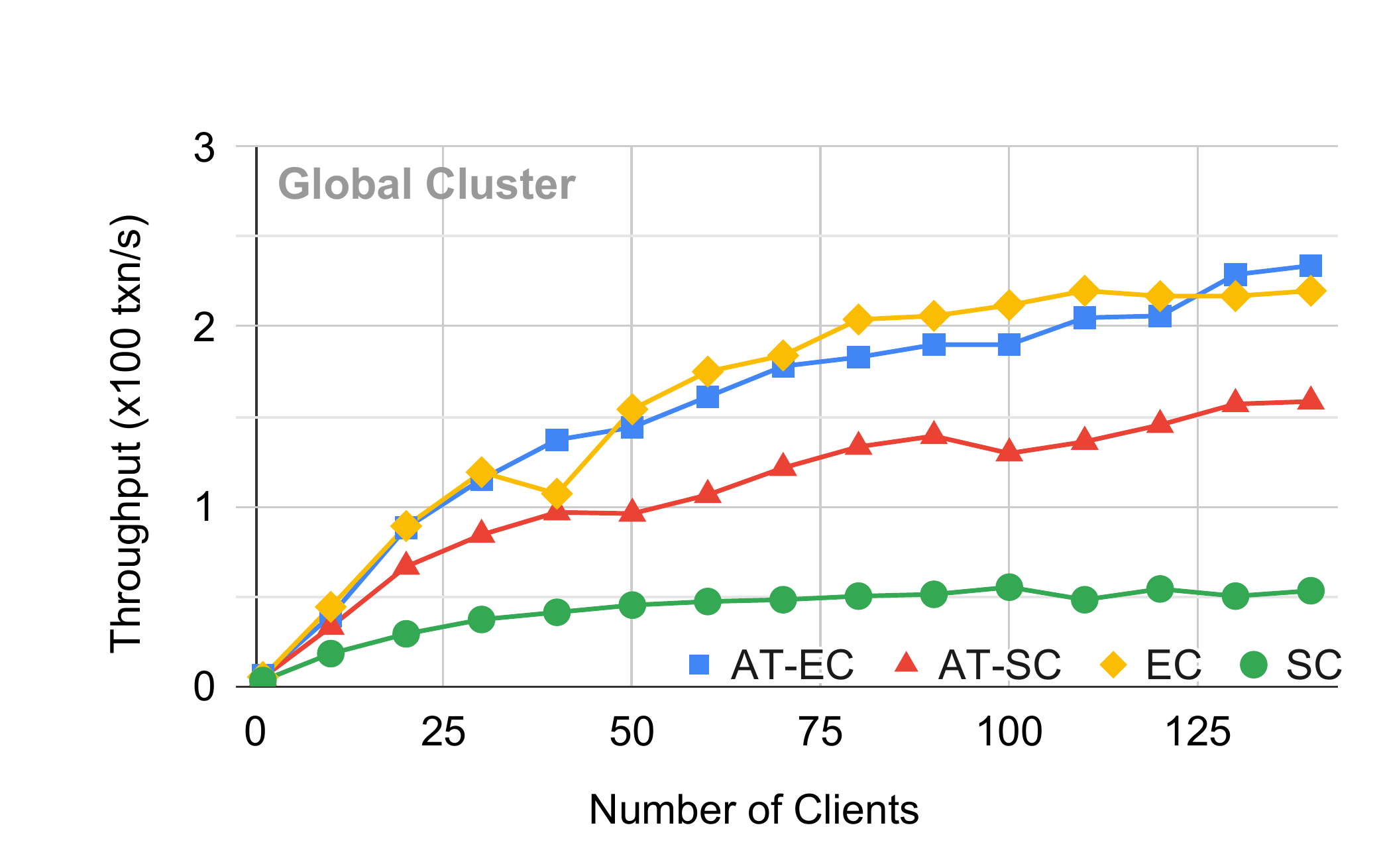}
\end{subfigure}
  \caption{Performance evaluation of TPC-C benchmark}
  \label{fig:perf_tpcc}
\end{figure*}

\subsection{Application-Level Invariants}
We ran further experiments on the
SmallBank benchmark.  This benchmark has six transactions which
maintain the details of customers and their accounts. Specifically,
each customer is assigned a checking entry and a savings entry in dedicated
tables.  By analysing this benchmark and similar
banking applications from the literature, we
defined the following application-level invariants which should be
preserved after the execution of any transaction:
  \begin{enumerate}
    \item The checking and saving balance of accounts must always be non-negative,
    \item Each account must reflect the correct total amounts based on the history of
      deposits performed on that account,
    \item Each client must always witness a consistent state of her checking and savings
      accounts. For example, when transferring money from one to the other, users
      must not witness an intermediate state where the money is deducted from
      the checking account, but is not deposited into savings.
  \end{enumerate}

\subsection{Comparison with Random Refactoring}
We investigated the utility of using the results of our oracle to
drive our repair procedure. For these experiments we removed the
initial phase of analysis and instead randomly introduced
tables and fields and prescribed value correspondences between them.
The results of these experiments for the three benchmarks with the
highest number of anomalies is presented in
\autoref{fig:eval_random2}. 
  Each
experiment was repeated for 5 hours, where at each round of experiments 10 random
refactorings were applied in the program. Each red dot in
\autoref{fig:eval_random2} records the number of anomalies at the end
of each round of (random) refactoring; the blue line is the number
of anomalies in the repaired program produced by \tool{}.  In our experiments, the vast majority of random
refactorings did not eliminate \emph{any} of the anomalies.  Even
those experiments that managed to repair some anomalies still resulted
in a program with many more serializability bugs than that returned by
\tool{}'s oracle-guided repair strategy.

\begin {figure*}[h]  
% TPC-C
\begin{subfigure}[h]{0.32\textwidth}  
\centering
  \includegraphics[width=\textwidth]{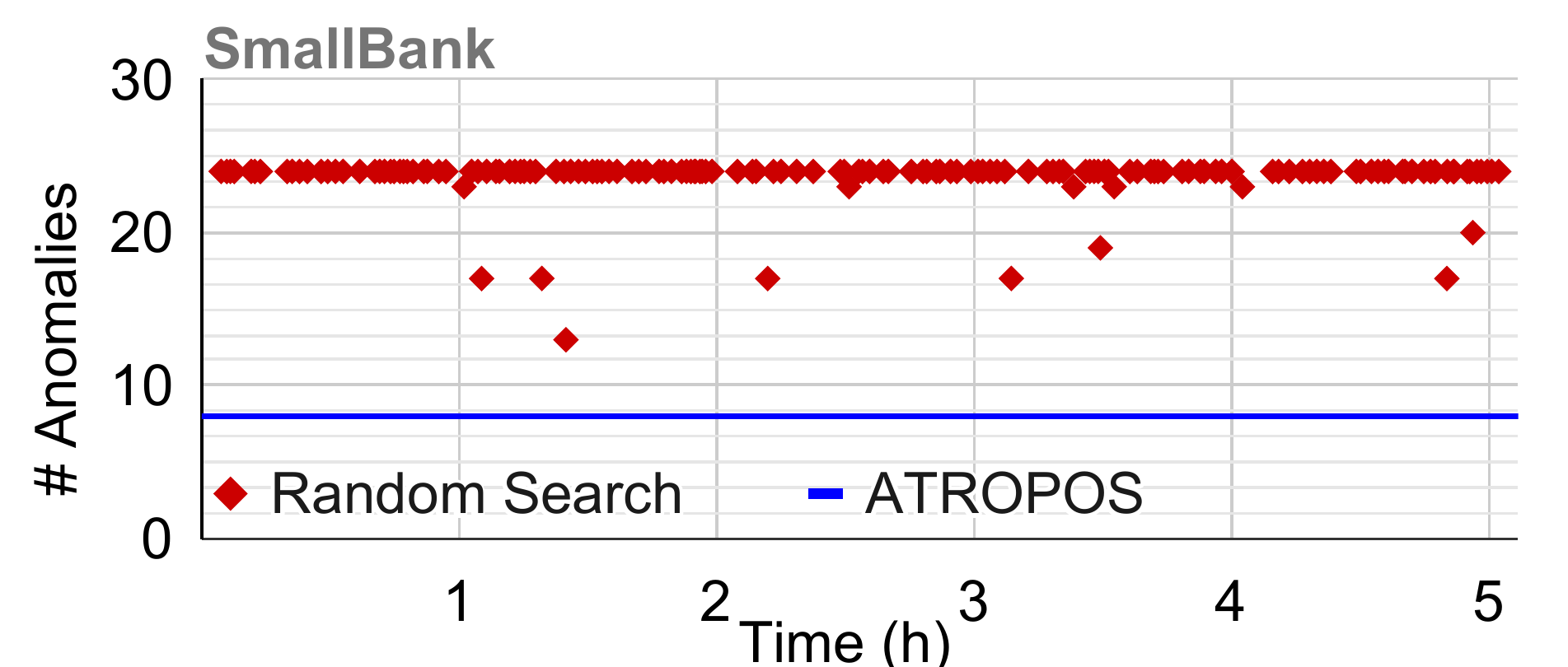}
 % \caption*{SmallBank}
\end{subfigure}
%%
% SEATS
\begin{subfigure}[h]{0.32\textwidth}  
\centering
  \includegraphics[width=\textwidth]{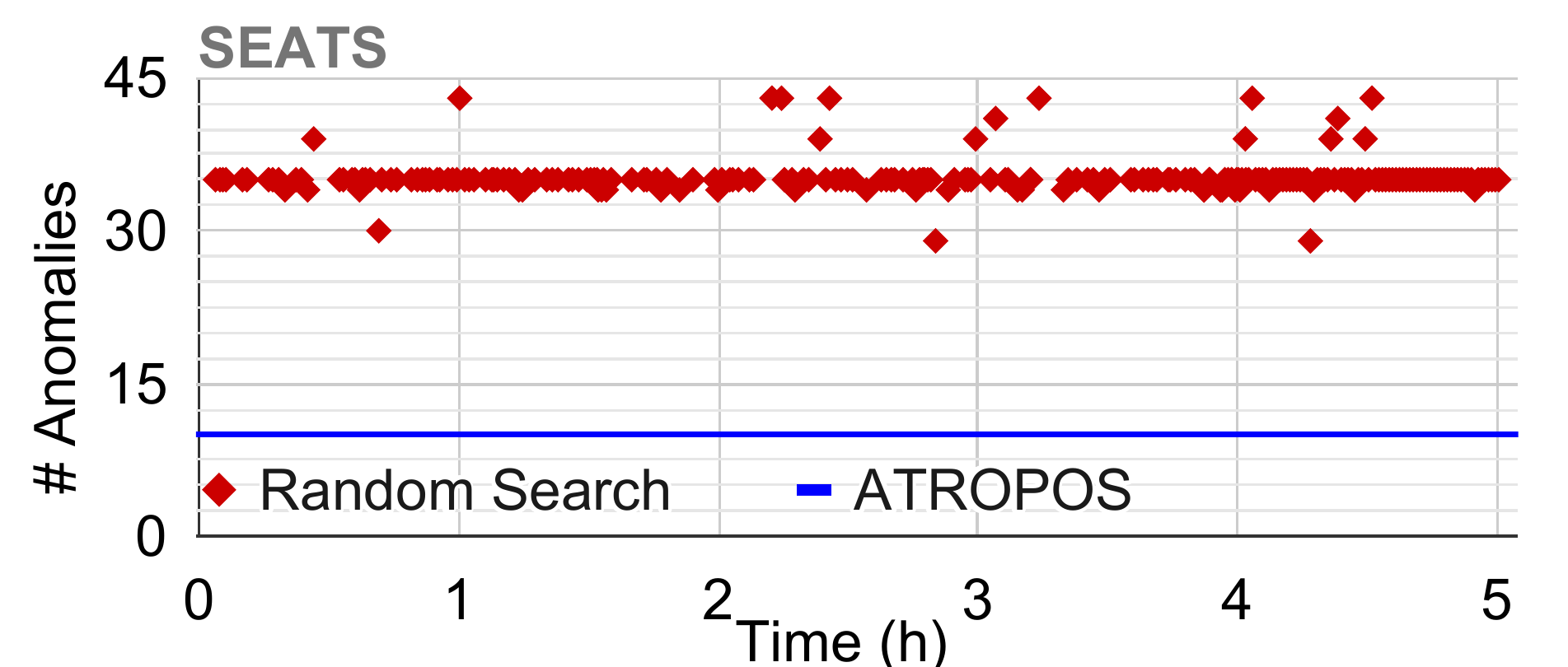}
 % \caption*{SEATS}
\end{subfigure}
%%
% TICKET
\begin{subfigure}[h]{0.32\textwidth}  
\centering
  \includegraphics[width=\textwidth]{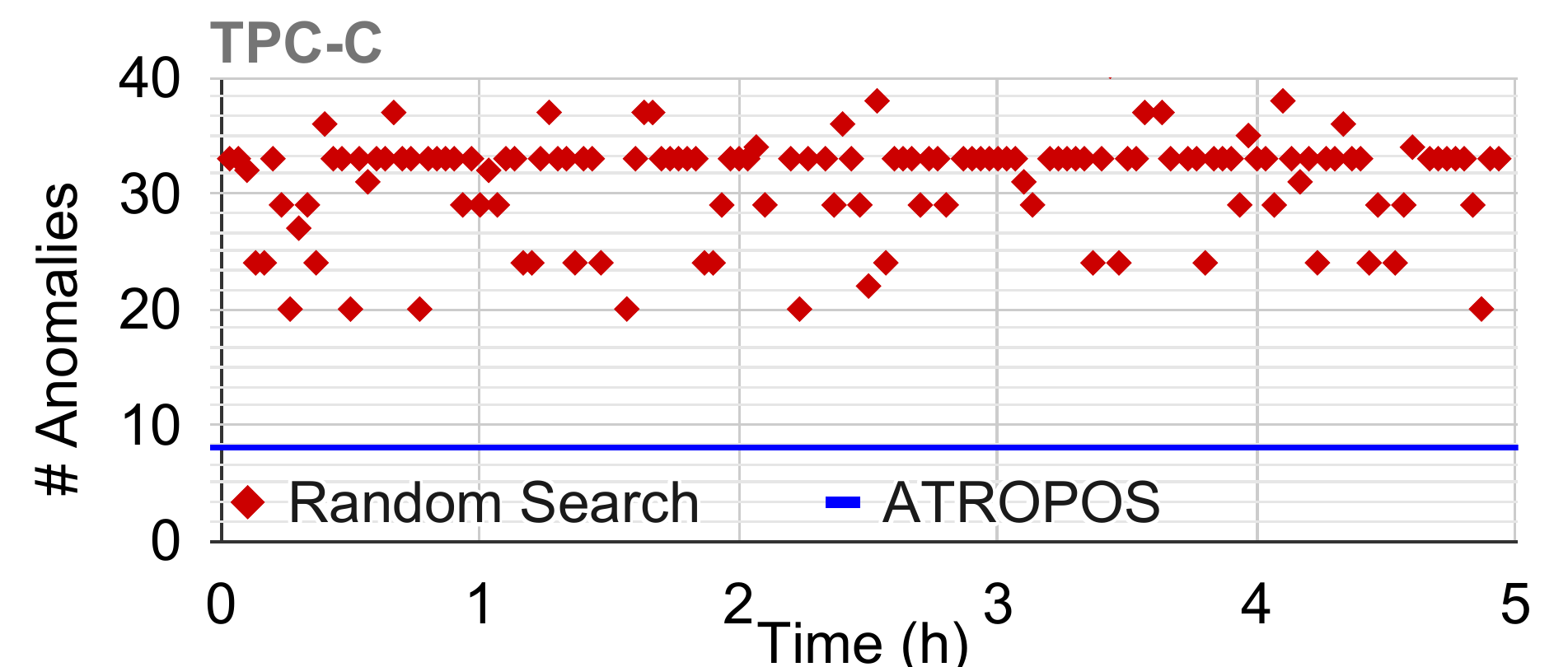}
%\caption*{TPC-C}
\end{subfigure}
  \caption{Number of anomalous access pairs in randomly
  refactored programs
}
  \label{fig:eval_random2}
\end{figure*}

\textbf{Definition of } $\Sigma(r.f)$: Given a database state $\Sigma = (\func{str},\func{vis},\func{cnt})$, a primary key $r \in R_{id}$ and a field $f$, we define $\Sigma(r.f)$ as follows:
$$
\Sigma(r.f) = v \Leftrightarrow \exists \eta \in \func{str}. \eta_r = r \wedge \eta_f = f \wedge \eta_v = v \wedge (\forall \eta' \in \func{str}. \eta'_r = r \wedge \eta'_f = f \Rightarrow \eta'_\tau \leq \eta_\tau)
$$

Given a program $P$ and a value correspondence $v=$ $(R, R',$ $f, f', \theta, \alpha)$, let $\mathcal{V}^{\func{id}}_{v,P}$ be the set containing all identity value correspondences for all fields of all relations of $P$ except the fields $f,f'$ and the value correspondence $v$. The notation also generalizes for a set of value correspondences $V$ (i.e. $\mathcal{V}^{\func{id}}_{V,P}$ is the union of the set of all identity value correspondences for all fields of all relations of $P$ except the fields which are source and target fields in $V$, and the set of the value correspondence $V$). In accordance with the definitions in the paper, here we define a containment relation between
sets of records as follows:
\begin{align*}
  X \sqsubseteq_V
  & X' ~~\equiv~~ \\
  &
    \forall (R,R',f,f',\theta,\alpha)\in V.\;
    \forall r\in R_\func{id}.\; r \in X_{id} \Leftrightarrow\\ & \theta(r) \cap X'_{\func{id}} \neq \emptyset \wedge \\
    & X(r.f)=\alpha(\{m\,|\,r'\in\theta(r) \cap X'_\func{id}~\wedge~ X'(r'.f')=m\})\\
\end{align*}

We begin the proofs by formalizing $\eval{.}$ and $R1-R3$ requirements in more detail. The transformation on select commands is defined as follows:
\begin{mathpar}
  \eval{c}_v = (x:=\,\sql{select} f'\; \sql{from} {R'}\ \sql{where}
  \func{redirect}(v, {\phi},)),
\end{mathpar}
 where  the function  $\func{redirect}: v \times \phi \rightarrow \phi$
 satisfies the following constraint:
 \ruleLabel{R1}
 \begin{equation*}
      %\tag{R1}
       \RULE{
         r\in R_\func{id}
 \qquad
   \Sigma\sqsubseteq_{\mathcal{V}^{\func{id}}_{v,P}} \Sigma'
   \qquad
   \Delta\sqsubseteq_{\mathcal{V}^{\func{id}}_{v,P}}\Delta'}
       {r \in \Sigma_{\func{id}} \wedge \Delta,\phi(\Sigma(r))\Downarrow \func{true} \Leftrightarrow
 \theta(r) \cap \{r' \in R'_\func{id}\,|\, r' \in \Sigma' \wedge
 \\
  \Delta',{\func{redirect}(v,\phi)}(\Sigma'(r')) \Downarrow \func{true}\} \neq \emptyset \wedge \\
 \Sigma(r.f) = \alpha(\{\Sigma'(r'.f')\ |\ r' \in \theta(r) \wedge \Delta',{\func{redirect}(v,\phi)}(\Sigma'(r')) \Downarrow \func{true}\})}
 \end{equation*}
 
 \ruleLabel{R2}
 \begin{equation*}
       \RULE{

       \Delta \sqsubseteq_{\mathcal{V}^{\func{id}}_{v,P}} \Delta'
        
       }{
       \Delta, e \Downarrow n \Leftrightarrow \Delta',\eval{e}_v \Downarrow n
     }
     \end{equation*}

The redirect function takes a value correspondence and a WHERE clause, and returns another WHERE clause. The rule (R1) ensures that if the database state $\Sigma$ and local state $\Delta$ are contained inside $\Sigma'$ and $\Delta'$ respectively, then the output of SELECT query in the original program must be contained in the output of the SELECT query in the refactored program. A sufficient condition for this is that for every record $r$ in $\Sigma$ which satisfies the original WHERE clause $\phi$ (expressed as $\Delta,\phi(\Sigma(r))\Downarrow \func{true}$), there must be records in $\theta(r)$ that satisfy the new WHERE clause (that is returned by $redirect$), and that the values in the selected fields in the original record $r$ can be determined using the records in $\theta(r)$.

The rule (R2) simply says that if the local database state $\Delta$ is contained inside $\Delta'$, then any expression $e$ on $\Delta$ or $\Delta'$ must evaluate to the same value.

 For any command
 $c=\sql{update} R\ \sql{set} f=e\ \sql{where}\phi$ in the original program, which is
 affected by the newly added  $v=(R,R',f,f',\theta,\alpha)$, the refactored command is given by:
\begin{equation}
\eval{c}_v :=  \sql{update} R'\ \sql{set} f' = 
\\
\func{duplicate}(e,v)\ \sql{where} \func{duplicate}(\phi,v)
\end{equation}
The refactored command must satisfy the following property:
 \ruleLabel{R3}
  \begin{equation*}
 \RULE
 {
   \Sigma\sqsubseteq_{\mathcal{V}^{\func{id}}_{v,P}} \Sigma' \qquad
   \Delta \sqsubseteq_{\mathcal{V}^{\func{id}}_{v,P}} \Delta'
   \qquad
   \Sigma',\Delta',\llbracket c \rrbracket_v \rightarrow \Sigma'',\Delta',\func{skip}
   \qquad
   r\in R_\func{id}
 }
 {
  r \in \Sigma_\func{id} \wedge \Delta,\phi(\Sigma(r)) \Downarrow \func{true} \Leftrightarrow \{r' \in R'_\func{id}\ |\ r' \in \Sigma'_\func{id} \wedge  \\
  \Delta', \func{duplicate}(\phi,v)(\Sigma(r')) \Downarrow \func{true}\} \cap \theta(r) \neq \emptyset
  \\    \Delta,e\Downarrow n \Leftrightarrow \alpha(\{ \Sigma''(r'.f')  \,|\, r'\in\theta(r) \wedge r' \in \Sigma''_\func{id} \}) = n
 }
 \end{equation*}

The rule (R3) states that if the database state $\Sigma$ and local state $\Delta$ are contained inside $\Sigma'$ and $\Delta'$ respectively, and if the execution of the refactored UPDATE leads to the new database state $\Sigma''$refactored program executes the UPDATE statement, then for every record $r$ whose field $f$ is modified by the original UPDATEprogram (modified to the valuation of the expression $e$), there must be records in $\theta(r)$ which must be modified by the refactored UPDATE, and that the modification in the original UPDATE (i.e. value of expression $e$) can be determined from the modified records obtained after the execution of the refactored UPDATE. This ensures that the original UPDATE is preserved by the refactored UPDATE. 

\begin{figure*}[h]
%    \begin{mdframed}[backgroundcolor=grey10,linecolor=grey10]
  \begin{minipage}{\textwidth}
  \footnotesize
\begin{itemize}
  
  \item[(comm)]
 $
    \begin{array}{llcl}
    (1) & \eval{c;c'}_v & \coloneqq & \eval{c}_v;\eval{c'}_v \\
    (2) & \eval{\func{skip}}_v & \coloneqq & \func{skip} \\
    (3) & \eval{\func{iterate}(e)\{c\}}_v & \coloneqq &
    (4)\func{iterate}(\eval{e}_v)\{\eval{c}_v\} \\
    (5) &  \eval{\func{if}(e)\{c\}}_v & \coloneqq &
    \func{if}(\eval{e}_v)\{\eval{c}_v\} \\
    \end{array}
    $
  
  \item [(whc)]
    $
     \begin{array}{llcl}
    (6) &   \eval{\phi\circ\phi'}_v & \coloneqq & \eval{\phi}_v\circ\eval{\phi'}_v \\
    (7) & \eval{\func{this}.f\odot e}_v & \coloneqq & \func{this}.f\odot \eval{e}_v \\
    \end{array}
    $

\item[(exp)] 
   $
    \begin{array}{llcl}
      (8) & \eval{\func{at}^e(x.f'')}_v & \coloneqq &
      \func{at}^{\eval{e}_v}(x.f'')\\
      (9) & \eval{\func{agg}(x.f'')}_v & \coloneqq &
      \func{agg}(x.f'')\\
      (10) & \eval{n}_v & \coloneqq & n \\
      (11) & \eval{a}_v & \coloneqq & a \\
      (12) & \eval{e\oplus e'}_v & \coloneqq & \eval{e}_v \oplus \eval{e'}_v\\
      (13) & \eval{e\odot e'}_v & \coloneqq &  \eval{e}_v \odot \eval{e'}_v \\
      (14) & \eval{e\circ e'}_v & \coloneqq &   \eval{e}_v \circ \eval{e'}_v\\
      (15) & \eval{\func{iter}}_v & \coloneqq & \func{iter}\\
  \end{array}
    $
  \item[(select)] 
    $
    \begin{array}{llcl}
      (16) & \eval{x:=\; \sql{select} f''\; \sql{from} {R}\ \sql{where} \stx{\phi}}_v
      & \coloneqq & (x:=\; \sql{select} {f''}\; \sql{from} {R}\ \sql{where}
  \eval{\phi}_v)\\
      (17) & \eval{x:=\; \sql{select} f''\; \sql{from} {R''}\ \sql{where} \stx{\phi}}_v
      & \coloneqq & (x:=\; \sql{select} {f''}\; \sql{from} {R''}\ \sql{where}
  \eval{\phi}_v)
%     \\
%      (18) & \eval{x:=\; \sql{select} f\; \sql{from} {R}\ \sql{where} \stx{\phi}}_v
%      &  \coloneqq  & (x:=\; \sql{select} {f'}\; \sql{from} {R'}\ \sql{where}
%     \param{\func{redirect}(\eval{\phi}_v,v)})
    \end{array}
    $

  \item[(update)]  $
    \begin{array}{llcl}
      (19) & \eval{\sql{update} R\; \sql{set} f''=e \;\sql{where}
  \phi}_v &  \coloneqq &  \sql{update} R\; \sql{set}
  f''=\eval{e}_v \;\sql{where}
  \eval{\phi}_v  \\
  (20) & \eval{\sql{update} R''\; \sql{set} f''=e \;\sql{where}
  \phi}_v &  \coloneqq &  \sql{update} R''\; \sql{set}
  f''=\eval{e}_v \;\sql{where}
  \eval{\phi}_v \\
%  (21) & \eval{\sql{update} R\; \sql{set} f=e \;\sql{where}
%  \phi}_v &\!\! \!\!\!\! \coloneqq & \!\!\! \!\!\!\sql{update} R'\; \sql{set} f'=\eval{e}_v \;\sql{where}
%  \param{\func{redirect}(\phi,v)}) \\
%  (22) & \eval{\sql{update} R\; \sql{set} f=e \;\sql{where}
%  \phi}_v & \!\!\! \!\!\! \coloneqq & \!\!\! \!\!\! \sql{update} R'\; \sql{set} \param{\func{duplicate}(f',\eval{e}_v)} \;\sql{where}
%  \param{\func{duplicate}(\eval{\phi}_v,v)})
    \end{array}
    $

\end{itemize}
\caption{Transformations triggered by a new value correspondence
$v=(R,R',f,f',\theta,\alpha)$}
  \label{fig:transformations}
\end{minipage}
%\end{mdframed}
\end{figure*}

\begin{lemma}
Given a program $P$, let $v=(R, R', f, f', \theta, \alpha)$ be a value correspondence and let $V = \mathcal{V}^{\func{id}}_{v,P}$. Then, if $(V_1,P) \hookrightarrow_{\textsc{v}} (V_1 \cup \{v\},P')$ and the refactoring semantics is equipped with the correct rewriting rules - which satisfy $R1$, $R2$ and $R3$, then $P' \preceq_V P$.
\end{lemma}
\begin{proof}
First we will show that for all histories $h' \in H(P')$, there exists a history $h \in H(P)$ such that $h' \preceq_{\mathcal{V}^{\func{id}}_{v,P}} h$. Given $h'$, the history $h$ is constructed by instantiating the same transactions as those in $h$ in the same order with the same arguments. Further, each command $\llbracket c \rrbracket_v$ of $P'$ executed in $h'$ corresponds to command $c$ of P. Hence, to construct the history $h$, we will execute the corresponding command for every execution step of $h'$. Formally, for a step $\Sigma'_1, \Delta'_1, \llbracket c \rrbracket_v \rightarrow \Sigma'_2, \Delta'_2, \llbracket c' \rrbracket_v$ in the history $h'$, we will perform the step $\Sigma_1, \Delta_1, c \rightarrow \Sigma_2, \Delta_2, c'$. We will show that the following invariant is maintained at every step: $\Sigma_2 \sqsubseteq_V \Sigma_2^{'}$, and $\Delta_2 \sqsubseteq_V \Delta_2^{'}$. We will use induction on the number of steps in the history $h'$ (which is same as the number of steps in history $h$). 

\textbf{Base Case}: In the beginning, both $\Sigma'$ and $\Delta'$ are empty. The first step in both $h$ and $h'$ is necessarily the application of the rule \textsc{txn-invoke}, which invokes an instance of some transaction $t$ in both $h$ and $h'$ with the same argument. Hence, after this step, both $\Sigma',\Delta'$ and $\Sigma,\Delta$ continue to remain empty and so the invariant holds trivially.

\textbf{Inductive Case}: Assume that the step $\Sigma_1^{'}, \Delta_1^{'}, \llbracket V,c_1 \rrbracket_v \rightarrow \Sigma_2^{'}, \Delta_2^{'}, \llbracket V,c_2 \rrbracket_v$ is performed in history $h'$ and the step $\Sigma_1, \Delta_1, c \rightarrow \Sigma_2, \Delta_2, c'$ is performed in the history $h$. By the inductive hypothesis, $\Sigma_1 \sqsubseteq_V \Sigma_1^{'}$, and $\Delta_1 \sqsubseteq_V \Delta_1^{'}$. We case split based on the type of step taken:

\textsc{txn-ret}: For a return expression $e$, by $(R2)$, since $\Delta_1 \sqsubseteq_V \Delta_1^{'}$, if $\Delta_1^{'}, \eval{e}_v \Downarrow n$, then $\Delta_1,e  \Downarrow n$. In this case, the corresponding transaction instances in $h$ and $h'$ return the same value $n$. Further, $\Sigma_2^{'} = \Sigma_1^{'}$, $\Delta_2^{'} = \Delta_1^{'}$ and $\Sigma_2 = \Sigma_1$, $\Delta_2 = \Delta_1$ and hence the invariant continues to hold. By similar reasoning, steps \textsc{cond-t}, \textsc{cond-f}, \textsc{iter} will also preserve the invariant.

\textsc{select}: We further case-split based on whether the $SELECT$ accesses the source field $f$ of the value correspondence $v$ or not. First, let us consider the case where the $SELECT$ accesses a different field. WLOG, let the command executed by $h'$ be $\eval{x:=\; \sql{select} f''\; \sql{from} {R''}\ \sql{where} \stx{\phi}}_v$. 

This is defined to be the command $$x:=\; \sql{select} {f''}\; \sql{from} {R''}\ \sql{where}\eval{\phi}_v$$. In $h'$, using the rule \textsc{consistent-view}, a view of the data store $\Sigma^{'*} \sqsubseteq \Sigma_1^{'}$ will be generated. Since the command accesses fields which are different from $f$, there are identity value correspondences for all the accessed fields in $V$, and since $\Sigma_1 \sqsubseteq_V \Sigma_1^{'}$, all the write events to fields accessed by the command in store $\Sigma_1^{'}$ will also be present in $\Sigma_1$. Hence, there exists a view $\Sigma^* \sqsubseteq \Sigma_1$ which contains the same write events to fields accessed by the command as $\Sigma^{'*}$. Finally, since $\Delta_1 \sqsubseteq_V \Delta_1^{'}$, by $(R2)$ all expressions $e$ present in the where clause $\phi$ will evaluate to same value in $h$ and $h'$. Hence, the same records will satisfy the where clauses $\eval{\phi}_v$ and $\phi$, leading to the same selection of records and hence the same binding to $x$, thus preserving $\Delta_2(x) \sqsubseteq_V \Delta_2^{'}(x)$.

Now, let us consider the case where the command executed by $h'$ is $$\eval{V,x:=\; \sql{select} f\; \sql{from} {R}\ \sql{where} \stx{\phi}}_v$$. This is defined to be the command $$x:=\; \sql{select} {f'}\; \sql{from} {R'}\ \sql{where} \param{\func{redirect}(V,\eval{\phi}_v,v)}$$.  In $h'$, using the rule \textsc{consistent-view}, a view of the data store $\Sigma^{'*} \sqsubseteq \Sigma_1^{'}$ will be generated to determine the records which satisfy the $WHERE$ clause $\func{redirect}(V,\eval{\phi}_v,v)$. Since $\Sigma_1 \sqsubseteq_V \Sigma_1^{'}$, for any record which satisfies $\func{redirect}(V,\eval{\phi}_v,v)$ and which is present in $\Sigma_1^{'}$, it corresponding record according to the value correspondence will also be present in $\Sigma_1$. Hence, we can construct a view of the data store $\Sigma^* \sqsubseteq \Sigma_1$ such that all fields of all records accessed in the $WHERE$ clause evaluate to the same value, i.e. $\Sigma^* \sqsubseteq \Sigma^{'*}$. Now, by $(R1)$, it is guaranteed that if records $\theta(r)$ are selected by $h'$, then the record $r$ must be selected by $h$. Further, since $\Sigma_1 \sqsubseteq_V \Sigma_1^{'}$ and $V$ includes the value correspondence $v$, the value of $f$ in the selected record in $h$ will satisfy the value correspondence $v$ w.r.t. the records $\theta(r)$ selected in $h'$. Hence, $\Delta_2(x) \sqsubseteq_V \Delta_2^{'}(x)$.

\textsc{update}: We case split based on whether the $UPDATE$ command modifies the source field $f$. First, let us consider the case where the $UPDATE$ modifies a different field. WLOG, let the command executed by $h$ be $\eval{\sql{update} R''\; \sql{set} f''=e \;\sql{where}\phi}_v$. This is defined to be the command $\sql{update} R''\; \sql{set} f''=\eval{e}_v \;\sql{where}\eval{\phi}_v$.

In $h'$, using the rule \textsc{consistent-view}, a view of the data store $\Sigma^{'*} \sqsubseteq \Sigma_1^{'}$ will be generated. Since the command accesses fields in the $WHERE$ clause which are different from $f$, there are identity value correspondences for all the accessed fields in $V$, and since $\Sigma_1 \sqsubseteq_V \Sigma_1^{'}$, all the write events to the accessed fields in store $\Sigma_1^{'}$ will also be present in $\Sigma_1$. Hence, there exists a view $\Sigma^* \sqsubseteq \Sigma_1$ which contains the same write events to the accessed fields as present in $\Sigma^{'*}$.  This guarantees that the same set of records in $R''$ will be selected by the $WHERE$ clauses in $h$ and $h'$ .Finally, since $\Delta_1 \sqsubseteq_V \Delta_1^{'}$, by $(R2)$, the expressions $\eval{V,e}_v$ and $e$ will evaluate to the same value. Hence, write events to the same field in the same set of records and writing the same value will be generated in $h$ and $h'$. This guarantees that $\Sigma_2 \sqsubseteq_V \Sigma_2^{'}$.

Now, let us consider the case where the command executed by $h'$ is $$\llbracket \sql{update} R\ \sql{set} f=e\ \sql{where}\phi \rrbracket_v$$. This is defined to be the command $ \sql{update} R'\ \sql{set} f' = \func{duplicate}(e,v)\ \sql{where} \func{duplicate}(\phi,v)$. n $h'$, using the rule \textsc{consistent-view}, a view of the data store $\Sigma^{'*} \sqsubseteq \Sigma_1^{'}$ will be generated to determine the records which satisfy the $WHERE$ clause $\func{duplicate}(\phi,v)$. Since $\Sigma_1 \sqsubseteq_V \Sigma_1^{'}$, if $\theta(r)$ is present in $\Sigma_1^{'}$, then $r$ would also be present in $\Sigma_1$. Hence, we can construct a corresponding view $\Sigma* \sqsubseteq \Sigma_1$ such that $\Sigma* \sqsubseteq_V \Sigma^{'*}$. Now, $(R3)$ guarantees that if any record in $\theta(r)$ is updated by $h'$, then the corresponding record $r$ would also be updated by $h$ (since it would satisfy the original $WHERE$ clause $\phi$). Further, $(R3)$ also guarantees the update $\func{duplicate}(e,v)$ in $h'$ would satisfy the value correspondence with the update in $h$. Hence, after the update $\Sigma_2 \sqsubseteq_V \Sigma_2^{'}$.

Note that the proof that for every serial execution of $P$, there exists an execution of $P'$ would follow the exact same pattern as above. In fact, we would actually construct a serial execution of $P'$ which would have the same behavior as the execution of $P$.
\end{proof}

\begin{lemma}
Given programs $P_1$, $P_2$ and $P_3$, a set of value-correspondences $V$ and a value correspondence $v$ such that the source and target fields of $v$ are not involved in any value correspondence in $V$, if $P_2 \preceq_{\mathcal{V}^{\func{id}}_{V,P_1} } P_1$ and $P_3 \preceq_{\mathcal{V}^{\func{id}}_{v,P_2}} P_2$, then $P_3 \preceq_{\mathcal{V}^{\func{id}}_{V \cup \{v\},P_1}} P_1$.
\end{lemma}
\begin{proof}
We are given $H_1: P_2 \preceq_{\mathcal{V}^{\func{id}}_{V,P_1} } P_1$ and $H_2:P_3 \preceq_{\mathcal{V}^{\func{id}}_{v,P_2}} P_2$. First we will show that for all histories $h \in H(P_3)$, there exists a history $h'' \in H(P_1)$ such that $h \preceq_{\mathcal{V}^{\func{id}}_{V \cup \{v\},P_1}} h''$. Let $h_{fin} = (\Sigma, \Gamma)$. By hypothesis $H_2$, there exists a history $h' \in H(P_2)$, $h' = (\Sigma', \Gamma')$ such that $\Gamma = \Gamma'$. By hypothesis $H_1$, there exists a history $h'' \in H(P_1)$, $h'' = (\Sigma'', \Gamma'')$ such that $\Gamma'' = \Gamma'$. Hence, $\Gamma'' = \Gamma$.

Now, by $H_2$, $H_3:\Sigma' \sqsubseteq_{\mathcal{V}^{\func{id}}_{v,P_2}} \Sigma$, and by $H_1$, we have $H_4:\Sigma'' \sqsubseteq_{\mathcal{V}^{\func{id}}_{V,P_1}} \Sigma'$. We will show that $\Sigma'' \sqsubseteq_{\mathcal{V}^{\func{id}}_{V\cup\{v\},P_1}}$.

Consider a value correspondence $u \in V$, such that $u = (R,R',f,f',\theta, \alpha)$. Let $r \in R_\func{id}$. Suppose $r \in \Sigma''_{\func{id}}$. Then, by $H_4$, $\theta(r) \subseteq \Sigma'_\func{id}$. Now, since $f'$ is a field of $P_2$ which is not involved in value correspondence $v$, there exists an identity value correspondence $u' \in \mathcal{V}^{\func{id}}_{v,P_2}$, such that $u' = (R', R', f', f', \func{id}, \func{id})$. Hence, by $H_3$, we have $\theta(r) \subseteq \Sigma_\func{id}$. In the other direction, suppose $\theta(r) \subseteq \Sigma_\func{id}$. Then, by $H_3$, we have $\theta(r) \subseteq \Sigma'_\func{id}$, and hence by $H_2$, we have $r \in \Sigma''_{\func{id}}$.

Now, by $H_4$, $\Sigma''(r.f) = \alpha(\{m\ |\ r' \in \theta(r) \wedge \Sigma'(r'.f') = m\})$. By $H_3$ and the existence of the value correspondence $u'$, we have $\Sigma'(r'.f') = \Sigma(r'.f')$. Thus, $\Sigma''(r.f) = \alpha(\{m\ |\ r' \in \theta(r) \wedge \Sigma(r'.f') = m\})$.

Consider the value correspondence $v$ itself. Let $$v = (R,R',f,f',\theta, \alpha)$$. Since $f$ is not involved in any value correspondence in $V$, there exists an identity value correspondence $v'=(R,R,f,f,\func{id}, \func{id})$ such that $v' \in \mathcal{V}^{\func{id}}_{V,P_1}$. Let $r \in R_\func{id}$. Suppose $r \in \Sigma''_\func{id}$. Then by $H_4$, $r \in Sigma'_\func{id}$. By $H_3$, $\theta(r) \in \Sigma_\func{id}$. The other direction is straightforward. 

Now, by $H_3$, $\Sigma'(r.f) = \alpha(\{m\ |\ r' \in \theta(r) \wedge \Sigma(r'.f') = m\})$. By $H_4$, $\Sigma''(r.f) = \Sigma'(r.f)$. 

Hence, $\Sigma''(r.f) =  \alpha(\{m\ |\ r' \in \theta(r) \wedge \Sigma(r'.f') = m\})$.

Finally, consider an identity value correspondence in $\mathcal{V}^{\func{id}}_{V \cup \{v\},P_1}$ whose source field is neither in $V$ or in $v$. Then, this value correspondence is present in both $\mathcal{V}^{\func{id}}_{V,P_1}$ and $\mathcal{V}^{\func{id}}_{v,P_2}$. Hence, by $H_3$ and $H_4$, this value correspondence will also be preserved between $\Sigma''$ and $\Sigma$. Hence, $\Sigma'' \sqsubseteq{\mathcal{V\cup\{v\}}^{\func{id}}_{V,P_1}} \Sigma$. Thus, $h \preceq_{\mathcal{V}^{\func{id}}_{V \cup \{v\},P_1}} h''$.

Now we will show that for all histories $h'' \in H(P_1)_{\func{ser}}$, there exists a history $h \in H(P_3)_{\func{ser}}$ such that $h \preceq_{\mathcal{V}^{\func{id}}_{V \cup \{v\},P_1}} h''$. Let $h''_{fin} = (\Sigma'', \Gamma'')$, such that $h'' \in H(P_1)_{\func{ser}}$. By hypothesis $H_1$, there exists a history $h' \in H(P_2)_\func{ser}$, $h' = (\Sigma', \Gamma')$ such that $\Gamma'' = \Gamma'$ and $\Sigma'' \sqsubseteq_{\mathcal{V}^{\func{id}}_{V,P_1}} \Sigma'$. By hypothesis $H_2$, there exists a history $h \in H(P_3)_\func{ser}$, $h = (\Sigma, \Gamma)$ such that $\Gamma' = \Gamma$ and $\Sigma' \sqsubseteq_{\mathcal{V}^{\func{id}}_{v,P_2}} \Sigma$. Hence, $\Gamma'' = \Gamma$, and $\Sigma'' \sqsubseteq{\mathcal{V\cup\{v\}}^{\func{id}}_{V,P_1}} \Sigma$. Thus, $h \preceq_{\mathcal{V}^{\func{id}}_{V \cup \{v\},P_1}} h''$.
\end{proof}

\begin{theorem}
The refactoring semantics equipped with correct rewriting rules - which
satisfy $R1$, $R2$ and $R3$ - is sound, i.e.
    $ \forall_{P,P'}.\; \hstepc{\emptyset, P}{}{V, P'} \Rightarrow
    P'\preceq_{\mathcal{V}^{\func{id}}_{V,P}} P
$.
    \begin{proof}
Note that in order to express the result more formally, we have introduced new notation in the Appendix which has resulted in a slight change in the theorem statement from the main text of the paper. However, both statements signify the same result, since $V'$ used in the main text is actually the set of value correspondences introduced by the refactoring rules (same as $V$ in the above statement), and $\mathcal{V}^{\func{id}}_{V,P}$ contains $V$.

We first assume that all additions to the schema (either new relations or new fields) are carried out first using the rules \textsc{intro $\rho$} and \textsc{intro $\rho.f$}. These rules do not change the program, and since the original program $P$ is guaranteed to not access the new additions, it is clear that the statement of the theorem continues to hold after the above sequence of refactorings. Let $\func{Fld}_o$ be the set of all fields in the original program, and $\func{Fld}_n$ be the set of newly added fields. The rest of refactoring steps will involve multiple applications of the \textsc{intro $v$} rule. That is,
$$
(\emptyset,P) \hookrightarrow_{v_1} (\{v_1\}, P_1) \ldots \hookrightarrow_{v_n} (V,P_n)
$$
where $v_1, \ldots, v_n$ are the value correspondences that are introduced. We assume that the source field $f_i$ of every $v_i$ is a field of the original program i.e. $f_i \in \func{Fld}_o$ and the target field is a newly added field. We also assume that the no two value correspondences have the same source field or the same target field. Let $V = \{v_1, \ldots, v_n\}$. We will show that $P_n \preceq_{\mathcal{V}^{\func{id}}_{V,P}} P$. We will use induction on the number of refactoring steps ($n$). 

\textbf{Base Case}: For the first step, by applying Lemma B.1, $P_1 \preceq_{\mathcal{V}^{\func{id}}_{v_1,P_1}} P$.

\textbf{Inductive Case}: Assume that after $k$ steps, $P_k \preceq_{\mathcal{V}^{\func{id}}_{V_k,P_1}} P$, where $V_k = \{v_1, \ldots, v_k\}$. For the $(k+1)$th step $(V_k,P_k) \hookrightarrow_{v_{k+1}} (V', P_{k+1})$. Then, by Lemma B.1, $P_{k+1} \preceq_{\mathcal{V}^{\func{id}}_{v_{k+1},P_k}} P_k$. Now, by Lemma B.2, $P_{k+1} \preceq_{\mathcal{V}^{\func{id}}_{V_k \cup \{v_{k+1}\},P_1}} P_1$.  
\end{proof}
  \end{theorem}
  
\begin{theorem}
  The rewrite rules described in this section (i.e. by (I1.1), (I1.2),
  (I2.1), (I2.2), (I3.1) and (I3.2))
  satisfy the correctness properties (R1), (R2) and (R3).
    \begin{proof}
We will first show that (I1.1) satisfies (R1). We are given a value correspondence $v=(R,R',f,f',\theta,\alpha)$. Let $V = \mathcal{V}^{\func{id}}_{v,P}$. The value correspondence $v$ satisfies the property that $\alpha = \func{any}$ and $\theta$ selects records in $R'$ based only equality predicates on all the primary-key fields of $R$. Assume that each primary key field $p \in R_{\func{pk}}$ is equated with expression $e_p$ in the where clause $\phi$. Given $\Sigma, \Sigma', \Delta, \Delta'$ such that $\Sigma \sqsubseteq \Sigma'$ and $\Delta \sqsubseteq \Delta'$, consider $r \in R_\func{id}$. Assume that $\Delta, e_p \Downarrow n_p$. By (R2), we have $\Delta, e_p \Downarrow n_p$. Suppose $r \in \Sigma_\func{id}$.

\begin{align*}
\Delta,\phi(\Sigma(r))\Downarrow \func{true} & \Rightarrow \bigwedge_{p \in R_\func{pk}} \Sigma(r.p) = n_p
\end{align*}

This in turn implies that $\theta(r) = \{r'\ |\ r' \in R'_\func{id} \wedge \forall p \in R_\func{pk}. \Sigma'(r'.\hat{\theta}(p)) = n_p\}$. Since $\Sigma \sqsubseteq_V \Sigma'$, at least one such record in $\theta(r)$ is also present in $\Sigma'$. Now, by definition of $\func{redirect}(v,\phi)$,

\begin{align*}
r' \in \Sigma' \wedge \Delta',{\func{redirect}(v,\phi)}(\Sigma'(r')) \Downarrow \func{true}\} \Rightarrow \bigwedge_{p \in R_\func{pk}} \Sigma'(r'.\hat{\theta}(p)) = n_p
\end{align*}

We have just shown that at least one such record is guaranteed to be present in $\Sigma'$, and hence the intersection of $\theta(r)$ and above set is non-empty. To prove the reverse direction, we follow the same steps in reverse order. Since $\func{redirect}$ picks exactly the records which satisfy the above condition, the containment relation $\Sigma \sqsubseteq_V \Sigma'$ guarantees that $\Sigma'(r'.f') = \Sigma(r.f)$. This completes the proof that (I1.1) satisfies (R1). 

Now we will show that (I2.1) satisfies (R2). Consider the simple expression $e = \func{at}^1(x.f)$. Since $\Delta(x) \sqsubseteq \Delta'(x)$, there exists at least one record bound to $x$ in $\Delta'$. Further,  the value at field $f'$ in this record must be equal to the value at field $f$ of any record bound to $x$ in $\Delta$. Hence, $\func{at}^1(x.f) = \func{at}^1(x.f')$. Thus, (R2) follows for $e$. For compound expressions, the proof directly follows using induction on the structure of the expression.

Now we will show that (I3.1) satisfies (R3). The first condition in the goal of (R3) is the same as the first condition in goal of (R1). Due to (R2), the second condition also directly holds.

Now, we will prove the results about the logger refactoring. Note that this for $SELECT$ statement and expression transformations, this refactoring uses the same $\func{redirect}$ definition as above. Hence, the above proofs for (R1) and (R2) directly apply in this case for (I1.2) and (I2.2). Let us now prove that (I3.2) satisfies (R3). The first condition in the antecedent of (R3) holds directly by the same reasoning as applied in the cases above. For the second condition, note that due to the $\func{sum}$ value correspondence, the following holds:

\begin{align*}
\Sigma(r.f) = \sum_{r' \in \theta(r) \cap \Sigma'_\func{id}} \Sigma(r'.f')
\end{align*}

Assuming that $\Delta, e \Downarrow n$ (also $\Delta', e \Downarrow n$ due to (R2)), the updated value in record $r$ satisfying $\phi$ would be $\Sigma(r.f) + n$. On the R.H.S, a new record in $\theta(r)$ would be added in $\Sigma''$, with the value in the field $f'$ being $n$. Thus, the same value is added on both sides of the above equation, resulting in equal values again. This completes the proof that (I3.2) satisfies (R3). 
  \end{proof}
\end{theorem}

\end{document}